\documentclass[11pt]{article}
\usepackage[margin=1in]{geometry}
\usepackage[CJKbookmarks=true,
            bookmarksnumbered=true,
            bookmarksopen=true,
            colorlinks=true,
            citecolor=red,
            linkcolor=blue,
            anchorcolor=red,
            urlcolor=blue
            ]{hyperref}

\usepackage{amsmath}
\usepackage{amssymb}

\usepackage{amsthm}
\usepackage{bbm}
\usepackage{algpseudocode}
\usepackage{algorithm}
\usepackage{enumitem}
\usepackage{color}
\usepackage{multirow,booktabs}
\usepackage{graphicx}

\usepackage{url}

\graphicspath{{../Fig/}}
\newtheorem{theorem}{Theorem}[section]
\newtheorem{lemma}[theorem]{Lemma}

\newtheorem{proposition}[theorem]{Proposition}
\newtheorem{corollary}[theorem]{Corollary}
\theoremstyle{definition}
\newtheorem{remark}[theorem]{Remark}

\newtheorem{definition}[theorem]{Definition}
\newtheorem{assumption}[theorem]{Assumption}
\numberwithin{equation}{section}

\newcommand{\bbR}{\mathbb{R}}
\newcommand{\cN}{\mathcal{N}}
\newcommand{\bbE}{\mathbb{E}}

\newcommand{\toW}{\overset{W_2}{\to}}

\newcommand{\MLE}{\operatorname{MLE}}

\newcommand{\mmse}{\operatorname{mmse}}
\newcommand{\diag}{\operatorname{diag}}
\newcommand{\Tr}{\operatorname{Tr}}

\newcommand{\eps}{\varepsilon}
\newcommand{\DKL}{D_{\mathrm{KL}}}
\newcommand{\bX}{\mathbf{X}}
\newcommand{\bY}{\mathbf{Y}}
\newcommand{\bP}{\mathbf{P}}
\newcommand{\bff}{\mathbf{f}}
\newcommand{\bg}{\mathbf{g}}

\newcommand{\bu}{\mathbf{u}}
\newcommand{\bv}{\mathbf{v}}
\newcommand{\bx}{\mathbf{x}}

\newcommand{\bU}{\mathbf{U}}
\newcommand{\bV}{\mathbf{V}}
\newcommand{\bW}{\mathbf{W}}
\newcommand{\bZ}{\mathbf{Z}}
\newcommand{\bF}{\mathbf{F}}
\newcommand{\bG}{\mathbf{G}}

\newcommand{\Id}{\mathrm{Id}}
\newcommand{\cP}{\mathcal{P}}

\newcommand{\btheta}{\pmb{\theta}}

\newcommand{\bTheta}{\mathbf{\Theta}}
\newcommand{\bO}{\mathbf{O}}
\newcommand{\bQ}{\mathbf{Q}}
\newcommand{\bA}{\mathbf{A}}
\newcommand{\bB}{\mathbf{B}}

\DeclareMathOperator*{\argmax}{arg\,max}
\newcommand{\der}{\mathsf{d}}

\DeclareMathOperator*{\PL}{PL}

\usepackage{xparse}

\ExplSyntaxOn
\keys_define:nn { miguel/label }
{
	label   .tl_set:N = \l_miguel_label_tl,
	unknown .code:n   = \clist_put_right:Nx \l_miguel_label_clist
	{ \l_keys_key_tl = \exp_not:n { #1 } }
}
\clist_new:N \l_miguel_label_clist
\box_new:N \l_miguel_label_box
\box_new:N \l_miguel_label_image_box
\NewDocumentCommand{\xincludegraphics}{O{}m}
{
	\tl_clear:N \l_miguel_label_tl
	\clist_clear:N \l_miguel_label_clist
	\keys_set:nn { miguel/label } { #1 }
	\tl_if_empty:NTF \l_miguel_label_tl
	{
		\miguel_includegraphics:Vn \l_miguel_label_clist { #2 }
	}
	{
		\hbox_set:Nn \l_miguel_label_image_box
		{
			\miguel_includegraphics:Vn \l_miguel_label_clist { #2 }
		}
		\hbox_set:Nn \l_miguel_label_box
		{
			\skip_horizontal:n { -15pt }
			\fcolorbox{white}{white}{\footnotesize \tl_use:N \l_miguel_label_tl}
		}
		\leavevmode
		\box_use:N \l_miguel_label_image_box
		\skip_horizontal:n { -\box_wd:N \l_miguel_label_image_box }
		\hbox_overlap_right:n
		{
			\box_move_up:nn
			{
				\box_ht:N \l_miguel_label_image_box - 
				\box_ht:N \l_miguel_label_box - -3pt
			}
			{ \box_use_drop:N \l_miguel_label_box }
		}
		\skip_horizontal:n { \box_wd:N \l_miguel_label_image_box }
	}
}
\cs_new_protected:Nn \miguel_includegraphics:nn
{
	\includegraphics[#1]{#2}
}
\cs_generate_variant:Nn \miguel_includegraphics:nn { V }
\ExplSyntaxOff

\begin{document}

\title{Empirical Bayes PCA in high dimensions}
\date{}
\author{Xinyi Zhong\thanks{These authors contributed equally.\newline
XZ: Yale University, Department of Statistics and Data Science.
\texttt{xinyi.zhong@yale.edu} \newline
CS: Yale University, Department of Biostatistics.
\texttt{c.su@yale.edu} \newline
ZF: Yale University, Department of Statistics and Data Science.
\texttt{zhou.fan@yale.edu}}
\and Chang Su\footnotemark[1] \and Zhou Fan}

\maketitle

\begin{abstract}
When the dimension of data is comparable to or larger than the number of
data samples, Principal Components Analysis (PCA) may exhibit problematic
high-dimensional noise. In this work, we propose an Empirical Bayes PCA method
that reduces this noise by estimating a joint prior distribution for the
principal components. EB-PCA is based on the classical Kiefer-Wolfowitz
nonparametric MLE for empirical Bayes estimation, distributional results
derived from random matrix theory for the sample PCs, and iterative refinement
using an Approximate Message Passing (AMP) algorithm. In theoretical ``spiked'' 
models, EB-PCA achieves Bayes-optimal estimation accuracy in the same settings as an
oracle Bayes AMP procedure that knows the true priors. Empirically, EB-PCA
significantly improves over PCA when there is strong prior structure, both in
simulation and on quantitative benchmarks constructed from the 1000 Genomes
Project and the International HapMap Project. An illustration is presented for
analysis of gene expression data obtained by single-cell RNA-seq.
\end{abstract}

\section{Introduction}

Principal components analysis (PCA) is a widely used technique for
dimensionality reduction. However, when the dimension of the data may be
comparable to or larger than the number of available data samples, it is known 
that the sample principal components (PCs) may exhibit phenomena of
high-dimensional noise \cite{lu2002sparse,johnstone2009consistency}.
We propose a method called
EB-PCA for reducing this noise, using the classical
statistical idea of empirical Bayes \cite{robbins1956empirical,efron2012large}.

Figure \ref{fig:1000genomesintro} illustrates EB-PCA on a genetics
example. Panel (a) displays the top 4 PCs of a genotype matrix
from the 1000 Genomes Project \cite{con2015global}, containing genotypes of
2504 individuals at 100{,}000 common single
nucleotide polymorphisms (SNPs). The PCs depict the stratification of these
individuals according to five broad ethnic populations. Here, the number of
SNPs far exceeds the dimension 2504 of each PC, and the estimation noise is
small. This allows us to interpret the PCs in panel (a) as an approximate
``ground truth''.

\begin{figure}[t]
\minipage{0.33\columnwidth}
\xincludegraphics[width=0.85\textwidth,label=(a)]{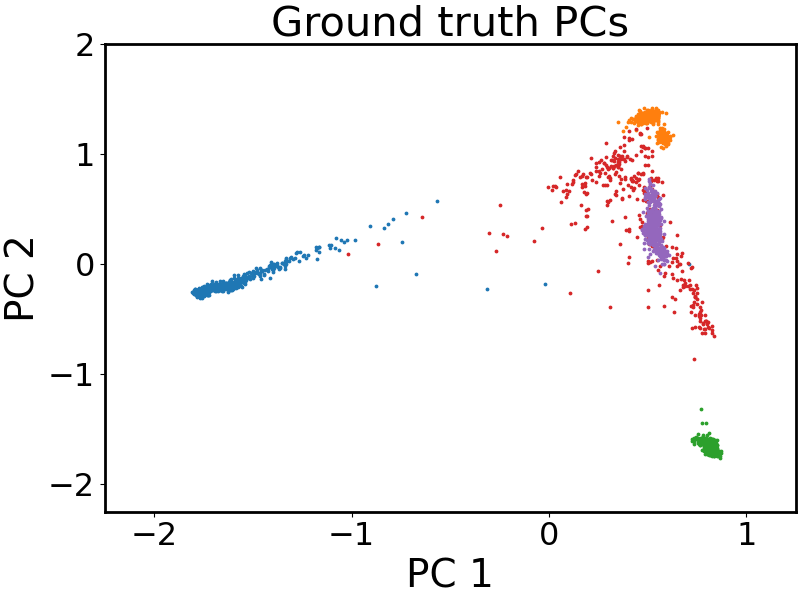}\\
\xincludegraphics[width=0.85\textwidth]{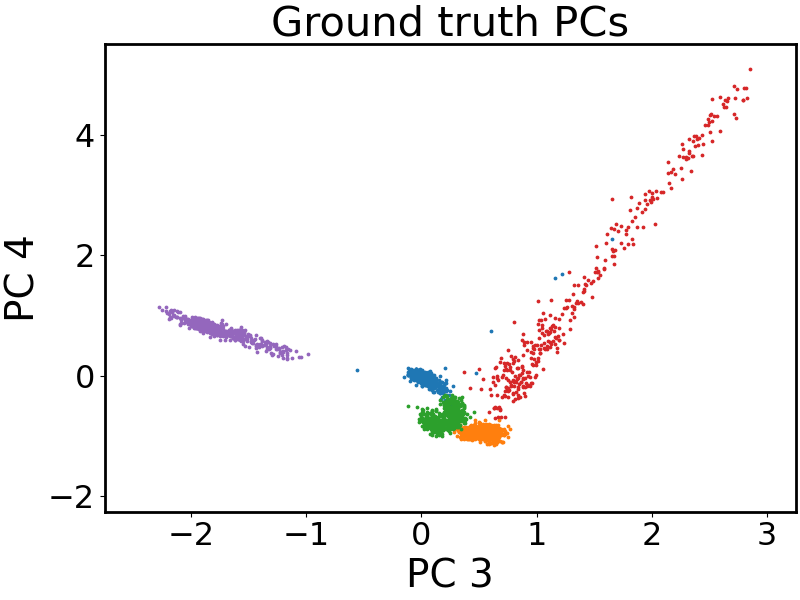}\endminipage
\minipage{0.33\columnwidth}
\xincludegraphics[width=0.85\textwidth,label=(b)]{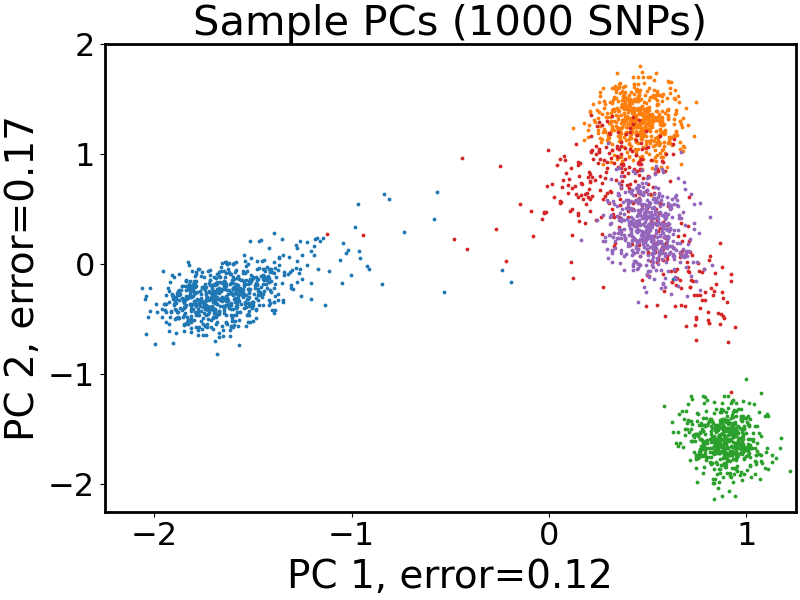}\\
\xincludegraphics[width=0.85\textwidth]{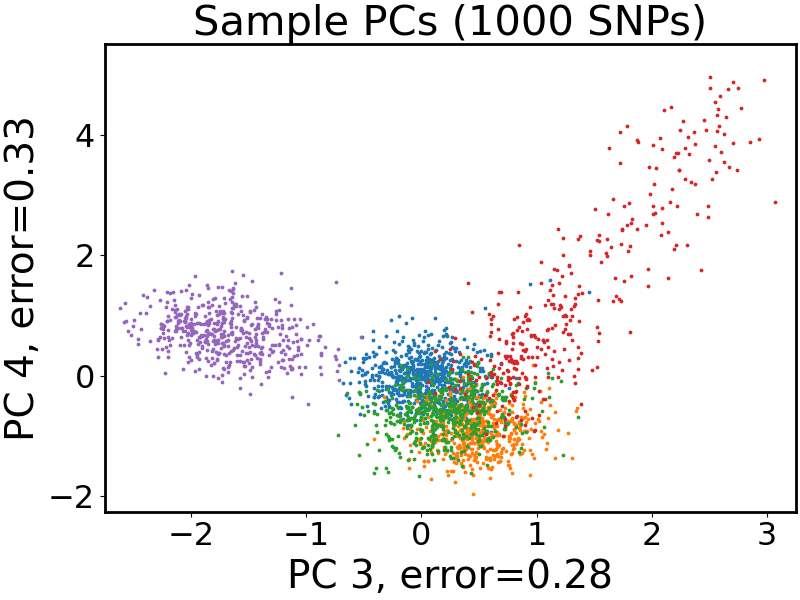}\endminipage
\minipage{0.33\columnwidth}
\xincludegraphics[width=0.85\textwidth,label=(c)]{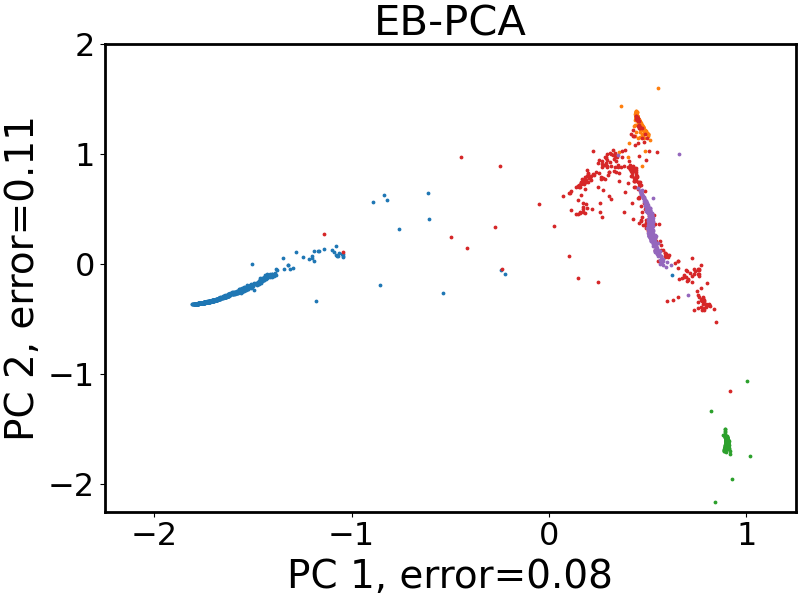}\\
\xincludegraphics[width=0.85\textwidth]{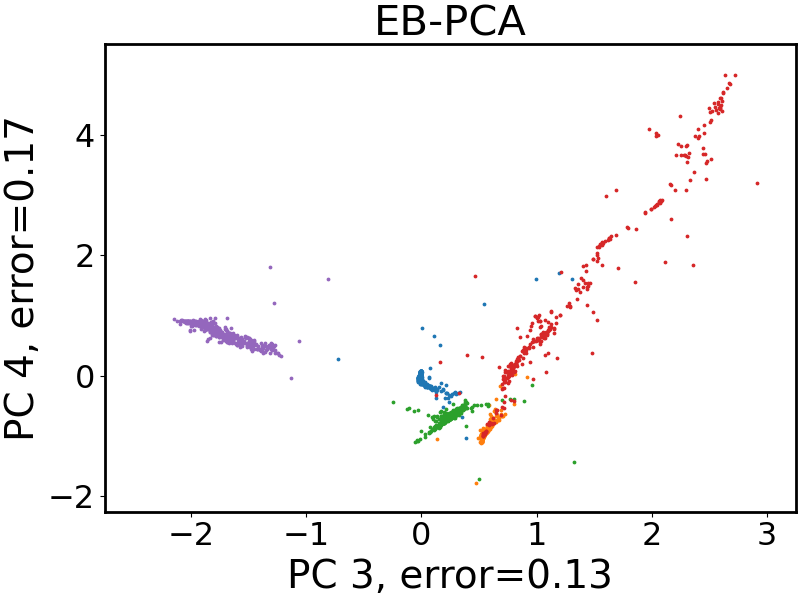}
\endminipage\\
\minipage{\columnwidth}
\includegraphics[width=\textwidth]{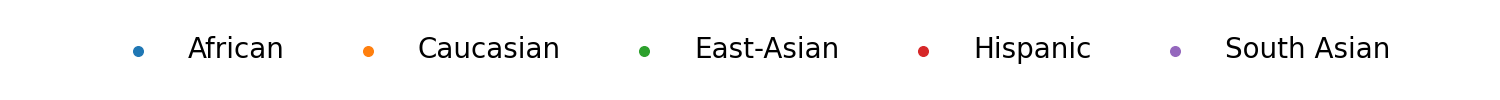}
\endminipage
\caption{Illustration of EB-PCA on genotype data from the 1000 Genomes Project.
(a) 1st vs.\ 2nd PC and 3rd vs.\ 4th PC, for the genotypes of 2504
individuals across $100{,}000$ common SNPs. Each scatterplot has 2504 data
points, representing the embedding of these individuals into a 4-dimensional
space, with points colored by the individuals' ethnicity. We take these PCs
as the ground truth. (b) PCs computed from a random subsample of $1000$
SNPs. Substantial high-dimensional noise is observed in these PCs. (c) The
EB-PCA estimates of the top 4 PCs, computed from the same subsampled data as
in panel (b). These estimates are much closer to the ground-truth PCs in panel
(a) and have quantitatively lower estimation error.}\label{fig:1000genomesintro}
\end{figure}

The phenomenon of high-dimensional noise is illustrated in
panel (b), which displays the top 4 PCs for genotypes of the same
2504 individuals subsampled at only
1000 randomly selected SNPs. Applying EB-PCA to this reduced data
of 1000 SNPs yields the PC estimates displayed in panel (c). These are
remarkably close to the PCs in panel (a) computed on all
100{,}000 SNPs, even though EB-PCA has only access to the 1000 subsampled SNPs.
In Section~\ref{sec:examples}, we use this subsampling
approach to demonstrate a sizeable quantitative improvement of EB-PCA over PCA.
We also illustrate an application to single-cell RNA-seq gene expression
data where a ground truth is unknown.

A central component of the method is a Bayes Approximate Message
Passing (AMP) procedure \cite{rangan2012iterative,montanari2017estimation}
that implements approximate Bayesian inference for low-rank matrix estimation in
high dimensions. EB-PCA adapts Bayes AMP, which requires knowledge of the true
prior distributions, to more typical settings in practice where such
information is unavailable, by nonparametrically estimating the priors from the 
sample PCs and the AMP iterates. Similar strategies can be applied to Bayes AMP
algorithms for other applications.

To describe the main ideas behind EB-PCA, consider
a rank-one signal-plus-noise model for the observed data,
\begin{equation}\label{eq:rankonemodel}
\bY=\frac{s}{n} \cdot \bu\bv^\top+\bW \in \bbR^{n \times d}
\end{equation}
where $\bu \in \bbR^n$ and $\bv \in \bbR^d$ are the left and right true PCs of
interest, with associated signal strength $s>0$, and $\bW \in \bbR^{n \times d}$ is i.i.d.\ Gaussian
observational noise. We discuss possible extensions to more general noise in Section~\ref{sec:conclusion}. We will refer to the leading
left- and right-singular vectors $\bff \in \bbR^n$ and $\bg \in \bbR^d$ of $\bY$ as the
sample PCs.

The EB-PCA approach consists of three main ideas, each of which is individually
well-studied:
\begin{enumerate}
\item {\bf Kiefer-Wolfowitz NPMLE.} Consider the classical \emph{compound
decision problem} of estimating $\btheta \in \bbR^n$ from a Gaussian
observation vector $\bx \sim \cN(\mu \cdot \btheta,\,\sigma^2 \cdot
\Id_{n \times n})$, for two known scalar parameters $\mu,\sigma^2>0$. The
empirical Bayes paradigm first posits a prior distribution $\pi_*$
for the coordinates of $\btheta$, then estimates $\pi_*$ by an estimator $\pi$ 
based on the marginal density of the observed coordinates of $\bx$, and
finally applies Bayes's rule defined by $\pi$ to ``denoise'' $\bx$ and obtain
the estimate of $\btheta$.

A nonparametric implementation of this paradigm was described in
\cite{robbins1950generalization,kiefer1956consistency}, which suggested
estimating $\pi_*$ by the nonparametric maximum likelihood estimator (NPMLE)
that maximizes the likelihood of $\bx$ over all prior probability
distributions $\pi$ on the real line. It was shown in
\cite{kiefer1956consistency,laird1978nonparametric,lindsay1983geometryI} that
such a maximizer $\pi$ exists with discrete and finite support. We denote by
\begin{equation*}
\theta(\bx \mid \mu,\sigma^2,\pi)=\bbE_\pi[\btheta \mid \bx]
\end{equation*}
the empirical Bayes posterior mean estimate of $\btheta$ using this estimated
prior $\pi$.

\item {\bf Random matrix asymptotics for sample PCs.}
In the model of \eqref{eq:rankonemodel}, an influential line of work
\cite{baik2005phase,paul2007asymptotics,nadler2008finite,benaych2012singular}
has quantified the asymptotic error of the sample PCs $(\bff,\bg)$ for the true
PCs $(\bu,\bv)$ when $n,d \to \infty$ simultaneously such that
$d/n \to \gamma \in (0,\infty)$.
This work showed that in this high-dimensional limit,
\[\langle \bff,\bu \rangle \to \bar{\mu}_* \equiv \bar{\mu}_*(s,\,\gamma),
\qquad \langle \bg,\bv \rangle \to \mu_* \equiv \mu_*(s,\,\gamma)\]
for two inner products $\mu_*,\bar{\mu}_* \in [0,1)$ that depend only on the
signal strength $s$ and the dimension ratio $\gamma$. For $s$ larger than
a certain phase transition threshold $s_*(\gamma)$, the leading singular value
of $\bY$ emerges as an outlier from the bulk distribution of its remaining
singular values, the inner products $\bar{\mu}_*,\mu_*$ are strictly positive,
and $\bg$ has an approximate entrywise Gaussian law
\[\bg \approx \cN(\mu_* \cdot \bv,\,\sigma_*^2 \cdot \Id_{d \times d}),
\qquad \sigma_*^2=1-\mu_*^2.\]
An analogous approximation holds for $\bff$ and $\bu$. 
This provides a connection to the compound decision problem above. EB-PCA
estimates $(\mu_*,\sigma_*^2)$ by estimating $s$, and applies
the Kiefer-Wolfowitz NPMLE to obtain an empirical Bayes estimate $\hat{\bv}$
for $\bv$.

\item {\bf Iterative refinement via AMP.} If this estimate $\hat{\bv}$ is more
accurate than the original sample PC $\bg$ for $\bv$, then we expect
$\bY\hat{\bv}$ to be more accurate than $\bY\bg \propto \bff$ for $\bu$. This
suggests that empirical Bayes denoising should be applied to $\bY\hat{\bv}$
instead of $\bff$ to estimate $\bu$, and leads to an iterative
idea \cite{wang2018empirical} of initializing $\bg^0=\bg$ and computing
\begin{align}
\begin{aligned}
\bv^t&=\theta(\bg^t \mid \mu_t,\sigma_t^2,\pi_t),\qquad
\bff^t=\bY \bv^t,\\
\bu^t&=\theta(\bff^t \mid \bar{\mu}_t,\bar{\sigma}_t^2,\bar{\pi}_t),\qquad
\bg^{t+1}=\bY^\top \bu^t.
\end{aligned}
\label{eq:naiveMF}
\end{align}
Here, $\pi_t,\bar{\pi}_t$ are nonparametrically estimated priors
and $\mu_t,\sigma_t^2,\bar{\mu}_t,\bar{\sigma}_t^2$ are scalar parameters
in each iteration. In the first iteration, $\bv^0=\hat{\bv}$ is the
above empirical Bayes estimate of $\bv$.

Unfortunately, this procedure does not
ensure that $(\bff^t,\bg^t)$ have approximate entrywise Gaussian
laws after this first iteration, breaking the connection to the compound 
decision problem in subsequent iterations. EB-PCA applies instead an
AMP algorithm as developed
in \cite{rangan2012iterative,montanari2017estimation},
\begin{align*}
\bv^t&=\theta(\bg^t \mid \mu_t,\sigma_t^2,\pi_t), \qquad
\bff^t=\bY \bv^t-b_t\bu^{t-1},\\
\bu^t&=\theta(\bff^t \mid \bar{\mu}_t,\bar{\sigma}_t^2,\bar{\pi}_t), \qquad
\bg^{t+1}=\bY^\top \bu^t-\bar{b}_t\bv^t.
\end{align*}
The \emph{Onsager corrections} $b_t\bu^{t-1}$ and $\bar{b}_t\bv^t$ are defined
so as to remove a bias of $(\bff^t,\bg^t)$ in the directions of
$(\bu^{t-1},\bv^t)$ and restore the entrywise Gaussian approximations.
\end{enumerate}

EB-PCA is most effective when there is strong prior structure for the true PCs.
We described the rank-one model of \eqref{eq:rankonemodel} for clarity, but in
many examples including Figure \ref{fig:1000genomesintro}, there is
stronger structure jointly over several PCs. In these examples, we
learn a joint prior in $k>1$ dimensions, where $k$ is the number of PCs to 
be simultaneously estimated. The result of
Figure \ref{fig:1000genomesintro}(c) is obtained by joint empirical Bayes
estimation for all $k=4$ depicted PCs, rather than estimating each PC
individually. We describe the method in more detail in Section~\ref{sec:AMP}
and present theoretical guarantees in Section~\ref{sec:theory}.

This application of empirical Bayes methodology to PCA via an iterative
algorithm is closely related to earlier and
inspirational work by \cite{wang2018empirical}, who proposed
an empirical Bayes matrix factorization (EBMF) method that yields
the iterations of \eqref{eq:naiveMF}. EBMF is
derived from a ``naive mean-field'' variational approximation to the 
posterior distribution of $(\bu,\bv)$, and we discuss further in
Section~\ref{sec:EBMF}
the relation between EB-PCA and this naive mean-field approach.

\subsection{Related literature}

The possible inconsistency of PCA in high dimensions has been discussed in
\cite{lu2002sparse,johnstone2009consistency,johnstone2018pca}, and
improving PCA using prior structure has been a long-standing goal.
A large body of literature has notably studied sparse PCA methods, which
improve over PCA under sparsity assumptions
\cite{cadima1995loading,jolliffe2003modified,d2005direct,zou2006sparse,amini2008high,birnbaum2013minimax,cai2013sparse,fan2013large,ma2013sparse,vu2013fantope}.
Figure \ref{fig:1000genomesintro} illustrates an example
where the PCs indeed have strong prior structure, but this structure
is not well-characterized by entrywise
sparsity. We believe that such examples may be common across scientific applications, and
this forms the primary motivation for our work.

EB-PCA is complementary to spectral shrinkage methods
that preserve the sample PCs but shrink or truncate the singular values
\cite{cai2010singular,ledoit2012nonlinear,shabalin2013reconstruction,nadakuditi2014optshrink,chatterjee2015matrix,gavish2017optimal}.
These methods have been motivated in part by a perspective that, in the absence 
of prior structural knowledge about the PCs, ``...it is reasonable to
require that covariance matrix estimators be rotation-equivariant [and have]
the same eigenvectors as the sample covariance matrix''
\cite{ledoit2012nonlinear}. Our work stands contrary to this perspective,
illustrating that empirical Bayes ideas can substantially improve over such
equivariant procedures even without knowledge of prior structure, as long as
some structure is present.

EB-PCA is an empirical Bayes implementation of the multivariate
Bayes AMP algorithm described by \cite{montanari2017estimation}.
AMP algorithms were first developed for CDMA, compressed sensing, and
generalized linear model applications by~\cite{kabashima2003cdma,donoho2009message,rangan2011generalized}.
Empirical Bayes versions of AMP for compressed sensing and GLMs were studied by
\cite{vila2011expectation,krzakala2012statistical,vila2013expectation,kamilov2014approximate},
in univariate and parametric contexts that are different from the
nonparametric perspective of our work.

AMP algorithms for PCA have been studied in a line of work including
\cite{rangan2012iterative,matsushita2013low,deshpande2014information,montanari2015non,lesieur2015phase,kabashima2016phase,deshpande2017asymptotic}.
These algorithms originally required an informative initialization
independent of $\bY$, and \cite{montanari2017estimation} provided the 
practical extension of initializing at the sample PCs.
A related line of work
\cite{lesieur2015mmse,barbier2016mutual,miolane2017fundamental,alaoui2018estimation,lelarge2019fundamental,barbier2019adaptive} has explored more generally the
limits of low-rank matrix estimation with Bayesian priors. In particular,
\cite{deshpande2014information,barbier2016mutual,deshpande2017asymptotic}
showed in various Bayesian rank-one models that AMP algorithms can achieve the
asymptotically optimal squared-error Bayes risk, which has been characterized in
\cite{lesieur2015mmse,barbier2016mutual,miolane2017fundamental,lelarge2019fundamental}.
A second motivation for our work is to bring this important body of
statistical theory a step closer to statistical practice. Our results imply
that AMP algorithms can achieve Bayes-optimal estimation even without
knowledge of the true priors.

The initial step of EB-PCA relies on quantitative understanding of spectral
behavior in spiked random matrix models
\cite{johnstone2001distribution,baik2005phase,baik2006eigenvalues,bai2008central}. 
We assume in this work a 
Gaussian model, where the error of the sample singular vectors was
first studied in \cite{paul2007asymptotics,nadler2008finite}. Such
results have been extended to non-Gaussian settings in
\cite{capitaine2011free,knowles2013isotropic,knowles2014outliers,bloemendal2016principal,capitaine2018limiting,ding2020high},
models with non-white noise in
\cite{mestre2008asymptotic,benaych2011eigenvalues,benaych2012singular,bai2012sample},
and more general asymptotic regimes in
\cite{jung2009pca,shen2013consistency,wang2017asymptotics}.
Related distributional properties of singular vectors were
recently studied in \cite{capitaine2018non,bao2018singular,bao2020statistical}.

The Kiefer-Wolfowitz NPMLE was proposed in
\cite{robbins1950generalization,kiefer1956consistency}.
Identifiability, existence and uniqueness, asymptotic consistency, and
discreteness of the support were studied in
\cite{kiefer1956consistency,simar1976maximum,laird1978nonparametric,jewell1982mixtures,lindsay1983geometryI,lindsay1983geometryII,lindsay1993uniqueness}, and a detailed treatment of these topics is provided in \cite{lindsay1995mixture}.
\cite{ghosal2001entropies,zhang2009generalized,jiang2009general,saha2020nonparametric}
studied the rate of convergence of the NPMLE and associated empirical Bayes
estimator, and our analyses draw on their techniques.
Recently, \cite{polyanskiy2020self} showed that these results on estimation
rates are connected also to the size of the discrete NPMLE support.
Computing and approximating the NPMLE has been discussed in
\cite{bohning1992computer,bohning1999computer,lashkari2008convex,koenker2014convex,feng2016approximate}.

\section{The EB-PCA method}\label{sec:method}

\subsection{Model}\label{sec:model}

The EB-PCA algorithm is derived in
the following rank-$k$ version of the model in \eqref{eq:rankonemodel},
\begin{equation}\label{eq:model}
\bY=\frac{1}{n} \cdot \bU{S}\bV^\top+\bW
=\sum_{i=1}^k \frac{s_i}{n}\cdot \bu_i\bv_i^\top+\bW \in \bbR^{n \times d}.
\end{equation}
The columns of $\bU=(\bu_1,\ldots,\bu_k) \in \bbR^{n \times k}$ and
$\bV=(\bv_1,\ldots,\bv_k) \in \bbR^{d \times k}$
are $k$ left and right principal components of interest, and
${S}=\diag(s_1,\ldots,s_k) \in \bbR^{k \times k}$
contains the signal strengths of these PCs.
$\bW \in \bbR^{n \times d}$ is observational noise,
which we assume has entries $w_{ij} \overset{iid}{\sim} \cN(0,1/n)$.

\begin{remark}\label{remark:scaling}
We write the noise variance of $w_{ij}$ for convenience as $1/n$, rather than a
more general $\tau^2/n$, to avoid carrying $\tau^2$ throughout our
formulas. This is without loss of generality, as $\bY,S,\bW$ may be rescaled by
a common factor $\tau$. In practice, we estimate this noise
variance and rescale the data to match this scaling.
Given $\bY_{\textrm{obs}}\in \mathbb{R}^{n\times d}$, we may estimate its
entrywise residual variance upon regressing out its top $k$ PCs,
\begin{equation}\label{eq:tausq}
\hat{\tau}^2=d^{-1}\cdot \|\mathbf{R}\|_F^2, \qquad
\mathbf{R}=\bY_{\textrm{obs}}-\textrm{top $k$ PCs of } \bY_{\textrm{obs}}.
\end{equation}
We then set $\bY=\bY_{\textrm{obs}}/\hat{\tau}$. Consistency of
$\hat{\tau}^2$ is discussed in Appendix \ref{sec:normalization}.
\end{remark}

We study this model in the high-dimensional limit $n,d \to \infty$ such
that $k$ and $\gamma \equiv d/n$
are both fixed. It is helpful to keep in mind a Bayesian
setting where the rows of $\bU$ and $\bV$ are generated according to two fixed
prior probability distributions $\bar{\pi}_*$ and $\pi_*$ on $\bbR^k$ (although we
will only require empirical convergence to these priors in the later theory).
The goal of EB-PCA is then to estimate these two priors from the data
$\bY$, and to use these estimated priors to perform Bayesian estimation of
$\bU$ and $\bV$.

To fix the scaling of the PCs, we normalize $\bar{\pi}_*$ and $\pi_*$ to satisfy
$\bbE_{U \sim \bar{\pi}_*}[U_i^2]=1$ and
$\bbE_{V \sim \pi_*}[V_i^2]=1$ for all $i=1,\ldots,k$. This ensures
\begin{equation}\label{eq:normalization}
n^{-1}\|\bu_i\|_2^2 \to 1, \qquad d^{-1}\|\bv_i\|_2^2 \to 1.
\end{equation}
We will also assume
$\bbE_{U \sim \bar{\pi}_*}[U_iU_j]=0$ and $\bbE_{V \sim \pi_*}[V_iV_j]=0$ for
all $1 \leq i \neq j \leq k$, so that
\begin{equation}\label{eq:orthogonality}
n^{-1}\bu_i^\top \bu_j \to 0, \qquad d^{-1}\bv_i^\top \bv_j \to 0,
\end{equation}
lending to the interpretations of $\bu_i$ and $\bv_i$ as the (orthogonal)
principal components. Under these scalings, the $k$ singular values
of $n^{-1}\bU{S}\bV^\top$ converge to the limits
$\sqrt{\gamma} s_1>\ldots>\sqrt{\gamma} s_k$
and we make the simplifying assumption that these limit values are distinct.
Note that we will not enforce the orthogonality conditions
$\bbE_{U \sim \bar{\pi}}[U_iU_j]=0$ and $\bbE_{V \sim \pi}[V_iV_j]=0$ for
the estimated priors in the later algorithm, but approximate orthogonality will
automatically hold from initializing the algorithm at the sample PCs.

Turning to the sample PCs, let us write the best rank-$k$ approximation
for $\bY$ as
\[\frac{1}{n}\cdot
\bF{\Lambda}\bG^\top=\sum_{i=1}^k \frac{\lambda_i}{n} \cdot \bff_i\bg_i^\top.\]
Here, the columns of $\bF=(\bff_1,\ldots,\bff_k) \in \bbR^{n \times k}$ and
$\bG=(\bg_1,\ldots,\bg_k) \in \bbR^{d \times k}$ are the top $k$ left and right
singular vectors of $\bY$, normalized analogously with a sign convention so that
for all $1\leq i \neq j \leq k$,
\begin{equation}
d^{-1}\|\bg_i\|^2=n^{-1}\|\bff_i\|^2=1,\qquad
\bu_i^\top \bff_i\geq 0,\quad\bv_i^\top\bg_i \geq 0,\qquad
d^{-1}\bg_i^\top \bg_j=n^{-1}\bff_i^\top \bff_j=0.
\end{equation}
We set ${\Lambda}=\diag(\lambda_1,\ldots,\lambda_k)$.
Then the $k$ largest singular values of $\bY$ are given by
$\sqrt{\gamma}\lambda_1 \geq \ldots \geq \sqrt{\gamma}\lambda_k$.

Under this model, the following phase transition occurs for
the leading $k$ sample singular values and singular vectors of $\bY$
\cite{baik2005phase,paul2007asymptotics,benaych2012singular}: setting $s_*(\gamma)=\gamma^{-1/4}$,
for \emph{super-critical} PCs such that $s_i>s_*(\gamma)$, we have
\[\lim_{n,d \to \infty} \sqrt{\gamma} \cdot \lambda_i>\lambda_+,
\qquad \lim_{n,d \to \infty}
n^{-1}\bff_i^\top \bu_i>0, \qquad \lim_{n,d \to \infty} d^{-1}\bg_i^\top \bv_i>0\]
where $\lambda_+=1+\sqrt{\gamma}$ is the upper edge of the ``bulk distribution''
of the noise singular values. Conversely,
for \emph{sub-critical} PCs such that $s_i \leq s_*(\gamma)$,
\[\lim_{n,d \to \infty} \sqrt{\gamma} \cdot \lambda_i=\lambda_+,
\quad \lim_{n,d \to \infty} n^{-1}\bff_i^\top \bu_i=0, \quad \lim_{n,d \to
\infty} d^{-1}\bg_i^\top \bv_i=0.\]
Thus the $i^\text{th}$ sample singular value is absorbed into the bulk,
and the sample PCs are nearly orthogonal to the true PCs. For notational and
expositional clarity, we will assume
\[s_i>s_*(\gamma) \quad \text{ for all } \quad i=1,\ldots,k,\]
i.e.\ all $k$ of the leading PCs are super-critical. Our theoretical results
may be extended to more general settings having both super-critical and
sub-critical PCs, where EB-PCA is applied only to the super-critical PCs that 
have positive alignment with the truth.

\begin{remark}\label{remark:spikedcovariance}
If the rows of $\bU$ are drawn from
$\bar{\pi}_*=\cN(0,\Id_{k \times k})$, then the rows of $\sqrt{n} \cdot \bY$
marginalized over $\bU$ are distributed as $\cN(0,\Sigma)$ where
\[\Sigma=\sum_{i=1}^k s_i^2 \cdot \frac{\bv_i\bv_i^\top}{n}+\Id_{d \times d}.\]
Thus $\bY^\top \bY$ follows the spiked covariance model introduced in
\cite{johnstone2001distribution}, and our results pertain also to 
estimating the spike eigenvectors of $\Sigma$.
In this model, it would be reasonable to consider a version
of EB-PCA that fixes $\bar{\pi}_*=\cN(0,\Id)$, only estimates $\pi_*$, and
performs Bayesian estimation of $\bV$ but not of $\bU$. We will focus instead
on the more general scenario where both $\bU$ and $\bV$ may have non-Gaussian
structure, and describe EB-PCA for estimating both matrices.
\end{remark}

\subsection{Empirical Bayes for the multivariate compound decision
problem}\label{sec:EB}

Let $\pi_*$ be a probability distribution on $\bbR^k$. For two
given matrices $M,\Sigma \in \bbR^{k \times k}$ where $\Sigma$ is
symmetric positive-definite, consider the compound decision model
\begin{equation}\label{eq:compoundmodel}
\Theta \sim \pi_*, \qquad X \mid \Theta \sim \cN(M \cdot \Theta,\;\Sigma)
\end{equation}
for $\Theta,X \in \bbR^k$. We will denote the Bayes posterior mean estimate
of $\Theta$ based on $X$ as
\begin{equation}\label{eq:multivariatedenoise}
\theta(X \mid M,\Sigma,\pi_*)=\bbE_{\pi_*}[\Theta \mid X].
\end{equation}

Suppose now that $\pi_*$ is unknown, but belongs to a known class of probability
distributions $\cP$ over $\bbR^k$. In a model of $n$ i.i.d.\ samples
$x_1,\ldots,x_n$ distributed according to \eqref{eq:compoundmodel}, stacked as
the rows of a matrix $\bX \in \bbR^{n \times k}$, consider the maximum
likelihood estimator
\begin{align}
\pi&=\MLE(\bX \mid M,\Sigma,\cP)\label{eq:NPMLE}\\
&\equiv \argmax_{\pi \in \cP} \prod_{i=1}^n
\int \frac{1}{(2\pi)^{k/2}|\Sigma|^{1/2}}\cdot
\exp\left(-\frac{(x_i-M \cdot \theta_i)^\top \Sigma^{-1}
(x_i-M \cdot \theta_i)}{2}\right) \der\pi(\theta_i).\notag
\end{align}
This integral is the marginal Gaussian mixture density of $x_i$ in the
model of \eqref{eq:compoundmodel}, and the notation makes explicit the dependence
of $\pi$ on the prior class $\cP$. We will be interested primarily
in nonparametric classes $\cP$, and
$\pi$ is a nonparametric maximum likelihood estimate (NPMLE) for $\pi_*$. 
In our implementation, we take 
$\cP$ as the class of all probability distributions
on $\bbR^k$,
and approximate this class $\cP$ using a discrete support by applying the ``exemplar method'' of \cite{lashkari2008convex}. See Appendix~\ref{sec:data} for details.

Stacking $\theta_1,\ldots,\theta_n$ as the rows of $\bTheta \in
\bbR^{n \times k}$, the model for $\bX$ may be written as
\begin{equation}\label{eq:compoundmodelmatrix}
\bX=\bTheta M^\top+\bZ\Sigma^{1/2}, \quad \bZ \in \bbR^{n \times k}
\text{ has i.i.d.\ } \cN(0,1) \text{ entries.}
\end{equation}
The NPMLE $\pi$ defines an empirical Bayes estimate of $\bTheta$,
which applies the posterior mean function $\theta(\cdot)$ for the
estimated prior $\pi$ row-wise to $\bX$. We denote this by
\begin{equation*}
\theta(\bX \mid M,\Sigma,\pi)=\bbE_\pi[\bTheta \mid \bX].
\end{equation*}

\subsection{Initial denoising of the sample PCs}\label{sec:initdenoise}

In the model of \eqref{eq:model}, as $n,d \to \infty$,
the precise forms of the limits of the super-critical singular values
$\sqrt{\gamma} \cdot \lambda_i$ and corresponding
PCs $\bff_i,\bg_i$ of $\bY$ are given by
\begin{gather}
\sqrt{\gamma} \cdot \lambda_i \to \sqrt{(\gamma s_i^2+1)
(s_i^2+1)/s_i^2},\notag\\
n^{-1} \bff_i^\top \bu_i \to \bar{\mu}_{*,i}
\equiv \sqrt{1-\bar{\sigma}_{*,i}^2},
\qquad
d^{-1} \bg_i^\top \bv_i \to \mu_{*,i} \equiv \sqrt{1-\sigma_{*,i}^2},\notag \\
\bar{\sigma}_{*,i}^2=\frac{1+s_i^2}{s_i^2 (\gamma s_i^2+1)},
\qquad \sigma_{*,i}^2=\frac{1+\gamma s_i^2}{\gamma s_i^2 (s_i^2+1)}. \label{eq:PCAsigma}
\end{gather}
(See Lemma \ref{lemma:BBP}.) Setting
\begin{align}
\bar{M}_*&=\diag(\bar{\mu}_{*,1},\ldots,\bar{\mu}_{*,k}),\qquad
M_*=\diag(\mu_{*,1},\ldots,\mu_{*,k}),\label{eq:PCAM}\\
\bar{\Sigma}_*&=\diag(\bar{\sigma}_{*,1}^2,\ldots,\bar{\sigma}_{*,k}^2),\qquad
\Sigma_*=\diag(\sigma_{*,1}^2,\ldots,\sigma_{*,k}^2),\label{eq:PCASigma}
\end{align}
a consequence is that $\bF \in \bbR^{n \times k}$ and
$\bG \in \bbR^{d \times k}$ have the Gaussian approximations
\begin{equation}\label{eq:PCAgaussianapprox}
\bF \approx \bU\bar{M}_*^\top+\bar{\bZ}\bar{\Sigma}_*^{1/2},
\qquad \bG \approx \bV M_*^\top+\bZ\Sigma_*^{1/2}
\end{equation}
for large $n$ and $d$, where $\bar{\bZ},\bZ \in \bbR^{n \times k}$ have
i.i.d.\ $\cN(0,1)$ entries.
This relates the behavior of the sample PCs
$\bF$ and $\bG$ to the multivariate compound decision
model in \eqref{eq:compoundmodelmatrix}.

As the true matrices $\bar{M}_*,M_*,\bar{\Sigma}_*,\Sigma_*$ are unknown,
we replace them by consistent estimates to derive empirical Bayes estimators for
$\bU$ and $\bV$:
observe that \eqref{eq:PCAsigma} implies each value $s_i^2$ may 
be consistently estimated by
\begin{equation}\label{eq:shat}
\hat{s}_i^2=\big(\gamma \lambda_i^2-(1+\gamma)
+\sqrt{(\gamma \lambda_i^2-(1+\gamma))^2-4\gamma}\big)/(2\gamma).
\end{equation}
These may be used to obtain plug-in estimators
$\bar{M},M,\bar{\Sigma},\Sigma$ for
$\bar{M}_*,M_*,\bar{\Sigma}_*,\Sigma_*$, which substitute
$\hat{s}_i$ for $s_i$ in \eqref{eq:PCAsigma}. This yields
the initial empirical Bayes estimates of $\bU$ and $\bV$ given by
\begin{align}
\bar{\pi}&=\MLE(\bF \mid \bar{M},\bar{\Sigma},\cP), \qquad
\hat{\bU}=\theta(\bF \mid \bar{M},\bar{\Sigma},\bar{\pi}),\label{eq:initU}\\
\pi&=\MLE(\bG \mid M,\Sigma,\cP), \qquad
\hat{\bV}=\theta(\bG \mid M,\Sigma,\pi).\label{eq:initV}
\end{align}

\subsection{Iterative refinement using AMP}\label{sec:AMP}

We now describe iterative refinement using an AMP algorithm, as discussed also
in Appendix J of \cite{montanari2017estimation}.
This may begin with either the estimate for $\bU$ or $\bV$---here, we
begin with $\bV$.

The algorithm takes the following form: let
$u_1,u_2,\ldots:\bbR^k \to \bbR^k$ and $v_1,v_2,\ldots:\bbR^k \to \bbR^k$ be two
arbitrary sequences of Lipschitz functions. Initialize $\bG^0=\bG$ as the
right sample PCs, and compute for $t=0,1,2,\ldots$
\begin{align}
\begin{aligned}
\bV^t&=v_t(\bG^t)\quad
&\bF^t&=\bY \bV^t-\bU^{t-1} \cdot \gamma \langle \der v_t(\bG^t) \rangle^\top\\
\bU^t&=u_t(\bF^t)\quad
&\bG^{t+1}&=\bY^\top\bU^t-\bV^t \cdot \langle \der u_t(\bF^t) \rangle^\top 
\end{aligned}
\label{eq:bayes_amp}
\end{align}
Here $u_t(\bF^t) \in \bbR^{n \times k}$ and $v_t(\bG^t) \in \bbR^{d \times k}$
denote the applications of $u_t$ and $v_t$ row-wise to $\bF^t$ and
$\bG^t$, $\der u_t:\bbR^k \to \bbR^{k \times k}$
and $\der v_t:\bbR^k \to \bbR^{k \times k}$ denote Jacobian
matrices of these functions, and
$\langle \der u_t(\bF^t) \rangle \in \bbR^{k \times k}$ and
$\langle \der v_t(\bG^t) \rangle \in \bbR^{k \times k}$
are the averages of $\der u_t$ and $\der v_t$ across the rows
of $\bF^t$ and $\bG^t$.

Under the model of \eqref{eq:model},
Gaussian approximations analogous to \eqref{eq:PCAgaussianapprox}
continue to hold for $\bF^t$ and $\bG^t$ across iterations, where
\begin{equation}\label{eq:AMPgaussianapprox}
\bF^t \approx \bU\bar{M}_{*,t}^\top+\bar{\bZ}\bar{\Sigma}_{*,t}^{1/2},
\qquad \bG^t \approx \bV M_{*,t}^\top+\bZ\Sigma_{*,t}^{1/2}.
\end{equation}
Here $\bar{M}_{*,t},M_{*,t},\bar{\Sigma}_{*,t},\Sigma_{*,t}$ are deterministic
matrices that prescribe the parameters of the compound decision model
associated to each iteration. In contrast to the initial state of
\eqref{eq:PCAM}--\eqref{eq:PCASigma}, these matrices are no longer diagonal in
later iterations, if the prior is a general multivariate distribution on
$\mathbb{R}^k$. They evolve over iterations according to a \emph{state evolution}
\[(M_{*,0},\Sigma_{*,0}) \mapsto (\bar{M}_{*,0},\bar{\Sigma}_{*,0}) \mapsto
(M_{*,1},\Sigma_{*,1}) \mapsto (\bar{M}_{*,1},\bar{\Sigma}_{*,1}) \mapsto \cdots\]
given by the initializations $(M_{*,0},\Sigma_{*,0}) \equiv (M_*,\Sigma_*)$
describing the sample PCs
$\bG^0 \equiv \bG$ in \eqref{eq:PCAM}--\eqref{eq:PCASigma}, and by the updates
\begin{align}
\begin{aligned}
\bar{M}_{*,t}=\gamma \bbE[v_t(G_t)V^\top] {S},& \qquad
\bar{\Sigma}_{*,t}=\gamma \bbE[v_t(G_t)v_t(G_t)^\top],\\
M_{*,t+1}=\bbE[u_t(F_t)U^\top] {S},& \qquad
\Sigma_{*,t+1}=\bbE[u_t(F_t)u_t(F_t)^\top].
\end{aligned}\label{eq:SE}
\end{align}
${S}=\diag(s_1,\ldots,s_k)$ is the diagonal matrix of signal strengths in 
\eqref{eq:model}, and the expectations are over the random vectors
\begin{equation*}
\begin{aligned}
U \sim \bar{\pi}_*, \quad
F_t \mid U \sim \cN(\bar{M}_{*,t} \cdot U,\;\bar{\Sigma}_{*,t}), \qquad
V \sim \pi_*,\quad
G_t \mid V \sim \cN(M_{*,t} \cdot V,\;\Sigma_{*,t}).
\end{aligned}
\end{equation*}
These laws of $F_t$ and $G_t$ approximate the row-wise distributions of $\bF^t$
and $\bG^t$.

If $\bar{\pi}_*,\pi_*$ and
$\bar{M}_{*,t},M_{*,t},\bar{\Sigma}_{*,t},\Sigma_{*,t}$ are all \emph{known},
then applying this algorithm with the Bayes posterior mean functions
\begin{equation}\label{eq:bayes_amp_denoisors}
u_t(\bX)=\theta(\bX \mid \bar{M}_{*,t},\bar{\Sigma}_{*,t},\bar{\pi}_*),\qquad
v_t(\bX)=\theta(\bX \mid M_{*,t},\Sigma_{*,t},\pi_*)
\end{equation}
implements an iterative variational
Bayesian inference scheme \cite{montanari2017estimation}.
We will call this the ``oracle'' Bayes AMP
algorithm. For these $u_t,v_t$, we have the identities
$\bbE[u_t(F_t)U^\top]=\bbE[u_t(F_t)u_t(F_t)^\top]$
and $\bbE[v_t(G_t)V^\top]=\bbE[v_t(G_t)v_t(G_t)^\top]$, so
\eqref{eq:SE} yields
\begin{equation}\label{eq:SigmaMrelation}
\bar{M}_{*,t}=\bar{\Sigma}_{*,t} \cdot {S},
\qquad M_{*,t+1}=\Sigma_{*,t+1} \cdot {S}.
\end{equation}
EB-PCA uses the posterior mean functions defined instead by NPMLEs
of $\bar{\pi}_*$ and $\pi_*$, together with the empirical estimates
$\bar{\Sigma}_t=n^{-1}(\bV^t)^\top \bV^t$,
$\bar{M}_t=\bar{\Sigma}_t\hat{{S}}$,
$\Sigma_{t+1}=n^{-1}(\bU^t)^\top \bU^t$, and
$M_{t+1}=\Sigma_{t+1}\hat{{S}}$
where $\hat{{S}}=\diag(\hat{s}_1,\ldots,\hat{s}_k)$ is the estimate of $S$
from \eqref{eq:shat}. (For $\bar{\Sigma}_t$, we have applied $\gamma/d=1/n$.)
These empirical estimates avoid the need to perform
Gaussian integrations to analytically evaluate the expectations that define
the true matrices $\bar{M}_{*,t},M_{*,t},\bar{\Sigma}_{*,t},\Sigma_{*,t}$.
EB-PCA is initialized at the right sample PCs $\bG^0=\bG$ and the plug-in
estimates $(M_0,\Sigma_0) \equiv (M,\Sigma)$ from the preceding section.
In particular, $\bV^0$ is the initial empirical Bayes estimate for $\bV$ based
on $\bG$ as previously described.

We summarize the full EB-PCA method as
Algorithm \ref{alg:EBPCA}.

\begin{remark}\label{remark:iterNPMLE}
In Lines 6 and 10 of Algorithm \ref{alg:EBPCA}, we form new NPMLEs
for $\bar{\pi}_*$ and $\pi_*$ in each iteration. This
allows for the possibility of improving these estimates as the signal-to-noise
ratios reflected by the state parameters $(M_t,\Sigma_t)$ and
$(\bar{M}_t,\bar{\Sigma}_t)$ improve across iterations.
In data examples with strong signal, this re-estimation of $\bar{\pi}_*$
and $\pi_*$ may be unneeded, and removed to improve computational
efficiency. See Appendix~\ref{subsec:practical} for further discussion.
\end{remark}

\begin{remark}\label{remark:initialization}
For the state evolution to correctly describe the AMP iterates under
the PCA initialization $\bG^0=\bG$, the first AMP iteration for $\bF^0$ in Line
8 should use $\bU^{-1}=\bF \cdot {\Sigma_0^{1/2}}$ as initialized in Line 4,
rather than $\bU^{-1}=0$ as described in~\cite{montanari2017estimation}. We
elaborate on this in Appendix~\ref{appendix:oracleSE}.
\end{remark}

\begin{algorithm}[t]
\begin{algorithmic}[1]
\Require Data matrix $\bY \in \bbR^{n \times d}$, normalized as in Remark
\ref{remark:scaling} to have average entrywise noise variance $1/n$. Number of
PCs $k$, number of AMP iterations $T$, prior class $\cP$.
\Statex {\it // Initialization}
\State Let $\gamma=d/n$. Let $(\sqrt{\gamma} \lambda_1,\ldots,
\sqrt{\gamma}\lambda_k)$, $\bF=(\bff_1,\ldots,\bff_k)$, and
$\bG=(\bg_1,\ldots,\bg_k)$ be the top $k$ singular values and singular vectors
of $\bY$, with $\|\bff_i\|_2=\sqrt{n}$ and $\|\bg_i\|_2=\sqrt{d}$.
\State Define $\hat{s}_i^2$ by \eqref{eq:shat}. Set
$\sigma_i^2=(1+\gamma \hat{s}_i^2)/(\gamma \hat{s}_i^4
+\gamma \hat{s}_i^2)$, $\mu_i^2=1-\sigma_i^2$,
$M_0=\diag(\mu_1,\ldots,\mu_k)$,
$\Sigma_0=\diag(\sigma_1^2,\ldots,\sigma_k^2)$,
$\hat{{S}}=\diag(\hat{s}_1,\ldots,\hat{s}_k)$.
\State $\bG^0 \gets \bG$ and $\bU^{-1} \gets \bF \cdot \Sigma_0^{1/2}$
\Statex {\it // Iterative refinement}
\For{$t=0,1,2,\ldots,T$}
\Statex \hspace{0.15in} {\it // Denoise left PCs}
\State $\pi_t \gets \MLE(\bG^t \mid M_t,\Sigma_t,\cP)$
\State $\bV^t \gets \theta(\bG^t \mid M_t,\Sigma_t,\pi_t)$
\State $\bF^t \gets \bY\bV^t-\bU^{t-1} \cdot \gamma \langle
\der \theta(\bG^t \mid M_t,\Sigma_t,\pi_t) \rangle^\top$
\State $\bar{\Sigma}_t \gets (\bV^t)^\top \bV^t/n$
and $\bar{M}_t \gets \bar{\Sigma}_t \cdot \hat{{S}}$
\Statex \hspace{0.15in} {\it // Denoise right PCs}
\State $\bar{\pi}_t \gets \MLE(\bF^t \mid \bar{M}_t,\bar{\Sigma}_t,\cP)$
\State $\bU^t \gets \theta(\bF^t \mid \bar{M}_t,\bar{\Sigma}_t,\bar{\pi}_t)$
\State $\bG^{t+1} \gets \bY^\top \bU^t-\bV^t \cdot \langle \der \theta(\bF^t \mid
\bar{M}_t,\bar{\Sigma}_t,\bar{\pi}_t) \rangle^\top$
\State $\Sigma_{t+1} \gets (\bU^t)^\top\bU^t/n$ and
$M_{t+1} \gets \Sigma_{t+1} \cdot \hat{{S}}$
\EndFor
\Ensure Final estimates $(\hat{\bU},\hat{S},\hat{\bV})
=(\bU^{T},\hat{S},\bV^{T})$
\end{algorithmic}
\caption{EB-PCA}\label{alg:EBPCA}
\end{algorithm}

\subsection{Relation to naive mean field variational Bayes}\label{sec:EBMF}

\cite{wang2018empirical} propose an empirical Bayes
matrix factorization (EBMF) algorithm similar to EB-PCA, based instead on
naive mean-field variational Bayes: in the
rank-one model of \eqref{eq:rankonemodel}, this approximates the
posterior law $p(\bu,\bv|\bY)$ by a factorized form
$\bar{q}(\bu)q(\bv)=\prod_{i=1}^n \bar{q}_i(u_i) \prod_{j=1}^d q_j(v_j)$.
The distributions $\bar{q}_i,q_j$ are chosen to minimize
the Kullback-Leibler divergence $\DKL(\bar{q}(\bu)q(\bv)
\|p(\bu,\bv|\bY))$ or, equivalently, to maximize the evidence lower bound
\[\mathcal{F}(\bar{q},q)
=\bbE_{\bu \sim \bar{q},\bv \sim q}\big[\log p(\bY,\bu,\bv)-
\log \bar{q}(\bu)q(\bv) \big].\]
Here, the joint density $p(\bY,\bu,\bv)$ depends on the priors
$(\bar{\pi},\pi)$ for $(\bu,\bv)$, and EBMF estimates
these by maximizing $\mathcal{F}$ jointly over $(\bar{q},\bar{\pi},q,\pi)$.
This maximization is performed via the iterative coordinate ascent
variational inference (CAVI) updates
\[(q_t,\pi_t)=\argmax_{q,\pi}
\mathcal{F}(\bar{q}_{t-1},\bar{\pi}_{t-1},q,\pi),\qquad
(\bar{q}_t,\bar{\pi}_t)=\argmax_{\bar{q},\bar{\pi}}
\mathcal{F}(\bar{q},\bar{\pi},q_t,\pi_t)\]

As shown in \cite{wang2018empirical}, the CAVI updates admit a simple form in
terms of the quantities
\[\bv^t \equiv \bbE_{\bv \sim q_t}[\bv], \quad
\bar{\sigma}_t^2 \equiv n^{-1}\bbE_{\bv \sim q_t}[\|\bv\|^2],
\quad \bu^t \equiv \bbE_{\bu \sim \bar{q}_t}[\bu], \quad
\sigma_{t+1}^2 \equiv n^{-1}\bbE_{\bu \sim \bar{q}_t}[\|\bu\|^2].\]
which is very similar to the iterations of Algorithm \ref{alg:EBPCA}
but does not incorporate the AMP Onsager correction terms.
The Onsager terms correct for weak dependences in the true
posterior distributions of $u_1,\ldots,u_n$ and of $v_1,\ldots,v_d$ in high
dimensions. These dependences are omitted in the naive mean field
approximation, and this is discussed further in
\cite{ghorbani2019instability,fan2018tap}. The differences between these
approaches vanish in the limit of infinite signal-to-noise ratio
$s \to \infty$, but are non-negligible in any bounded signal-to-noise setting,
and become more pronounced for weak signals near 
the phase transition threshold $s_*(\gamma)=\gamma^{-1/4}$. 

We believe there are two particular appeals of applying AMP over the naive mean
field approximation in this specific application: first, in the limit
$s \to \infty$, the sample PCs become increasingly accurate, and there is less
to gain from an empirical Bayes approach. It is precisely in settings of weaker 
signals that empirical Bayes may yield the largest improvements over PCA.
Second, even if the model of \eqref{eq:model} is correctly specified, the
connection between the CAVI iterates $\bff^t,\bg^t$ and the Gaussian
compound decision model is inexact---see Figures \ref{fig:iteratesuniform}
and \ref{fig:iteratestwopoint}---whereas the AMP iterates $\bff^t,\bg^t$
are exactly described by the Gaussian models as $n,d \to \infty$.
This provides a stronger justification for applying empirical Bayes procedures
based on these Gaussian models to the iterates of AMP. However, we note that
EBMF is developed in a more general model of heteroscedastic noise
$w_{ij} \sim \cN(0,1/\tau_{ij})$, whereas our current derivation and analysis
of EB-PCA are limited to a setting of uniform noise variance.

There are a few other distinctions in perspective between
EBMF and our work: for the rank-$k$ model,
\cite{wang2018empirical} propose a CAVI scheme that iteratively updates each
rank-one component separately, and hence does not learn a joint prior for the
multivariate distribution of several PCs. For each rank-one component, there is
a stronger emphasis in \cite{wang2018empirical} on sparsity-inducing priors
that are unimodal at 0, connecting the approach more to sparse PCA.
Depending on the characteristics of the data at hand, we believe that a
multivariate approach of learning a fully nonparametric joint prior
for several PCs has the potential of yielding improved accuracy. 

\section{Simulated examples}\label{sec:simulations}

\subsection{Univariate priors}

\begin{figure}[t]
\centering
\minipage{0.49\columnwidth}
\xincludegraphics[width=0.49\textwidth,label=(a)]{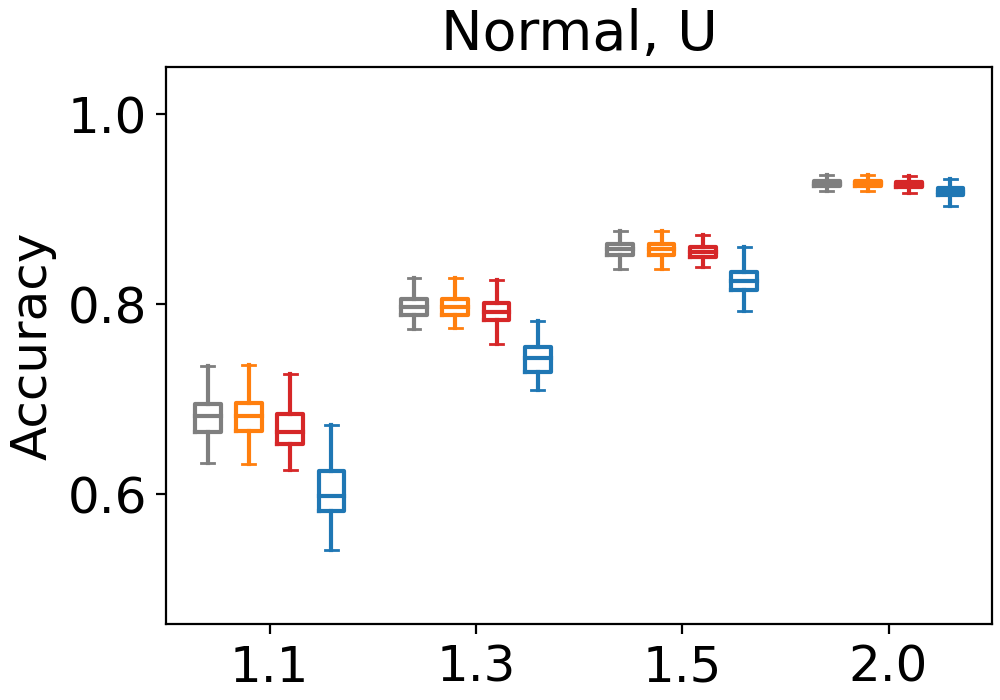}
\xincludegraphics[width=0.49\textwidth]{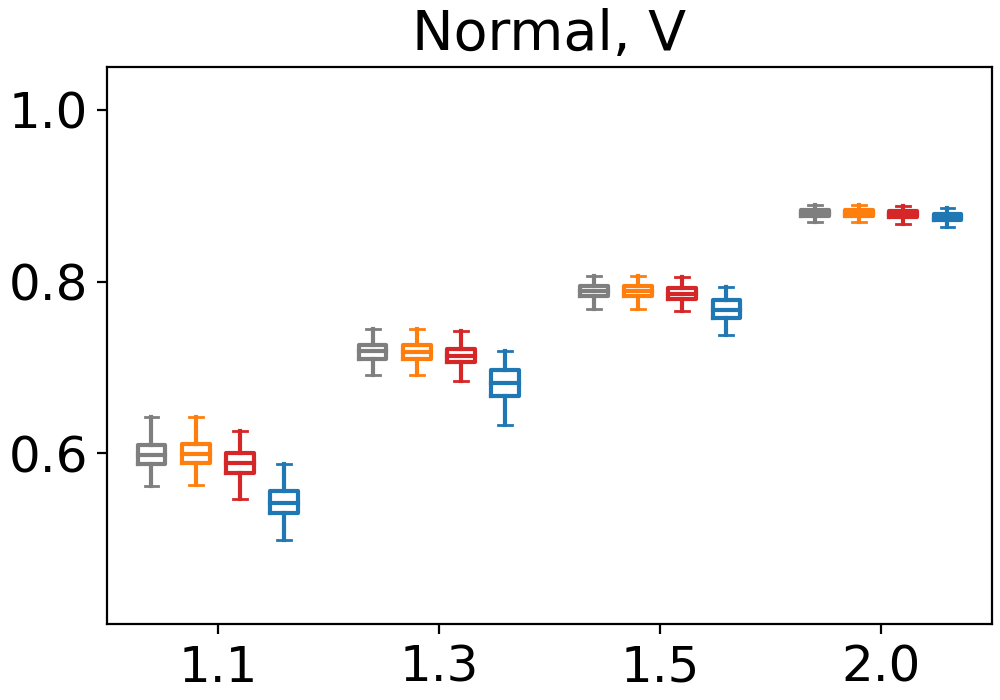}
\endminipage
\minipage{0.49\columnwidth}
\xincludegraphics[width=0.49\textwidth,label=(b)]{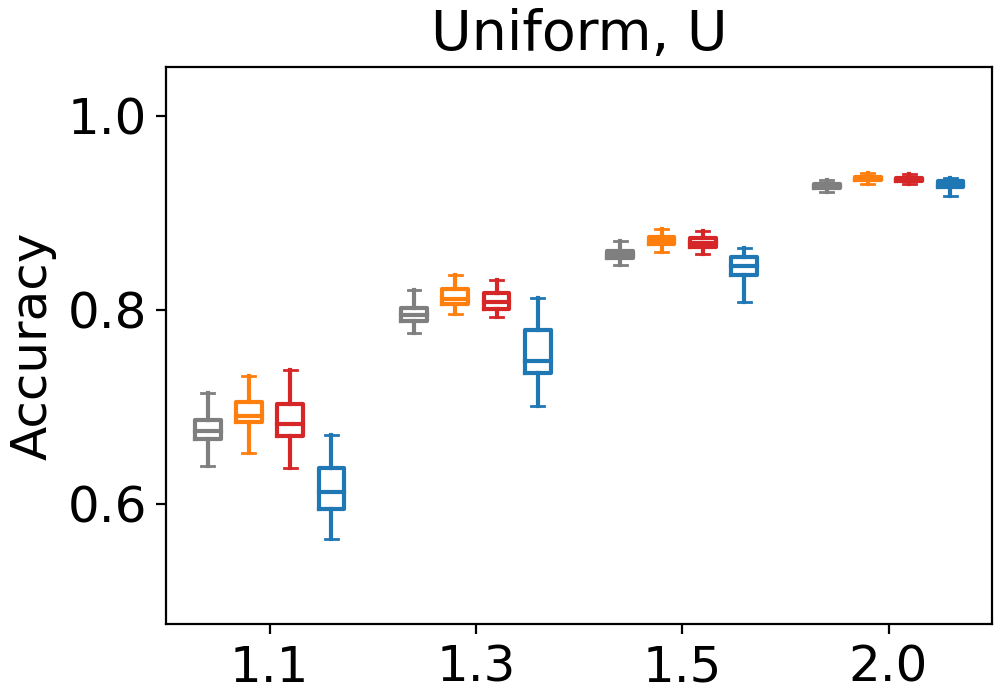}
\xincludegraphics[width=0.49\textwidth]{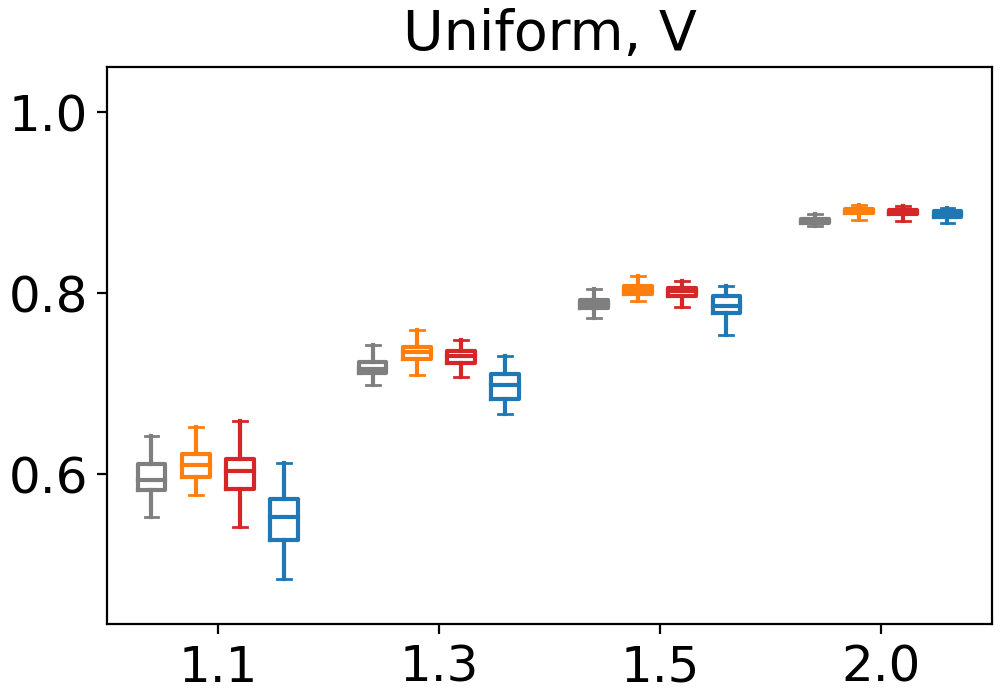}
\endminipage\\
\minipage{0.49\columnwidth}
\xincludegraphics[width=0.49\textwidth,label=(c)]{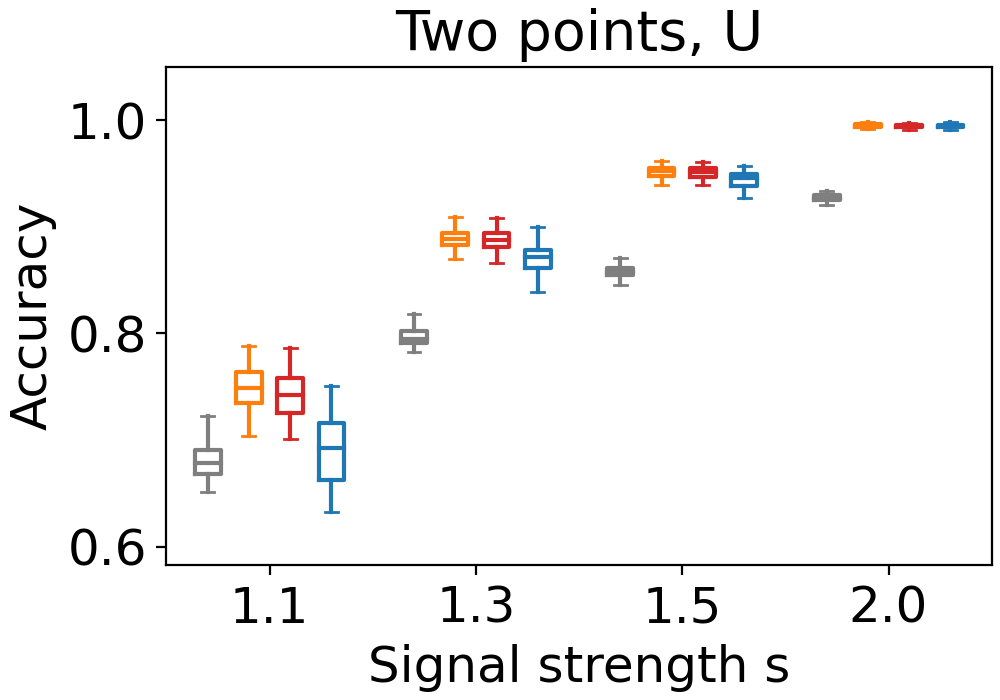}
\xincludegraphics[width=0.49\textwidth]{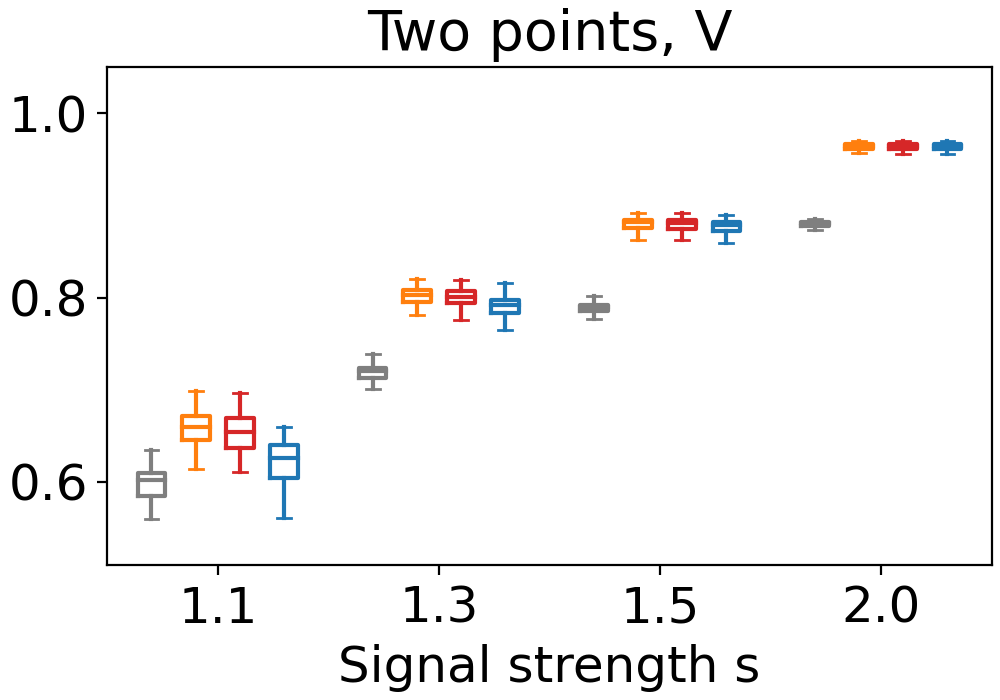}
\endminipage
\minipage{0.49\columnwidth}
\xincludegraphics[width=0.49\textwidth,label=(d)]{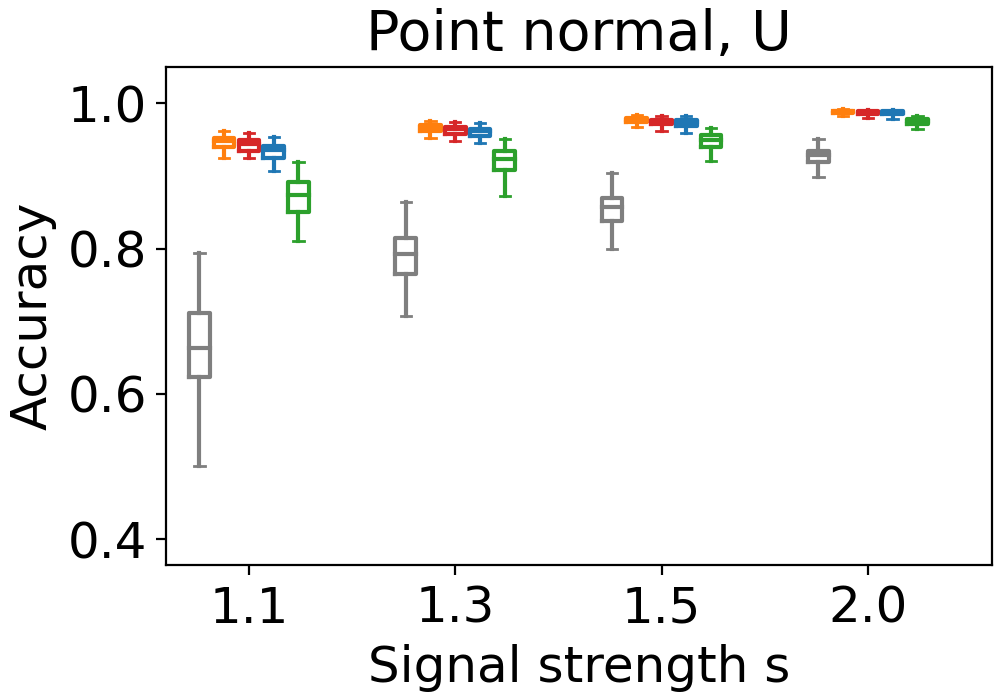}
\xincludegraphics[width=0.49\textwidth]{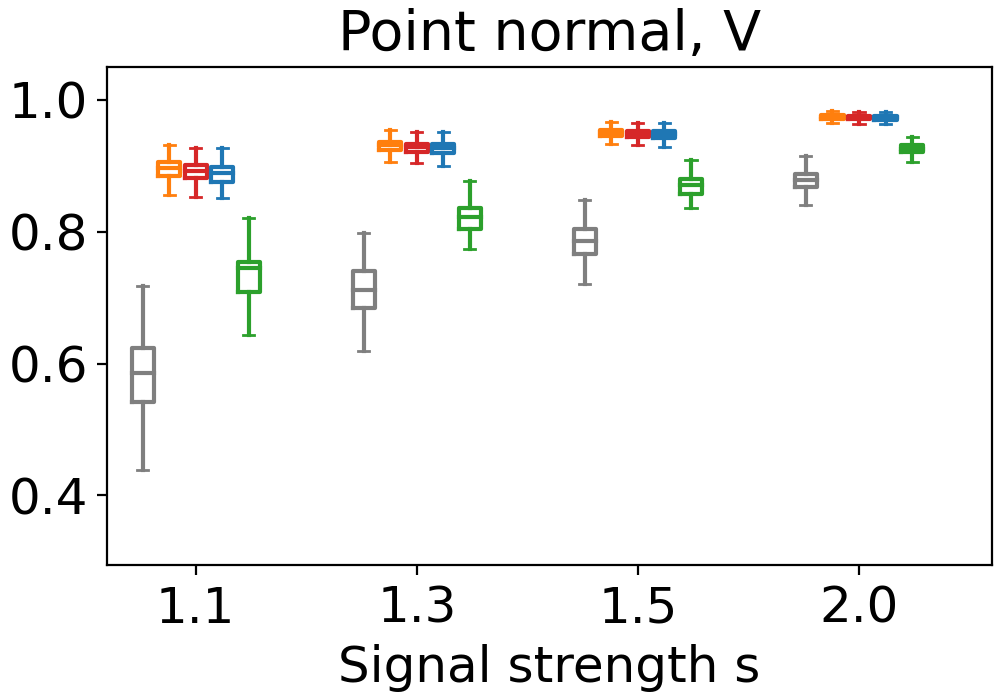}
\endminipage\\
\minipage{\columnwidth}
\includegraphics[width=\textwidth]{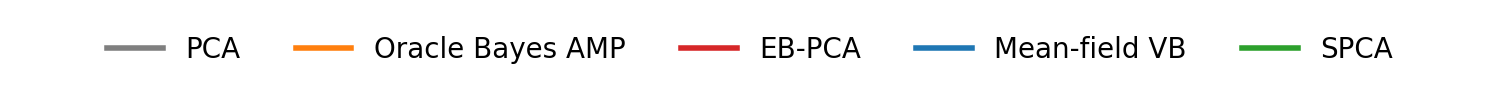}
\endminipage
\caption{Estimation accuracy of PCA, oracle Bayes
AMP, EB-PCA, and naive mean-field VB, across 50 simulations of a rank-one model
with four priors.
For (d), comparison
with \texttt{spca}
is also shown.
In all cases, EB-PCA nearly matches the accuracy of the oracle Bayes AMP
procedure and improves over mean-field VB. EB-PCA also
improves over standard PCA for the non-Gaussian priors (b--d), and over
\texttt{spca} for the sparse prior (d).
}\label{fig:rankoneaccuracy}
\end{figure}

We compare EB-PCA with several other methods on simulated data from the
rank-one model of \eqref{eq:rankonemodel}, using four different univariate priors:
(a) standard Gaussian $\cN(0,1)$,
(b) $\operatorname{Uniform}[-\sqrt{3},\sqrt{3}]$,
(c) $\operatorname{Bernoulli}\{+1,-1\}$, and
(d) sparse point-normal $0.9\,\delta_0+0.1\,\cN(0,10)$.
For simplicity, we use the same prior for both $\bu$ and $\bv$.
The two-point and point-normal
priors (c--d) represent simple clustering and sparse-PCA applications.

We compare EB-PCA with standard PCA, the oracle Bayes AMP procedure of
\cite{montanari2017estimation} that knows the true priors, and
naive mean-field variational Bayes with priors estimated by NPMLE, 
corresponding to a version of EBMF in \cite{wang2018empirical} with
known homoscedastic noise variance.
For the sparse point-normal prior (d), we compare also to
the \texttt{spca} method of~\cite{zou2006sparse}. Simulation details
can be found in Appendix~\ref{sec:data}.

Figure \ref{fig:rankoneaccuracy} displays the accuracy of these procedures, in
terms of the alignments \linebreak
${\langle \hat{\bu}, \bu \rangle}/{\| \bu\|
\|\hat{\bu}\|}$ and ${\langle \hat{\bv}, \bv \rangle}/{\|\bv\| \|\hat{\bv}\|}$.
These alignments measure the accuracy of the estimated PC directions, and do
not account for further improvements of EB-PCA/Bayes-AMP resulting from
the shrinkage of $(\bu,\bv)$. The dimensions tested are
$(n,d)=(2000,4000)$, so that the estimate of $\bu \in \bbR^n$ is more accurate
than that of $\bv \in \bbR^d$. The phase-transition point for super-critical
signal strength is $s_*(\gamma)=0.84$, and we tested the range of signal
strengths $s \in \{1.1, 1.3, 1.5, 2.0\}$.

Under the standard Gaussian prior (a), the posterior mean function in
(\ref{eq:multivariatedenoise}) for
oracle Bayes AMP is linear in each iteration. Thus its iterates remain
proportional to the sample PCs, and the asymptotic accuracy of both EB-PCA and 
Bayes-AMP in the above alignment metric coincide with
standard PCA. Thus the Gaussian prior (a) provides a ``control'' setting in
which we hope to match the performance of PCA.

The following observations summarize these comparisons:

\begin{itemize}
\item \textbf{Oracle Bayes AMP.} EB-PCA is nearly as
accurate as oracle Bayes AMP in all cases, without knowing the true priors.
There is a small decrease in accuracy of EB-PCA for very weak signals, due
to estimation variance for $s$.
\item \textbf{Standard PCA.} In the control setting of the
$\cN(0,1)$ prior, EB-PCA yields accuracy comparable to PCA. EB-PCA yields
improved accuracy in all remaining settings, with this improvement being
more substantial for weaker signals and
for the two-point and sparse point-normal priors that reflect
stronger prior structure.
\item \textbf{SPCA.} EB-PCA and mean-field VB both improve significantly
over \texttt{spca} for the sparse point-normal prior, with the added advantage
of being tuning-parameter free.
\item \textbf{Mean-field VB.} EB-PCA improves over mean-field VB in all
cases. The improvement is small for the point-normal prior and for larger signal
strengths, but is larger in the remaining settings. Mean-field VB
seems to yield worse estimation accuracy than standard PCA for
the $\cN(0,1)$ and continuous uniform priors.
\end{itemize}

To provide a more detailed comparison of EB-PCA with mean-field VB at a weak
signal strength, panels (a) and (b) of Figures~\ref{fig:iteratesuniform} and
\ref{fig:iteratestwopoint}
display the entrywise distributions of several
iterates for the uniform and two-point priors with $s=1.3$. Overlaid are the
convolution densities of the true prior with the estimated levels of Gaussian
noise. Panel (c) displays the accuracies across iterations.
The overlaid convolution densities in EB-PCA closely match the empirical
distributions of the iterates,
whereas discrepancies accumulate for mean-field VB. These discrepancies can
cause mean-field VB to estimate an increasingly incorrect prior, and to have
decreasing accuracy across iterations.

\begin{figure}
\minipage{0.33\columnwidth}
\xincludegraphics[width=0.9\textwidth,label=(a)]{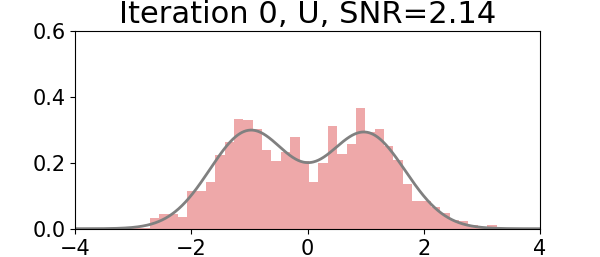}\\\\
\xincludegraphics[width=0.9\textwidth]{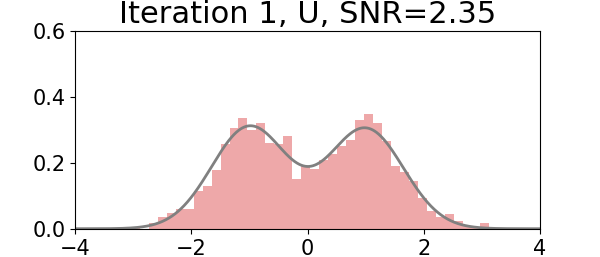}\\\\
\xincludegraphics[width=0.9\textwidth]{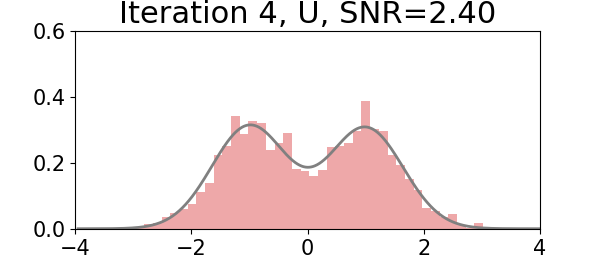}
\endminipage%
\minipage{0.33\columnwidth}
\xincludegraphics[width=0.9\textwidth,label=(b)]{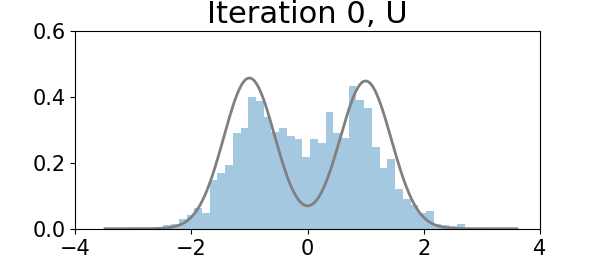}\\\\
\xincludegraphics[width=0.9\textwidth]{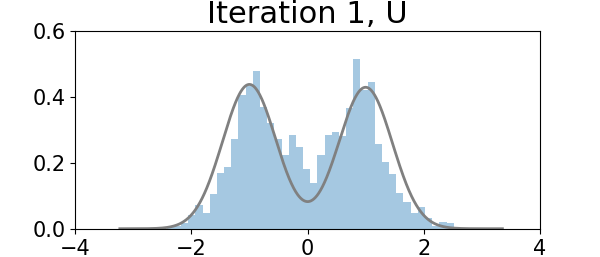}\\\\
\xincludegraphics[width=0.9\textwidth]{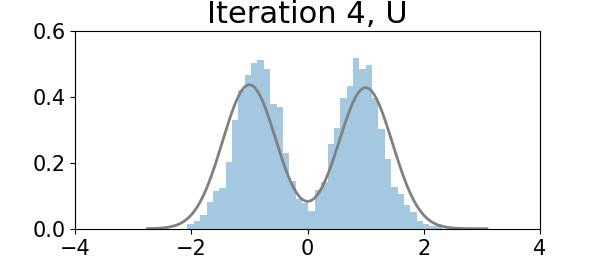}
\endminipage%
\minipage{0.33\columnwidth}
\xincludegraphics[width=0.9\textwidth,label=(c)]{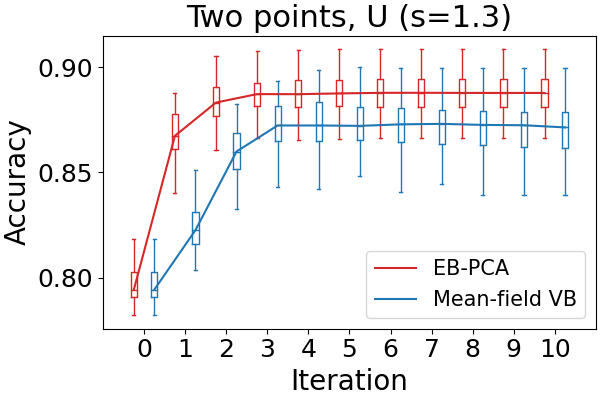}
\endminipage
\caption{Iterates of EB-PCA versus mean-field VB with
$\operatorname{Uniform}[-\sqrt{3},\sqrt{3}]$ priors for $\bu$. (a) Distribution of entries of the EB-PCA
iterates $\bff_0, \bff_1, \bff_4$. A close agreement is observed with the
overlaid convolution densities
$\operatorname{Uniform}[-\sqrt{3},\sqrt{3}] * \cN(0,\bar{\sigma}_t^2/\bar{\mu}_t^2)$ for $t=0,1,4$, where $\bar{\mu}_t^2/\bar{\sigma}_t^2$ is the estimated signal-to-noise
ratio (SNR) used to perform empirical Bayes denoising. (b) Analogous picture
for mean-field VB, with $\bar{\mu}_t^2/\bar{\sigma}_t^2$ as estimated in
mean-field VB. A discrepancy is observed between the iterates
and the convolution densities.
(c) Accuracy of $\bu_t$ across iterations,
for EB-PCA versus mean-field VB.
Results for $(\bv,\bg_t)$ are similar and omitted for brevity.}\label{fig:iteratesuniform}
\end{figure}

\subsection{Bivariate priors}\label{sec:bivariate}
We now demonstrate that multivariate EB-PCA, which estimates a joint prior over
several PCs, can improve substantially over EB-PCA applied marginally to each
PC when there is strong multivariate structure.
We consider dimensions $(n,d)=(1000,1000)$ and signal strengths $(s_1,s_2)=(4,2)$, for
the two bivariate priors presented in Figure~\ref{fig:bivariate}(a): a discrete
three-point prior and a uniform prior on a circle.

\begin{figure}
\minipage{0.25\columnwidth}
\xincludegraphics[width=0.8\textwidth,label=(a)]{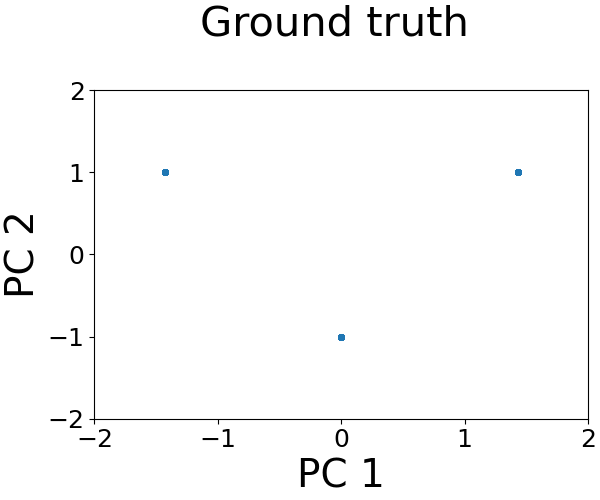}\\
\vspace{0.5pt}
\xincludegraphics[width=0.8\textwidth]{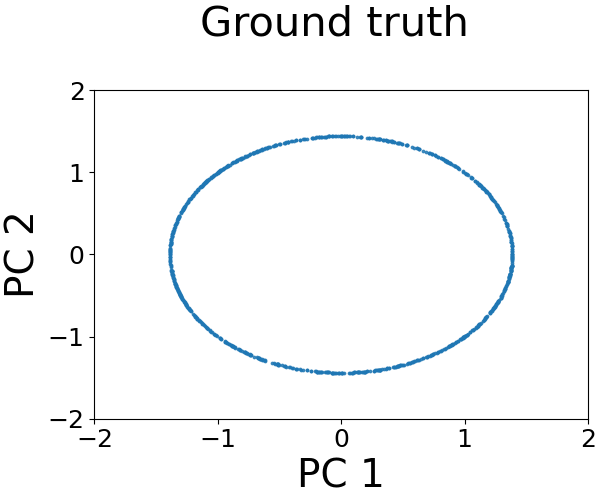}%
\endminipage
\minipage{0.25\columnwidth}
\xincludegraphics[width=0.8\textwidth,label=(b)]{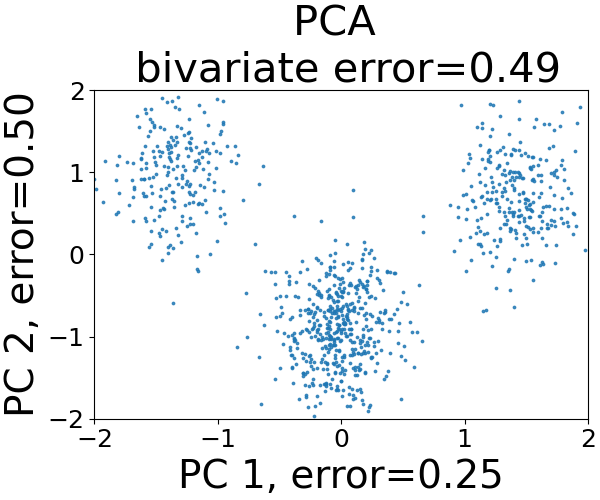}\\ 
\vspace{0.5pt}
\xincludegraphics[width=0.8\textwidth]{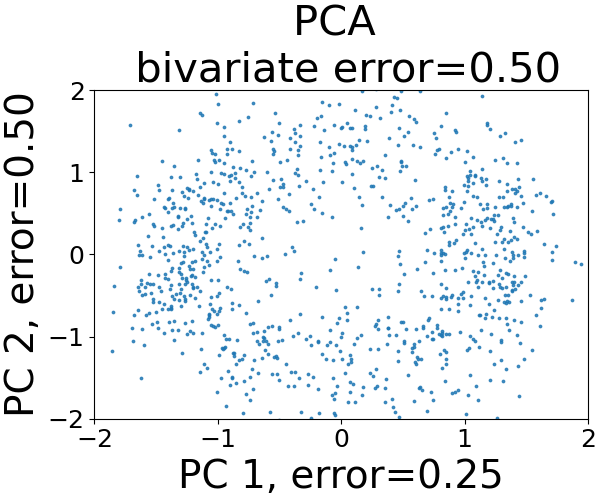}%
\endminipage
\minipage{0.25\columnwidth}
\xincludegraphics[width=0.8\textwidth,label=(c)]{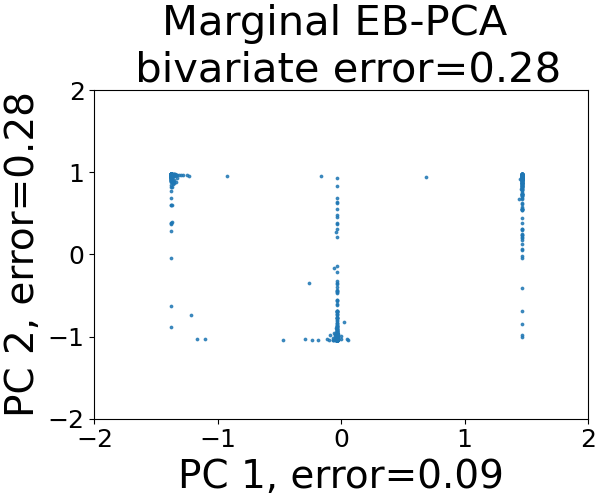}\\
\vspace{0.5pt}
\xincludegraphics[width=0.8\textwidth]{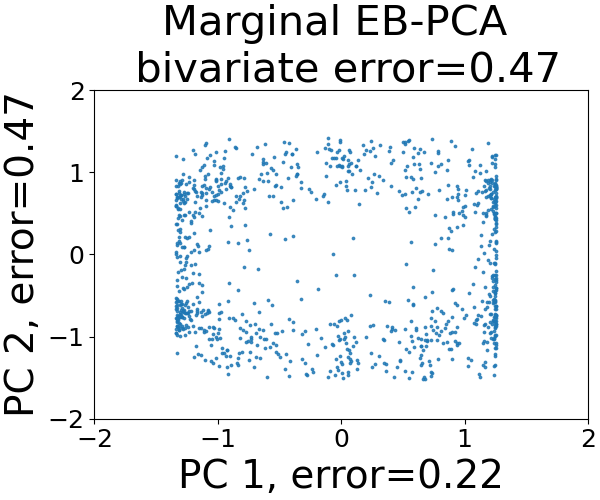}%
\endminipage
\minipage{0.25\columnwidth}
\xincludegraphics[width=0.8\textwidth,label=(d)]{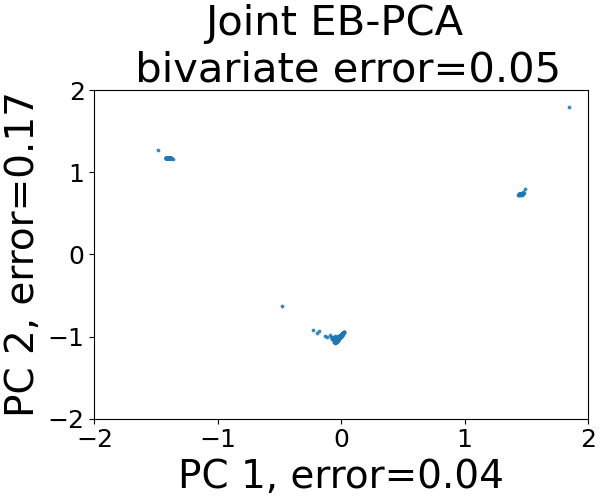}%
\\
\vspace{0.5pt}
\xincludegraphics[width=0.8\textwidth]{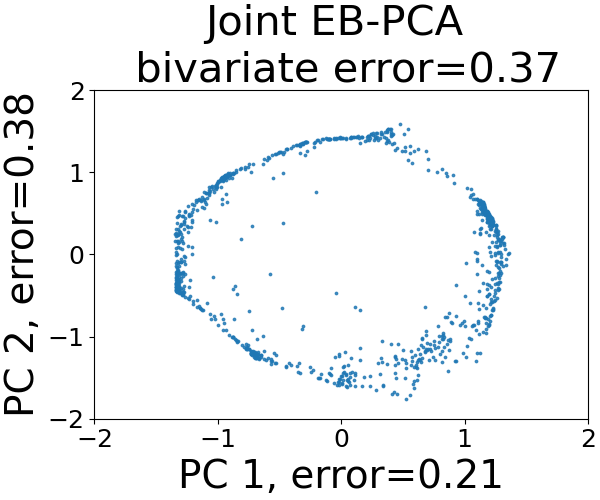}
\endminipage
\caption{Comparison of standard PCA, EB-PCA applied marginally, and EB-PCA
applied jointly
for two bivariate priors: a discrete three-point mixture
prior (top row) and a uniform prior on a circle (bottom row). Displayed
are scatter-plots of the columns of $\bU$ for (a) the true PCs, (b) the sample PCs,
(c) EB-PCA applied marginally for each PC, and (d)
EB-PCA applied jointly to learn the bivariate prior. Estimation errors
are displayed in the figure titles.
Marginal EB-PCA improves over standard PCA,
and joint EB-PCA further improves over marginal EB-PCA.
}\label{fig:bivariate}
\end{figure}

Titles in Figure~\ref{fig:bivariate} display the estimation error
 as defined by the subspace distance (see Section 2.5 of \cite{van1983matrix}) 
between the column spans of
$\hat{\bU}$ and $\bU$ and also between the spans of the individual PCs. 
Table~\ref{tab:bivariate} reports the average of such estimation errors across
50 simulations. These results indicate that by leveraging the
underlying joint structure, multivariate EB-PCA learns a more accurate prior
and has lower error
both for the estimated two-dimensional subspace and for the individual PCs.
The simultaneous estimation of $k$ PCs has the additional benefit of reducing
the computation time by a factor of roughly $k$ over the univariate approach.

\section{Applications}\label{sec:examples}

We illustrate EB-PCA on three high-dimensional genetics datasets:
genotype data from the 1000 Genomes Project and 
the third phase of the International HapMap Project (HapMap3)~\cite{international2010integrating}, 
and single cell
RNA-seq (scRNA-seq) gene expression data on Peripheral Blood Mononuclear 
Cells (PBMC) from 10X Genomics. Preprocessing procedures and
implementation details are provided in Appendix~\ref{sec:data}.

\subsection{1000 Genomes Project genotypes}\label{sec:1000genomes}

PCA is commonly-used to correct for population
stratification in genome-wide association studies. As the number of
SNPs $n$ often far exceeds the number of individuals $d$, the
estimated PCs in $\bbR^d$ suffer minimally from high-dimensional noise for
the leading PCs.
This provides a ground truth by which we may quantitatively compare estimation
accuracy on subsampled data.

\begin{table}
  \caption{PC estimation errors on subsampled genotype matrices from 1000 Genomes\label{tab:1000g}}
  \centering
  \begin{tabular}[c]{c|c|ccccc}
      \toprule
      \# SNPs& Error & PC1 & PC2 & PC3 & PC4 & Joint \\ \hline\hline
      \multirow{2}{*}{100} & PCA &.35(.039) & .49(.036) & .75(.078) & .83(.059) & .79(.036) \\ 
      & EB-PCA &.20(.058) & .33(.042) & .60(.13) & .62(.12) & .56(.036) \\ \hline
      \multirow{2}{*}{1000} & PCA  & .11(.0073) & .17(.0070) & .31(.024) & .36(.019) & .34(.0091) \\ 
      & EB-PCA & .072(.0099) & .10(.0098) & .17(.042) & .20(.034) & .18(.0055) \\\hline
      \multirow{2}{*}{10000} & PCA &.034(.0016) & .051(.0017) & .10(.0089) & .11(.0086) & .11(.0018) \\ 
      & EB-PCA &.028(.0023) & .040(.0032) & .082(.013) & .091(.011) & .081(.0017) \\ 
      \bottomrule
  \end{tabular}
\end{table}

Implementing this experiment for genotype data from the 1000 Genomes Project,
we extracted 100{,}000 common SNPs for 2504 individuals, from which we computed the ground truth PCs.
We then estimated PCs on subsamples of 100, 1000, or 10{,}000 randomly
selected SNPs using both EB-PCA and PCA. Figure~\ref{fig:singvals}
plots the singular values for a typical
subsample of 1000 SNPs, in which 4 clear outlier values are apparent.
Thus we chose to estimate the leading 4 PCs using EB-PCA.

Table \ref{tab:1000g} compares errors across
50 random subsamples of each size, where the error is the subspace distance
(as previously used in Section~\ref{sec:bivariate}) against the ground
truth PCs. The results show a clear improvement of EB-PCA
over PCA, with slightly larger improvement 
for the lower PCs having smaller singular values.

A visual comparison for 1000 subsampled SNPs was presented
previously in Figure~\ref{fig:1000genomesintro}. The EB-PCA estimates are closer
to the ground truth than the sample PCs, and better separate the
subjects by ethnicity. For example, the Caucasian, African, and
East Asian populations are mixed in the two-dimensional plot of the 3rd vs.\
4th sample PC, whereas they are separated in the EB-PCA estimates
and also in the ground truth. The information used by EB-PCA for 
performing this separation in PCs 3 and 4 comes from learning a joint prior with
PCs 1 and 2, where these populations have clear separation.

\subsection{HapMap3 genotypes}\label{sec:Hapmap3}

We performed a similar experiment on genotype data from HapMap3.
For the 1397 individuals in HapMap3, we computed ground truth PCs from
142{,}185 common SNPs.
Based on the singular value distribution in Figure~\ref{fig:singvals} on
a subsample of 5000 SNPs, we chose to estimate the leading 4 PCs using EB-PCA.

Table~\ref{tab:hapmap3} compares the errors of EB-PCA and PCA across 50 subsampled
data sets of 1000, 5000, and 10{,}000 SNPs, and Figure~\ref{fig:hapmap3} depicts
results for 5000 SNPs. We again observe a consistent decrease in estimation
error for EB-PCA, which is larger for the weaker PCs.

\begin{figure}
\minipage{0.33\columnwidth}
\xincludegraphics[width=0.8\textwidth,label=(a)]{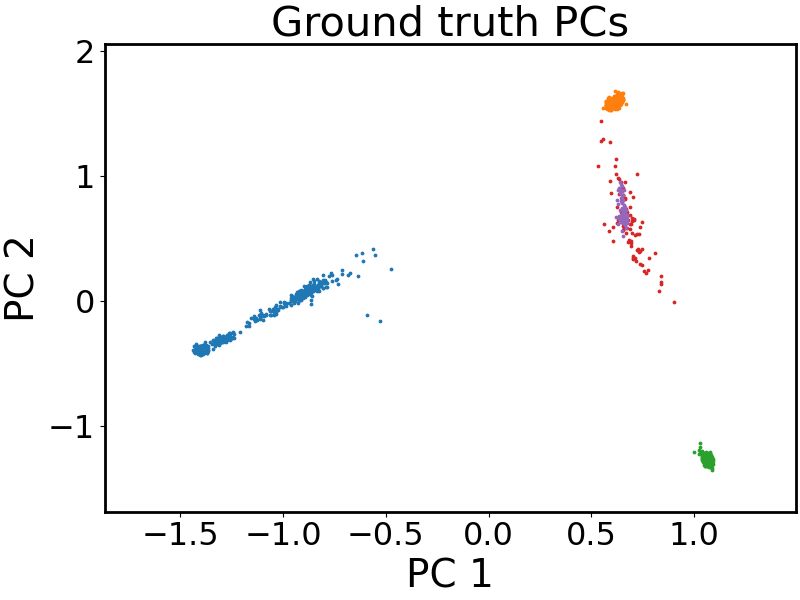}\\
\xincludegraphics[width=0.8\textwidth]{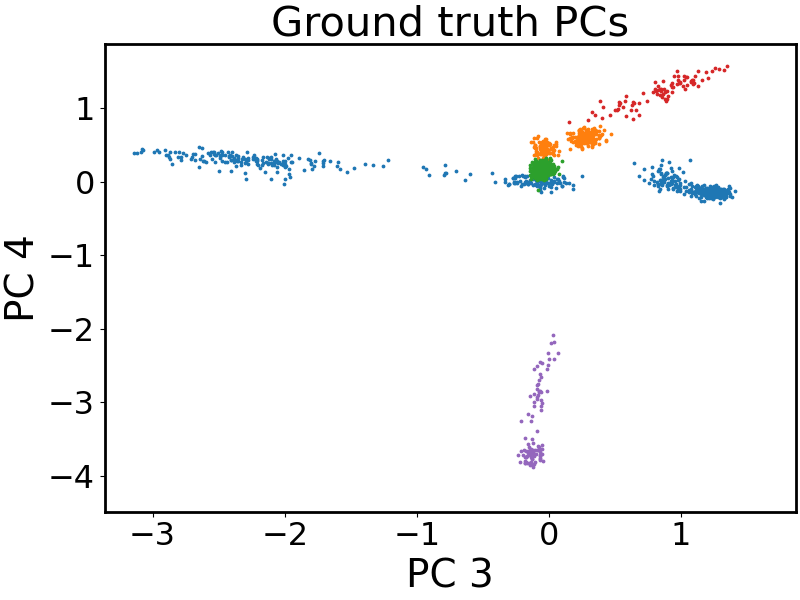}\endminipage
\minipage{0.33\columnwidth}
\xincludegraphics[width=0.8\textwidth,label=(b)]{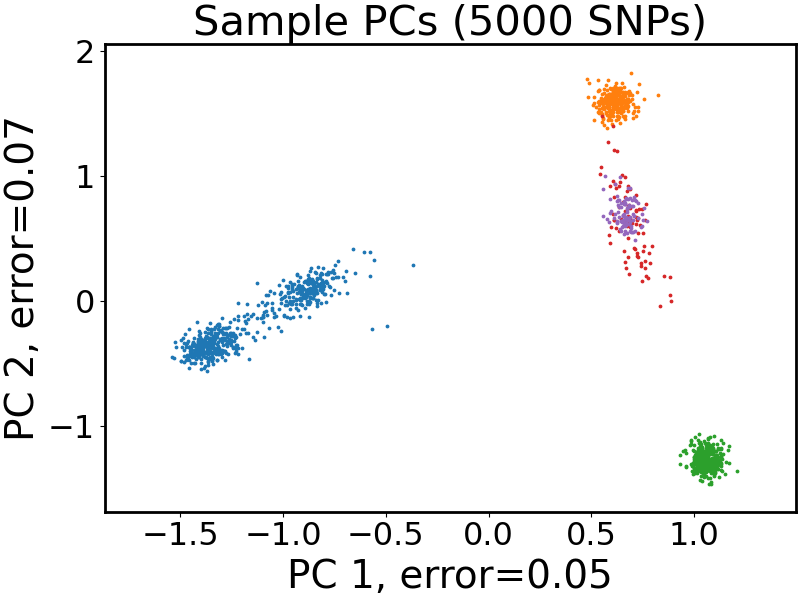}\\
\xincludegraphics[width=0.8\textwidth]{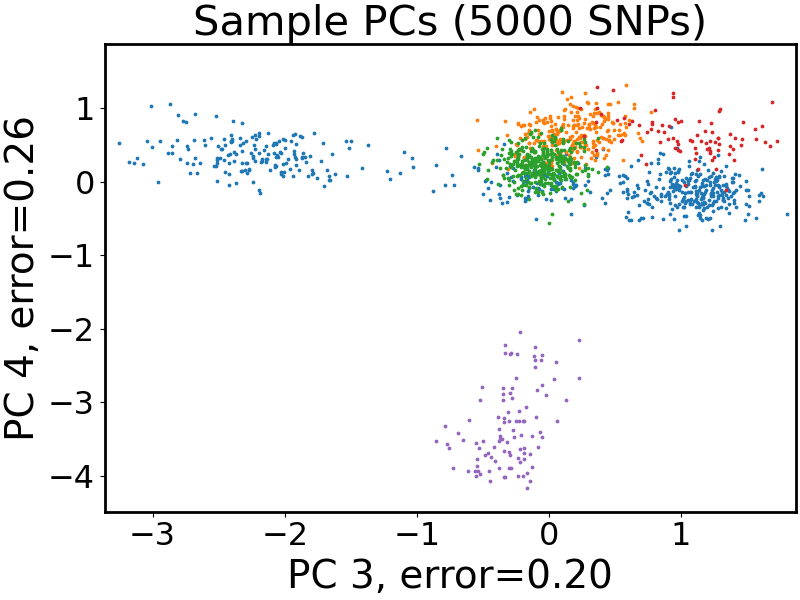}\endminipage
\minipage{0.33\columnwidth}
\xincludegraphics[width=0.8\textwidth,label=(c)]{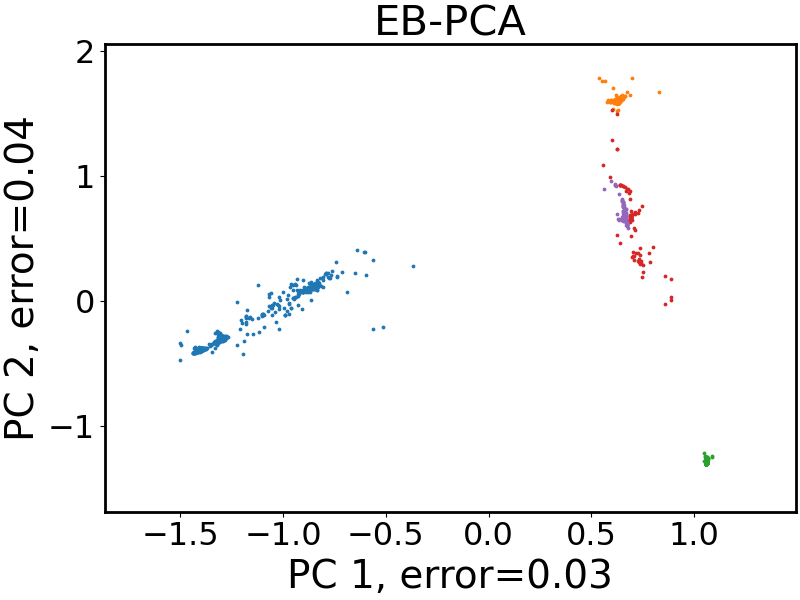}\\
\xincludegraphics[width=0.8\textwidth]{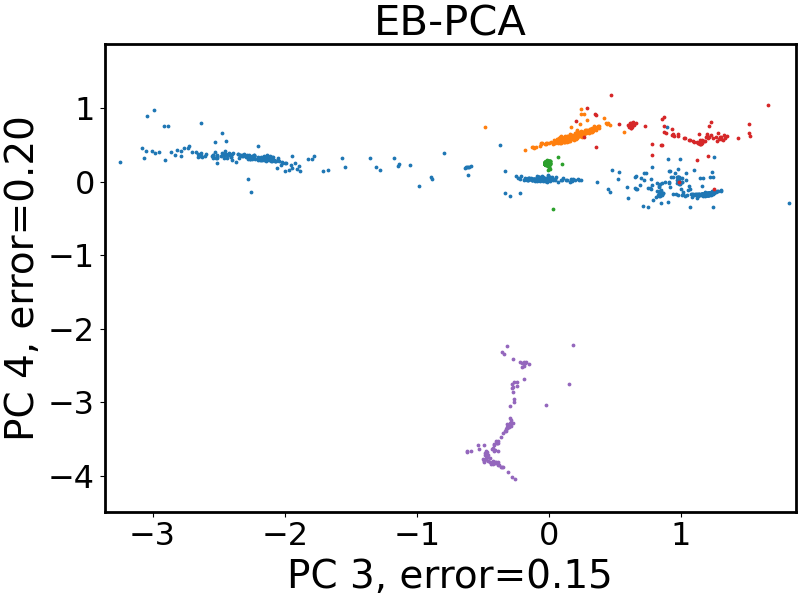}\endminipage\\
\minipage{\columnwidth}
\includegraphics[width=\textwidth]{Fig/1000G/1000G_legend_horizontal.png}
\endminipage
\caption{Illustration of EB-PCA on genotype data from HapMap3, similar to
Figure \ref{fig:1000genomesintro}. (a) Ground truth PCs,
with points colored by the individuals' ethnicity.
(b) Sample PCs computed on a subsample of 5000 SNPs.
(c) EB-PCA estimates of the four PCs, using the same subsample as in panel (b).
}\label{fig:hapmap3}
\end{figure}

\subsection{10x Genomics PBMC single-cell RNA-seq}\label{sec:10x}

\begin{figure}
\centering
\minipage{0.48\columnwidth}
\xincludegraphics[width=0.48\textwidth,label=(a)]{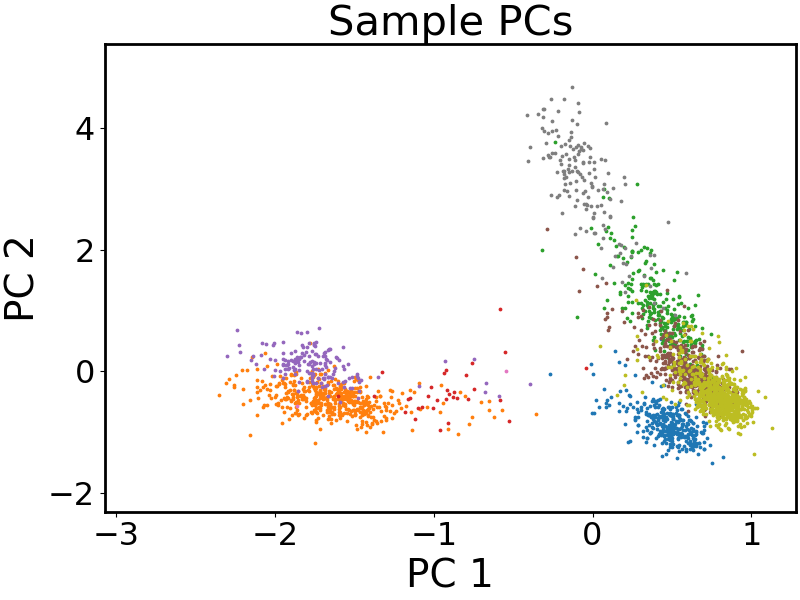}
\xincludegraphics[width=0.48\textwidth]{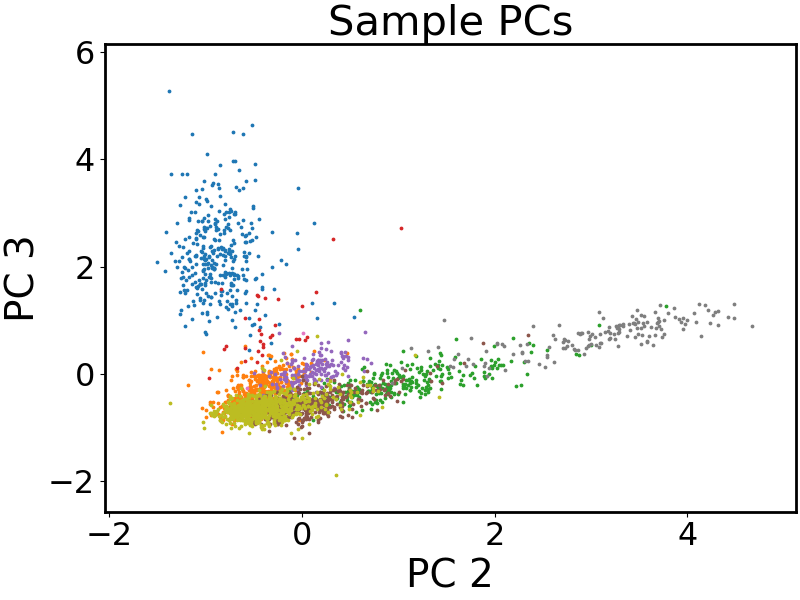}\endminipage%
\minipage{0.48\columnwidth}
\xincludegraphics[width=0.48\textwidth,label=(b)]{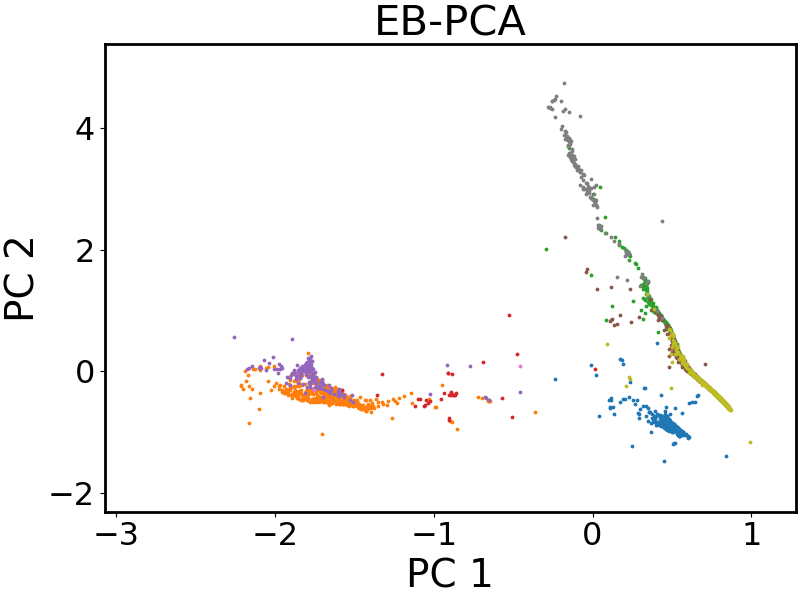}
\xincludegraphics[width=0.48\textwidth]{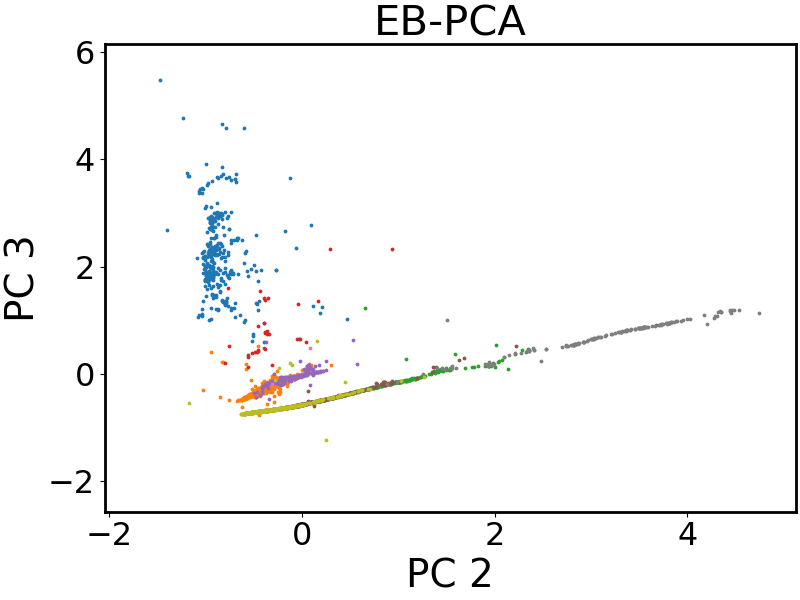}\endminipage%
\\
\minipage{\textwidth}
\xincludegraphics[width=\textwidth]{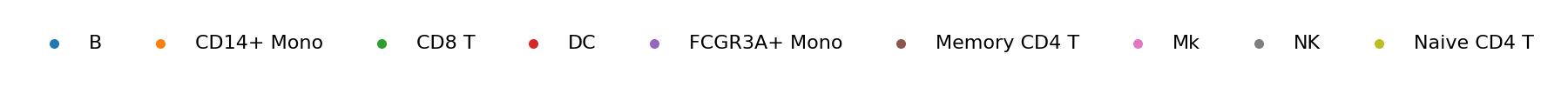}	
\endminipage
\caption{Estimated PCs using (a) standard PCA and (b) EB-PCA, for PBMC
single-cell RNA-seq gene expression data, with points colored by cell type.
The EB-PCA estimates
exhibit reduced estimation noise and clearer separation of cells by
cell type.
}\label{fig:PBMC}
\end{figure}

PCA is often the first step in single-cell gene expression data
analysis pipelines such as \texttt{Seurat} and \texttt{scanpy},
to capture the signatures of cell identity.
Nonlinear dimensionality reduction methods are often then applied with the
estimated PCs as input, to perform cell clustering and to infer cell types.

We illustrate an application of EB-PCA on scRNA-seq data of $d=2626$
Peripheral Blood Mononuclear
Cells (PBMCs) with $n=13{,}711$ gene expressions, from 10X Genomics.
This data is more representative of typical applications of EB-PCA,
in which there are insufficient samples to provide a known ground truth.
Therefore, we applied EB-PCA without subsampling. 
The singular values are shown in Figure~\ref{fig:singvals}, and
3 large outlier values are apparent. The corresponding
leading 3 PCs are depicted as two scatterplots in Figure~\ref{fig:PBMC}(a).
Qualitatively, these scatterplots exhibit estimation noise that resembles the
noise previously observed in the \emph{subsampled} 1000 Genomes and HapMap3
genotype matrices, suggesting that the noise may be a consequence of high
dimensionality.

Figure~\ref{fig:PBMC}(b) displays the estimated PCs using EB-PCA. Although there
is not a basis for quantitative comparison with PCA in this example, we
believe that the EB-PCA results may be more accurate for an underlying ground
truth. Qualitatively, the estimation noise exhibited in
Figure~\ref{fig:PBMC}(a) has been reduced. There is a clearer separation
between different cell types. For example, Na\"{\i}ve CD4 T cells are separated 
from CD14+ Monocytes in the plot of the 2nd vs.\ 3rd PCs estimated by EB-PCA,
whereas these overlap in the plot of the sample PCs.

\section{Theoretical guarantees}\label{sec:theory}

In this section, we summarize theoretical guarantees for EB-PCA. In the
context of the signal-plus-noise model (\ref{eq:model}),
our results show that EB-PCA asymptotically achieves the same (first-order)
estimation accuracy as the oracle Bayes AMP procedure. Consequently, the
asymptotic estimation error provably decreases across iterations, and can
approach the Bayes-optimal error as the number of iterations $t \to \infty$
under general conditions for the true priors.

The strategy of proof is to first establish asymptotic consistency of the NPMLEs
for the prior distributions, and then, through an inductive
comparison argument, show that the iterates of EB-PCA are characterized by the
same state evolution as oracle Bayes AMP. We note that prior work
on asymptotic consistency for the NPMLE assumes that a compound decision
model \eqref{eq:compoundmodel} with i.i.d.\ Gaussian errors holds exactly,
whereas such an error model holds only approximately for the sample PCs and AMP
iterates. Our arguments show that this approximation is sufficient for
asymptotic consistency.

\subsection{Assumptions}\label{sec:assumptions}

Consider the rank-$k$ model of \eqref{eq:model}.
We assume the empirical distributions of rows of
$\bU \in \bbR^{n \times k}$ and $\bV \in \bbR^{d \times k}$
converge in Wasserstein-2 distance to fixed distributions
$\bar{\pi}_*,\pi_*$ on $\bbR^k$, as $n,d \to \infty$.
This means that for any continuous function $\psi:\bbR^k \to \bbR$ 
with $\bbE_{U \sim \bar{\pi}_*}[\psi(U)^2]<\infty$, we have
$\lim_{n \to \infty} n^{-1}\sum_{i=1}^n \psi(u_{i1},\ldots,u_{ik})
=\bbE_{U \sim \bar{\pi}_*}[\psi(U)]$,
and similarly for $\bV$. We denote this convergence both as
$\bU \toW \bar{\pi}_*$ and as $\bU \toW U$ for a random vector
$U \sim \bar{\pi}_*$. Our model assumptions are then summarized as follows.

\begin{assumption}\label{assump:model}
$\cP$ is a family of probability distributions on $\bbR^k$ having
finite second moment, and $n,d \to \infty$ such that
\begin{enumerate}[topsep=0pt,itemsep=-1ex,partopsep=1ex,parsep=1ex]
\item[(a)] $\bW \in \bbR^{n \times d}$ has entries $w_{ij} \overset{iid}{\sim}
\cN(0,1/n)$.
\item[(b)] $k$, $s_1,\ldots,s_k$, and $\gamma \equiv d/n$
remain constant, where $s_1>\ldots>s_k>s_*(\gamma) \equiv \gamma^{-1/4}$.
\item[(c)] $\bU \toW \bar{\pi}_*$ and $\bV \toW \pi_*$ for two distributions
$\bar{\pi}_*,\pi_* \in \cP$ that satisfy the normalizations,
for all $1\leq i \neq j \leq k$,
\begin{equation*}
	\bbE_{U \sim \bar{\pi}_*}[U_i^2]=1, 
	\quad
	\bbE_{U \sim \bar{\pi}_*}[U_iU_j]=0, 
	\qquad
	\bbE_{V \sim \pi_*}[V_i^2]=1, 
	\quad
	\bbE_{V \sim \pi_*}[V_iV_j]=0
\end{equation*}
\item[(d)] For any non-singular $M_* \in \bbR^{k \times k}$,
symmetric positive-definite $\Sigma_* \in \bbR^{k \times k}$,
and $\pi_* \in \cP$, there is a weakly open neighborhood $O$ of $\pi_*$
such that $\theta(x\mid M_*, \Sigma_*, \pi)$ is
Lipschitz in $x$ uniformly over $\pi \in O$.	
\end{enumerate}
\end{assumption}

Part (a) makes a Gaussian assumption for the
noise, and part (b) assumes for simplicity that the signal values
$s_1,\ldots,s_k$ are distinct and super-critical. 

Part (c) ensures the normalizations in
\eqref{eq:normalization} and \eqref{eq:orthogonality}. This can hold both when
$\bU \in \bbR^{n \times k}$ and $\bV \in \bbR^{d \times k}$ are deterministic
matrices whose columns are the (exactly) orthogonal true PCs, as well as
almost surely in a Bayesian setting when $\bU$ and $\bV$ are random with
independent rows generated from $\bar{\pi}_*$ and $\pi_*$. We will assume
$\bU$ and $\bV$ are deterministic, i.e.\ our results apply
conditional on $(\bU,\bV)$ in the Bayesian setting.

Part (d) assumes a Lipschitz property for the posterior mean functions,
as is common in analyses of AMP. This places a small
restriction on the prior class $\cP$; for example, the assumption holds if
$\cP$ is the class of all priors supported on a compact domain of $\bbR^k$.

We analyze EB-PCA in a slightly idealized setting where the noise variance
$1/n$ in part (a) is known,
rather than estimated as in Remark \ref{remark:scaling},
and where the NPMLE in Lines 5 and 9 of Algorithm \ref{alg:EBPCA}
are computed exactly in each iteration.
Our results may be extended to incorporate a consistent estimate of the noise
variance and an approximate NPMLE computed on a sufficiently fine
discretization of the support, by a standard comparison argument with this
idealized setting---we omit the details of such an argument for brevity.

\subsection{Limiting risk for the initial empirical Bayes
estimates}\label{sec:PCAanalysis}

In the compound decision model $\Theta \sim \pi$ and $X \mid \Theta \sim
\cN(M \cdot \Theta,\;\Sigma)$, we denote the squared error
Bayes risk for estimating $\Theta$ based on $X$ as
\[\mmse(\pi \mid M,\Sigma)=\bbE\big[\|\Theta-\bbE[\Theta \mid X]\|_2^2\big].\]

The entrywise Gaussian approximation in \eqref{eq:PCAgaussianapprox} for the
sample PCs is formalized as the following proposition. Lemma C.1 of
\cite{montanari2017estimation} proves a similar result for the
symmetric spiked model.

\begin{proposition}\label{prop:PCAgaussian}
Let $\bar{M}_*,M_*,\bar{\Sigma}_*,\Sigma_* \in \bbR^{k \times k}$ be defined
by \eqref{eq:PCAM} and \eqref{eq:PCASigma}. Under Assumption \ref{assump:model},
almost surely $(\bU,\bF) \toW (U,F)$ and $(\bV,\bG) \toW (V,G)$
where $U \sim \bar{\pi}_*$, $F \mid U \sim
\cN(\bar{M}_* \cdot U,\;\bar{\Sigma}_*)$, 
$V \sim \pi_*$, and $G \mid V \sim \cN(M_* \cdot V,\;\Sigma_*)$.
\end{proposition}

Combined with an asymptotic consistency result for the NPMLE in approximate compound
decision models, shown in Lemma~\ref{lemma:mNPMLE} and
Corollary~\ref{cor:empBayes_relativeError}, this yields
the following asymptotic squared-error risks
for the initial empirical Bayes estimates of $\bU$ and $\bV$.

\begin{corollary}\label{cor:PCAmmse}
Let $\hat{\bU}$ and $\hat{\bV}$ be the initial empirical Bayes estimators in
\eqref{eq:initU} and \eqref{eq:initV}.
Let $\bar{M}_*,M_*,\bar{\Sigma}_*,\Sigma_* \in \bbR^{k \times k}$ be defined
by \eqref{eq:PCAM} and \eqref{eq:PCASigma}. Under Assumption \ref{assump:model},
almost surely
\begin{align*}
n^{-1}\|\hat{\bU}-\bU\|_F^2 \to \mmse(\bar{\pi}_* \mid
\bar{M}_*,\bar{\Sigma}_*), \quad
\qquad d^{-1}\|\hat{\bV}-\bV\|_F^2 \to \mmse(\pi_* \mid M_*,\Sigma_*).
\end{align*}
\end{corollary}

Thus the asymptotic squared-error risk for $\hat{\bU}$ is the Bayes risk
for estimating $\Theta \sim \bar{\pi}_*$ based on
$X \mid \Theta \sim \cN(\bar{M}_*\cdot \Theta,\;\bar{\Sigma}_*)$, and similarly
for $\hat{\bV}$.
In contrast, the analogous risks for the naive sample
PCs $\bF$ and $\bG$, or more generally of the best shrinkage estimators obtained
by rescaling their columns,
correspond to the Bayes risks of the best \emph{linear} estimators of $\Theta$
in these compound decision problems. These linear risks may be substantially
larger if the priors $\bar{\pi}_*$ and $\pi_*$ are far from the standard
Gaussian law.

\subsection{Limiting risk and Bayes optimality of EB-PCA}\label{sec:bayesoptimal}

For the full EB-PCA method, the following
verifies that the NPMLEs $\bar{\pi}_t,\pi_t$
remain consistent for $\bar{\pi}_*,\pi_*$ across iterations.
Consequently, the EB-PCA iterates satisfy the same Gaussian
approximations \eqref{eq:AMPgaussianapprox} and are tracked by the same state
evolution \eqref{eq:SE} as the oracle Bayes AMP algorithm.

\begin{theorem}\label{thm:EBPCA}
Suppose Assumption \ref{assump:model} holds.
Then almost surely for each fixed iteration $t \in \{0,\ldots,T\}$ of Algorithm
\ref{alg:EBPCA}, $\pi_t$ converges
weakly to $\pi_*$ and $\bar{\pi}_t$ converges weakly to $\bar{\pi}_*$.
Furthermore,
\[(\bU,\bF^t) \toW (U,F_t) \quad \text{and} \quad (\bV,\bG^t) \toW (V,G_t)\]
where $U \sim \bar{\pi}_*$, $F_t \mid U \sim \cN(\bar{M}_{*,t} \cdot U,\;
\bar{\Sigma}_{*,t})$, $V \sim \pi_*$, and
$G_t \mid V \sim \cN(M_{*,t} \cdot V,\;\Sigma_{*,t})$. The matrices
$\bar{M}_{*,t},\bar{\Sigma}_{*,t},M_{*,t},\Sigma_{*,t}$ are
defined iteratively by the state evolution in \eqref{eq:SE}.
\end{theorem}

As a corollary, the asymptotic squared-error risk for each EB-PCA iterate is
the same as that achieved by the oracle Bayes AMP algorithm with known priors.

\begin{corollary}\label{cor:EBPCArisk}
Let $\bar{M}_{*,t},\bar{\Sigma}_{*,t},M_{*,t},\Sigma_{*,t}$ be
defined iteratively by the state evolution in \eqref{eq:SE}.
Under Assumption \ref{assump:model}, almost surely for each fixed
iteration $t \in \{0,\ldots,T\}$,
\begin{align}\label{eq:EBPCAmmse}
\begin{aligned}
\frac{1}{n}\|\bU^t-\bU\|_F^2 \to \mmse(\bar{\pi}_* \mid \bar{M}_{*,t},
\bar{\Sigma}_{*,t}), \quad
\frac{1}{d}\|\bV^t-\bV\|_F^2 \to \mmse(\pi_* \mid M_{*,t}, \Sigma_{*,t}).
\end{aligned}
\end{align}
\end{corollary}

To study the decrease of these errors across iterations,
we follow \cite{miolane2017fundamental} and
introduce the positive-definite state matrices
\begin{align*}
\bar{Q}_{*,t}= \frac{1}{\gamma}{S}^{-1/2}\bar{M}_{*,t}^\top \bar{\Sigma}_{*,t}^{-1}
\bar{M}_{*,t}{S}^{-1/2},\quad
Q_{*,t}={S}^{-1/2}M_{*,t}^\top \Sigma_{*,t}^{-1} M_{*,t}{S}^{-1/2}.
\end{align*}
These are matrix-valued measures of the signal-to-noise ratios in the
compound decision models associated to $\bF^t$ and $\bG^t$ in each iteration.
For the standardized compound decision model
\[\Theta \sim \pi, \qquad X \mid \Theta \sim \cN(\Theta,\,Q^{-1})\]
(with $M=\Id$) parametrized by $\pi$ and $Q$, define the map
\[F_\pi(Q)=\bbE\big[\bbE[\Theta \mid X] \cdot \bbE[\Theta \mid X]^\top\big].\]
We verify in Appendix~\ref{sec:limiting_risk}
that the state evolution in \eqref{eq:SE} is equivalently expressed as
\begin{equation}\label{eq:SEQ}
\bar{Q}_{*,t}=F_{S^{1/2}\pi_*}(Q_{*,t}),
\quad Q_{*,t+1}=F_{S^{1/2}\bar{\pi}_*}(\gamma \cdot \bar{Q}_{*,t}),
\end{equation}
where ${S}^{1/2}\pi$ is the distribution of
${S}^{1/2} \cdot \Theta$ when $\Theta \sim \pi$.
The following shows that the state evolution converges to a
fixed point of these maps, the squared-error risks
of the EB-PCA iterates improve over the initial empirical Bayes
estimates, and these risks decrease monotonically over iterations.

\begin{proposition}\label{prop:SE}
Under Assumption \ref{assump:model}, for each $t=0,1,2,\ldots$
\begin{enumerate}[topsep=0pt,itemsep=-1ex,partopsep=1ex,parsep=1ex]
\item[(a)] $\bar{Q}_{*,t+1} \succeq \bar{Q}_{*,t}$ and
$Q_{*,t+1} \succeq Q_{*,t}$. Furthermore, as $t \to \infty$,
these matrices converge to a fixed point of
\begin{equation}\label{eq:fixedpoint}
\bar{Q}=F_{S^{1/2}{\pi}_*}(Q),
\qquad Q=F_{S^{1/2}\bar{\pi}_*}(\gamma \cdot \bar{Q} ).
\end{equation}
\item[(b)] Let $\mmse(\bar{\pi}_* \mid \bar{M}_*,\bar{\Sigma}_*)$ and
$\mmse(\pi_* \mid M_*,\Sigma_*)$ be the risks of the
initial empirical Bayes estimates in Corollary \ref{cor:PCAmmse}. Then
the asymptotic risks in \eqref{eq:EBPCAmmse} satisfy
\begin{align*}
\mmse(\bar{\pi}_* \mid \bar{M}_{*,t+1},\bar{\Sigma}_{*,t+1}) \leq \mmse(\bar{\pi}_* \mid \bar{M}_{*,t},\bar{\Sigma}_{*,t})\leq \mmse(\bar{\pi}_* \mid \bar{M}_*,\bar{\Sigma}_*), \\
\mmse(\pi_* \mid M_{*,t+1},\Sigma_{*,t+1})
\leq \mmse(\pi_* \mid M_{*,t},\Sigma_{*,t})\leq \mmse(\pi_* \mid M_*,\Sigma_*).
\end{align*}
\end{enumerate}
\end{proposition}

Finally, suppose $\bar{\pi}_*,\pi_*$ are such that
\eqref{eq:fixedpoint} has a unique fixed point $(\bar{Q},Q)$. In the Bayesian
setting where rows of $\bU$ and $\bV$ are random and i.i.d., Proposition 3 of
\cite{miolane2017fundamental} shows that this fixed
point characterizes the Bayes-optimal squared-error risk for estimating
$\bU{S}\bV^\top$:
\begin{multline*}
\lim_{n,d \to \infty}\left(\inf_{\hat{\bX}}
\frac{1}{nd}\bbE\Big[\|\hat{\bX}(\bY)-\bU{S}\bV^\top\|_F^2\Big]\right)
=\Tr \bbE_{(U,V) \sim \bar{\pi}_* \times \pi_*}
[UU^\top {S} VV^\top {S}]-\Tr \bar{Q}Q,
\end{multline*}
where the infimum is over all (measurable)
estimators $\hat{\bX}(\bY)$ and achieved at the Bayes estimator $\bbE[\bU S
\bV^\top \mid \bY]$. Exact computation of this Bayes estimator may be
intractable. The below verifies that the asymptotic error of
EB-PCA (and hence also of the oracle Bayes AMP algorithm)
approaches this Bayes-optimal error as $t \to \infty$. This type of result has 
been stated for rank $k=1$ in \cite{barbier2016mutual,montanari2017estimation}.

\begin{proposition}\label{prop:bayes_optimality}
Under Assumption \ref{assump:model}, suppose the fixed point
$(\bar{Q},Q)$ of \eqref{eq:fixedpoint} is unique. For any fixed $t \geq 1$,
let $(\bU^t,\bV^t)$ be the EB-PCA estimates of $(\bU,\bV)$ in iteration $t$
of Algorithm \ref{alg:EBPCA},
and let $\hat{S}=\diag(\hat{s}_1,\ldots,\hat{s}_k)$ be the estimate of $S$ from
\eqref{eq:shat}. Then almost surely as $n,d \to \infty$,
\begin{multline*}
\frac{1}{nd}\|\bU^t\hat{S}(\bV^t)^\top-\bU{S}\bV^\top\|_F^2
 \to \Tr \bbE_{(U,V) \sim \bar{\pi}_* \times \pi_*}[UU^\top {S} VV^\top {S}]-\Tr
\bar{Q}Q-o_t(1)
\end{multline*}
where $o_t(1)$ is a deterministic quantity satisfying $o_t(1) \to 0$
as $t \to \infty$.
\end{proposition}

For rank $k=1$, we refer readers to
\cite{barbier2016mutual,montanari2017estimation,lelarge2019fundamental} for
examples of priors $\bar{\pi}_*,\pi_*$ for which uniqueness of the fixed
point to \eqref{eq:fixedpoint} does and does not hold. It has been conjectured
that in examples where this fixed point is not unique, the asymptotic risk
corresponding to the fixed point that is reached by EB-PCA/oracle-Bayes AMP
is the smallest Bayes risk that is attainable by any estimator in polynomial
time \cite{lesieur2015phase,barbier2016mutual,lelarge2019fundamental,antenucci2019glassy}.

\section{Conclusion}\label{sec:conclusion}

We have described an EB-PCA procedure for performing PCA in high dimensions,
which couples classical empirical Bayes ideas with high-dimensional asymptotic
theory. In applications where the joint distribution of PCs has
non-Gaussian structure, EB-PCA can improve estimation accuracy by
obtaining a nonparametric estimate of this structure. 

EB-PCA is an example of a more general paradigm of carrying
out ``TAP-corrected'' variational Bayesian inference in high dimensions using
an empirical Bayes approach. The high dimensionality of the latent variable
space becomes a blessing in such problems, enabling the estimation of
complex and nonparametric prior distributions for these latent variables using
empirical Bayes ideas. This general paradigm may potentially be implemented
with other inference algorithms and extended to other inference problems.

We conclude with a discussion of a direction for future work. The
quantitative form of EB-PCA---and of AMP-based procedures more generally---is
derived assuming that the noise matrix $\bW$ has independent entries with
common variance. (To keep the proofs simple, we have also assumed that these
entries are Gaussian, but it is known that the Gaussian assumption can be
substantially weakened: universality results of this type have been shown for
the sample PCs in \cite{bao2018singular,bao2020statistical}, and for the AMP
state evolution in \cite{bayati2015universality,chen2020universality}.)

It is an open question to extend the procedure to settings
where $\bW$ has correlation structure, which is commonly reflected in data
by an overdispersed singular value distribution. One such setting
that is partially understood is that of bi-rotationally invariant matrices $\bW$
that satisfy the equality in law
$\bW\overset{L}{=}\bO^\top \bW\bQ$ for any orthogonal matrices
$\bO \in \bbR^{n \times n}$ and $\bQ \in \bbR^{d \times d}$. This reflects an
assumption that $\bW$ may have an arbitrary distribution of singular values
but ``generic'' singular vectors. The asymptotic behavior of sample PCs in this
setting was studied by \cite{benaych2012singular}, who provided the
quantitative forms for the limits \eqref{eq:PCAsigma}. In the rank-one case,
it was shown also in \cite{fan2020approximate} that the
Onsager correction in \eqref{eq:bayes_amp} and AMP state
evolution \eqref{eq:SE} are to be replaced by certain series expressions
defined by the free cumulants of $\bW$, which may be estimated from the 
singular value distribution of $\bY$.

Other models for $\bW$---for example general covariance models 
$\bW=\bZ\bB^{1/2}$ \cite{bai2012sample,bao2020statistical} or
separable covariance models
$\bW=\bA^{1/2}\bZ\bB^{1/2}$ \cite{yang2019edge,ding2019spiked} where
$\bZ \in \bbR^{n \times d}$ has i.i.d.\ entries---may also be studied. 
Developing variational Bayesian procedures that are both asymptotically
exact and computationally efficient in these and related models is an
interesting direction for future work, and we
believe it is likely that such developments may improve and robustify the
EB-PCA procedure, when they become available.

\section{Acknowledgments}

We thank Yixuan Ye, Chao Zhou and Hongyu Zhao for their help in collecting and interpreting data,
and Yihong Wu for helpful discussions
about the Kiefer-Wolfowitz NPMLE and ideas in the early stages of the
theoretical analyses. This research was supported in part by
NSF Grant DMS-1916198.


\newpage

\appendix


\def\ball{\mathbb{B}}
\newcommand{\bR}{\mathbf{R}}
\newcommand{\asto}{\overset{\mathrm{a.s.}}{\to}}

\section{Proof Preliminaries}\label{sec:preliminaries}
\paragraph{Notations and conventions.}

We reserve bold-face letters for vectors and matrices involving the increasing
dimensions $n$ and $d$. $\bu \otimes \bv=\bu\bv^\top$ is the outer product,
and $\bu^{\otimes 2}=\bu \otimes \bu$. For a function $f:\bbR^p \to \bbR^q$,
$\der f \in \bbR^{q \times p}$ is its Jacobian matrix. When $q=1$,
$\der f \in \bbR^{1 \times p}$ is the gradient as a row vector, $\der^2 f \in
\bbR^{p \times p}$ is the Hessian, and $\der^j f \in \bbR^{p \times \ldots
\times p}$ is the tensor of $j^\text{th}$ partial derivatives.

$\|\cdot\|$ is the Euclidean norm for vectors and Euclidean operator norm for
matrices and tensors. $\|\cdot\|_F$ is the Frobenius norm for matrices.
$\ball^k(B) = \{x \in \bbR^k:\|x\|\leq B\}$ is the closed
ball of radius $B$. For positive-definite $\Sigma$, $\phi_\Sigma(x)$ is the
Lebesgue density function of the distribution $\cN(0, \Sigma)$.
$f_{(M \cdot \pi)*\cN(0,\Sigma)}(x)$ is the marginal density of $X$ in the model
$\Theta \sim \pi$ and $X \mid \Theta \sim \cN(M \cdot \Theta,\Sigma)$.

\subsection{Wasserstein convergence and pseudo-Lipschitz functions}

\begin{definition}[Wasserstein convergence]
Let $\pi_n,\pi$ be probability measures on $\bbR^k$ with finite second moment.
Then $\pi_n$ converges to $\pi$ in the Wasserstein-2 distance as $n \to \infty$
if, for any $C>0$ and continuous function $f:\mathbb{R}^k \to \mathbb{R}$
satisfying 
	\begin{align}\label{eq:polynomial_growth}
	f(u) \leq C(1+\|u\|^2) \qquad \text{ for all } u \in \bbR^k,
	\end{align}
we have $\mathbb{E}_{U \sim \pi_n}[f(U)] \to \mathbb{E}_{U \sim \pi}[f(U)]$.
\end{definition}
For a matrix $\bU \in \bbR^{n \times k}$, we write $\bU \toW \pi$ if the
empirical distribution of rows of $\bU$ converges in Wasserstein-2 distance to
$\pi$. We denote this also as $\bU \toW U$, for $U \sim \pi$.

\begin{definition}[Pseudo-Lipschitz functions]
A function $\psi:\mathbb{R}^k \to \mathbb{R}$ is 2-pseudo-Lipschitz if
there exists a constant $L>0$ such that for all $u,u'\in \mathbb{R}^k$,
	\begin{align}\label{eq:PL_definition}
		|\psi(u) - \psi(u')| \leq L(1 + \|u\| + \|u'\|)\|u-u'\|.
	\end{align}
We write $\psi \in \PL(2)$.
\end{definition}

Note that any $\psi \in \PL(2)$ satisfies 
$\psi(u) \leq \psi(0)+L(1+\|u\|)\|u\| \leq C(1+\|u\|)^2$, so if $\bU \toW U$
and $\bU$ has rows $u_1,\ldots,u_n$, then
\[n^{-1}\sum_{i=1}^n \psi(u_i) \to \bbE[\psi(U)].\]

\subsection{Data normalization}\label{sec:normalization}

Consider the model
\[\bY_\text{obs}=\frac{1}{n} \cdot \bU S \bV^\top+\bW\]
where $\bU S \bV^\top$ is any matrix of rank $k$, and
$w_{ij} \overset{iid}{\sim} \cN(0,\tau^2/n)$.
We show that the estimate $\hat{\tau}^2$ in (\ref{eq:tausq}) of the Main Text
is consistent for $\tau^2$, for fixed $k$ as $n,d \to \infty$ (at any rates, and
not necessarily proportionally).

Let $\lambda_1(\bY_{\text{obs}}) \geq \ldots \geq \lambda_d(\bY_{\text{obs}})$ be
the singular values of $\bY_{\text{obs}}$,
where the last $d-n$ are 0 if $d>n$. Then
\[\hat{\tau}^2=\frac{1}{d}\sum_{i=k+1}^d \lambda_i(\bY_{\text{obs}})^2.\]
Let $\lambda_1(\bW) \geq \ldots \geq \lambda_d(\bW)$ be the singular values of
$\bW$. By Weyl's inequality,
\[\lambda_i(\bY_{\text{obs}}) \leq \lambda_{i-k}(\bW) \text{ for all } i \geq
k+1, \qquad \lambda_i(\bY_{\text{obs}}) \geq \lambda_{i+k}(\bW)
\text{ for all } i \leq d-k.\]
Applying this to the above form of $\hat{\tau}^2$,
\[\frac{1}{d}\|\bW\|_F^2-\frac{1}{d}\sum_{i=1}^{2k} \lambda_i(\bW)^2
=\frac{1}{d}\sum_{i=2k+1}^d \lambda_i(\bW)^2 \leq
\hat{\tau}^2 \leq \frac{1}{d} \sum_{i=1}^{d-k} \lambda_i(\bW)^2
\leq \frac{1}{d}\|\bW\|_F^2.\]
Since $w_{ij} \overset{iid}{\sim} \cN(0,\tau^2/n)$, we have $d^{-1}\|\bW\|_F^2=d^{-1}\sum_{i,j}
w_{ij}^2 \to \tau^2$ a.s.\ as $n,d \to \infty$ by the law of large numbers.
We have also $\lambda_i(\bW) \leq \lambda_1(\bW) \leq C(1+\sqrt{d/n})$
a.s.\ for an absolute constant $C>0$ and all large $n$ and $d$,
see e.g.\ Theorem 4.4.5 of \cite{vershynin2018high}. So
$d^{-1}\sum_{i=1}^{2k} \lambda_i(\bW)^2 \to 0$,
and $\hat{\tau}^2 \to \tau$ a.s.

\subsection{Singular vectors of spiked random matrices}

The following is a special case of
\cite[Theorem 2.9]{benaych2012singular}. An analogue
in the spiked covariance model was first
established in \cite{paul2007asymptotics}.

\begin{lemma}\label{lemma:BBP}
Under Assumption \ref{assump:model}, for each $i=1,\ldots,k$ and some
choices of signs for $\bff_i$ and $\bg_i$, a.s.\ as $n,d \to \infty$,
\begin{gather*}
\sqrt{\gamma} \cdot \lambda_i \to \sqrt{(\gamma s_i^2+1)
(s_i^2+1)/s_i^2},\\
\begin{aligned}
n^{-1} \bff_i^\top \bu_i \to \bar{\mu}_{*,i}, \quad
d^{-1} \bg_i^\top \bv_i \to \mu_{*,i}
\end{aligned}
\end{gather*}
for the values $\bar{\mu}_{*,i}$ and $\mu_{*,i}$ defined in \eqref{eq:PCAsigma}.
Furthermore, $n^{-1}\bff_i^\top \bu_j \to 0$
and $d^{-1}\bg_i^\top \bv_j \to 0$ a.s.\ for
all $j \in \{1,\ldots,k\} \setminus \{i\}$.
\end{lemma}
\begin{proof}
Recall
\[\bY=\frac{1}{n}\bU S\bV^\top+\bW.\]
Write $\tilde{\bU}=(\tilde{\bu}_1,\ldots,\tilde{\bu}_k)
\in \mathbb{R}^{n \times k}$ for the Gram-Schmidt
orthogonalization of $\bU$, with columns scaled such that
$n^{-1}\|\tilde{\bu}_i\|^2=1$. Assumption \ref{assump:model}(c) implies that
$n^{-1}\bU^\top \bU \to \Id_{k \times k}$, from which it may be verified that
the Gram-Schmidt procedure yields
\begin{equation}\label{eq:gramschmidtU}
n^{-1}\|\bU-\tilde{\bU}\|_F^2 \to 0.
\end{equation}
Similarly, letting $\tilde{\bV}=(\tilde{\bv}_1,\ldots,\tilde{\bv}_k)
\in \mathbb{R}^{d \times k}$ be the Gram-Schmidt
orthogonalization of $\bV$ scaled such that $d^{-1}\|\tilde{\bv}_i\|^2=1$,
we have 
\begin{equation}\label{eq:gramschmidtV}
d^{-1}\|\bV-\tilde{\bV}\|_F^2 \to 0.
\end{equation}
Define $\tilde{\bY}=n^{-1}\tilde{\bU} S\tilde{\bV}^\top+\bW$, and denote its
leading rank-$k$ singular component as
\[n^{-1}\tilde{\bF}\tilde{\Lambda} \tilde{\bG}^\top
=\sum_{i=1}^n \frac{\tilde{\lambda}_i}{n} \cdot
\tilde{\bff}_i\tilde{\bg}_i^\top.\]
Then \eqref{eq:gramschmidtU} and \eqref{eq:gramschmidtV} together imply
$\|\tilde{\bY}-\bY\| \to 0$ in operator norm, so that by the condition of
distinct singular values in Assumption \ref{assump:model}(b) and the Weyl and
Davis-Kahan inequalities,
\[\|S-\tilde{S}\| \to 0, \qquad n^{-1}\|\bF-\tilde{\bF}\|_F^2 \to 0,
\qquad d^{-1}\|\bG-\tilde{\bG}\|_F^2 \to 0.\]
Combining this with \eqref{eq:gramschmidtU} and \eqref{eq:gramschmidtV} and
applying Cauchy-Schwarz,
\[\lim_{n,d \to \infty} n^{-1} \bU^\top \bF
=\lim_{n,d \to \infty} n^{-1}\tilde{\bU}^\top \tilde{\bF},
\qquad \lim_{n,d \to \infty} d^{-1} \bV^\top \bG
=\lim_{n,d \to \infty} d^{-1}\tilde{\bV}^\top \tilde{\bG}.\]

\cite[Theorem 2.9]{benaych2012singular} applies to describe the almost sure
limits of $\tilde{S}$,
$n^{-1}\tilde{\bU}^\top \tilde{\bF}$, and $d^{-1}\tilde{\bV}^\top
\tilde{\bG}$. The latter two limits are diagonal with the diagonal entries
\[n^{-1}\tilde{\bu}_i^\top \tilde{\bff}_i
\to \sqrt{1-\frac{(\gamma+\gamma s_i^2)}{\gamma s_i^2(\gamma s_i^2+1)}},
\qquad d^{-1}\tilde{\bv}_i^\top \tilde{\bg}_i
\to \sqrt{1-\frac{\gamma(1+\gamma s_i^2)}{\gamma s_i^2(\gamma s_i^2+\gamma)}}.\]
These formulas are given by \cite[Eqs.\ 10--11]{benaych2012singular}
when $\bW$ is Gaussian, with the notational identifications $c
\leftrightarrow \gamma$, $\theta_i \leftrightarrow \sqrt{\gamma} \cdot
\tilde{s}_i$, and a swap of $\bu$ with $\bv$ because $\bW$ in 
\cite{benaych2012singular} is scaled to have variance $1/d$ rather than $1/n$.
We have replaced in these limits $\tilde{s}_i$ by $s_i$, using the above
convergence $\|S-\tilde{S}\| \to 0$. These limits are exactly $\bar{\mu}_{*,i}$
and $\mu_{*,i}$ from \eqref{eq:PCAsigma}, establishing the lemma.
\end{proof}

As the columns of $\bF$ and $\bG$ are only defined up to an arbitrary sign, we
henceforth fix these signs to match those in Lemma \ref{lemma:BBP}.
We now prove Proposition~\ref{prop:PCAgaussian} of the Main Text using the
results of \cite{benaych2012singular}.

\begin{proof}[Proof of Proposition~\ref{prop:PCAgaussian}]
We show the result for $(\bV,\bG)$; the statement for $(\bU,\bF)$ is analogous.
Let $\tilde{\bV}=(\tilde{\bv}_1,\ldots,\tilde{\bv}_k) \in
\mathbb{R}^{d \times k}$ be the Gram-Schmidt orthogonalization of $\bV$,
satisfying $d^{-1}\|\tilde{\bv}_i\|^2=1$ and \eqref{eq:gramschmidtV}. The 
projections onto and orthogonal to the column span of $\bV$ are given by
$\bP_\parallel=\tilde{\bV}\tilde{\bV}^\top/d$ and
$\bP_\perp=\Id-\tilde{\bV}\tilde{\bV}^\top/d$, respectively. Note that
\eqref{eq:gramschmidtV} and Lemma \ref{lemma:BBP} imply, almost surely,
\[\lim_{n,d \to \infty} d^{-1}\tilde{\bV}^\top \bG
=\lim_{n,d \to \infty} d^{-1}\bV^\top \bG=M_*^\top.\]
Then, since $d^{-1}\bG^\top \bG=\Id$, we have
\[d^{-1}(\bP_\perp\bG)^\top (\bP_\perp\bG) \to
\Id-M_*M_*^\top=\Sigma_* \succ 0.\]
This shows that $\bP_\perp\bG$ has full column rank $k$ almost surely for all
large $n,d$.

Now let $\tilde{\bV}_\perp=(\tilde{\bv}_1^\perp,\ldots,\tilde{\bv}_k^\perp)
\in \mathbb{R}^{d \times k}$ be the Gram-Schmidt orthogonalization of
$\bP_\perp\bG$, normalized so that $d^{-1}\|\tilde{\bv}_i^\perp\|^2=1$.
Then we may represent
	\begin{align}\label{eq:G_decomposition}
		\bG=\bP_\parallel\bG+\bP_\perp\bG=\tilde{\bV}\Omega
+\tilde{\bV}_\perp \Omega_\perp, \qquad \Omega=d^{-1}\tilde{\bV}^\top \bG,
\qquad \Omega_\perp=d^{-1}\tilde{\bV}_\perp^\top \bG.
	\end{align}
Here, almost surely
\begin{equation}\label{eq:Omegalim}
\Omega \to M_*^\top, \qquad \Omega_\perp^\top \Omega_\perp
=d^{-1}(\bP_\perp\bG)^\top (\bP_\perp\bG) \to \Sigma_*
\end{equation}
so $\Omega_\perp$ is invertible for all large $n,d$.
Let $\bQ \in \mathbb{R}^{d \times d}$ be any orthogonal matrix such that
$\bQ\tilde{\bV}=\tilde{\bV}$. Then also $\bQ\bV=\bV$, and
we have the equality in law
\[\bY \bQ^\top=\frac{1}{n}\bU S \bV^\top \bQ^\top + \bW \bQ^\top
		\overset{L}{=} \frac{1}{n} \bU S \bV^\top+ \bW
		= \bY.\]
Since $\bG$ contains the leading $k$ right singular vectors of $\bY$, this
implies $\bQ\bG\overset{L}{=}\bG$, so
\[\tilde{\bV}\Omega+\bQ\tilde{\bV}_{\perp}\Omega_\perp
=\bQ\bG \overset{L}{=}\bG=\tilde{\bV}\Omega+\tilde{\bV}_{\perp}\Omega_\perp.\]
Then $\tilde{\bV}_\perp$ satisfies the rotational invariance in law
$\bQ\tilde{\bV}_{\perp}\overset{L}{=}\tilde{\bV}_{\perp}$ for any such $\bQ$,
so it is Haar-uniformly distributed over the Stiefel manifold $\{\tilde{\bV} \in
\mathbb{R}^{d \times k}:\bV^\top \tilde{\bV}_\perp=0,\;
d^{-1}\tilde{\bV}_\perp^\top \tilde{\bV}_\perp=\Id\}$. Letting
$\bZ \in \mathbb{R}^{d\times k}$ be a random matrix with i.i.d.\ entries $Z_{ij}
\sim \cN(0,1)$, this means that we can construct $\tilde{\bV}_\perp=\bP_\perp \bZ \Sigma^{-1/2}$, where $\Sigma=\bZ^\top \bP_\perp \bZ/d$.
Then
\[\bG=\tilde{\bV}\Omega+\bP_\perp \bZ \Sigma^{-1/2}\Omega_\perp.\]

By the law of large numbers, $\Sigma \to \Id_{k \times k}$
and $d^{-1}\|\bP_\perp \bZ-\bZ\|_F^2=d^{-1}\|\tilde{\bV}(\tilde{\bV}^\top
\bZ)/d\|_F^2 \to 0$ almost surely. Thus
\[d^{-1}\Big\|\bG-\tilde{\bV}\Omega-\bZ \Omega_\perp\Big\|_F^2 \to 0.\]
Applying \eqref{eq:gramschmidtV} and \eqref{eq:Omegalim}, we then have
$(\bV,\bG) \toW (V,G)$ where $G=M_*V+\Sigma_*^{1/2}Z$ and
$Z \sim \cN(0,\Id)$, as desired.
\end{proof}

\section{Analysis of the NPMLE}\label{sec:npmle}
\newcommand{\toas}{\overset{\textnormal{a.s.}}{\to}}
\def\ball{\mathbb{B}}

For $M \in \bbR^{k \times k}$ invertible and
$\Sigma \in \bbR^{k \times k}$ symmetric positive-definite, consider
the compound decision model
$\Theta \sim \pi$ and $X \mid \Theta \sim \cN(M \cdot \Theta,\,\Sigma)$.
Denote the Lebesgue density function of $\cN(0,\Sigma)$ by
\[\phi_\Sigma(x)=\frac{1}{(2\pi)^{k/2}|\Sigma|^{1/2}}
\exp\Big(-\frac{1}{2} \cdot x^\top \Sigma^{-1} x \Big).\]
Then the marginal density of $X$ is
\begin{equation}\label{eq:Xmarginal}
f_{(M \cdot \pi) * \cN(0,\Sigma)}(x)
=\bbE\Big[\phi_\Sigma(x-M\Theta)\Big]
=\int \phi_\Sigma(x-M\theta)\der \pi(\theta).
\end{equation}
Noting that $\der \phi_\Sigma(x)=-x^\top \Sigma^{-1} \cdot \phi_\Sigma(x)$, we have
\begin{align*}
\frac{\der f^\top_{(M \cdot \pi) * \cN(0,\Sigma)}}
{f_{(M \cdot \pi) * \cN(0,\Sigma)}}(x)
&=\frac{\bbE[-\Sigma^{-1}(x-M\Theta) \cdot \phi_\Sigma(x-M\Theta)]}
{\bbE[\phi_\Sigma(x-M\Theta)]}\\
&=-\Sigma^{-1}x+\Sigma^{-1}M \cdot
\frac{\bbE[\Theta \cdot \phi_\Sigma(x-M\Theta)]}
{\bbE[\phi_\Sigma(x-M\Theta)]}.
\end{align*}
Thus we obtain Tweedie's form of the posterior mean denoiser,
\begin{equation}\label{eq:tweedie}
\theta(x \mid M,\Sigma,\pi)
=\frac{\bbE[\Theta \cdot \phi_\Sigma(x-M\Theta)]}
{\bbE[\phi_\Sigma(x-M\Theta)]}
=M^{-1}\bigg(x+\Sigma \cdot
\frac{\der f^\top_{(M \cdot \pi) * \cN(0,\Sigma)}}
{f_{(M \cdot \pi) * \cN(0,\Sigma)}}(x)\bigg).
\end{equation}

Assumption \ref{assump:model}(d) guarantees a uniform Lipschitz property of
$\theta(\cdot)$ in a neighborhood of $\pi_*$. This in fact implies a uniform
Lipschitz property in a neighborhood of $(M_*,\Sigma_*,\pi_*)$---we record this
fact in the following proposition.

\begin{proposition}\label{prop:uniformlipschitz}
Under Assumption~\ref{assump:model}(d), there is an open neighborhood $O$ of
$(M_*,\Sigma_*,\pi_*)$ (weakly with respect to $\pi_*$) such that
$\theta(x \mid M,\Sigma,\pi)$ is Lipschitz in $x$ uniformly over $(M,\Sigma,\pi)
\in O$.
\end{proposition}
\begin{proof}
For any $(M,\Sigma,\pi) \in O$, define
$\tilde{\pi}=M_*^{-1} \Sigma_*^{1/2}\Sigma^{-1/2} M \cdot \pi$. Then
\begin{align*}
\theta(x \mid M, \Sigma,\pi) &=
\bbE_{\Theta\sim \pi}[\Theta\mid M\Theta +\Sigma^{1/2} Z=x]\\
&=\bbE_{\Theta\sim \pi}\Big[\Theta\;\Big|\;
\Sigma^{1/2}_{*}\Sigma^{-1/2} M \Theta +
\Sigma^{1/2}_*Z=\Sigma^{1/2}_{*}\Sigma^{-1/2} x\Big]\\
&=\underbrace{(M_*^{-1}\Sigma_*^{1/2}\Sigma^{-1/2} M)^{-1}}_{K_1}
\bbE_{\Theta\sim \tilde{\pi}}\Big[\Theta \;\Big|\; M_*\Theta + \Sigma_*^{1/2}
Z=\underbrace{\Sigma_*^{1/2} \Sigma^{-1/2}}_{K_2} x\Big] \\
&= K_1 \theta(K_2 x \mid M_*,\Sigma_*, \tilde{\pi}).
\end{align*}
Since $\theta(x \mid M_*,\Sigma_*,\tilde{\pi})$ is Lipschitz in $x$ uniformly
over all $\tilde{\pi}$ in an open neighborhood of $\pi_*$ by
Assumption \ref{assump:model}(d), the result follows for a sufficiently small
neighborhood $O$ of $(M_*,\Sigma_*,\pi_*)$.
\end{proof}

\subsection{Consistency of the NPMLE}

The entrywise Gaussian approximations of \eqref{eq:PCAgaussianapprox} and
\eqref{eq:AMPgaussianapprox} for the sample PCs and AMP iterates will hold in
the sense of empirical Wasserstein convergence as $n,d \to \infty$. Thus we
first show a consistency result for the NPMLE in this setting, where
$\bX \mid \bTheta$ only approximately follows this Gaussian model.

\begin{lemma}\label{lemma:mNPMLE}
Fix $k \geq 1$, $M_* \in \bbR^{k \times k}$ non-singular, and $\Sigma_* \in
\bbR^{k \times k}$ symmetric positive-definite. Let $\cP$ be a class of
distributions on $\bbR^k$ satisfying Assumption~\ref{assump:model}, and fix
$\pi_* \in \cP$.
As $n \to \infty$, let $\bTheta,\bX \in \bbR^{n \times k}$ satisfy
\begin{equation}\label{eq:ThetaXconvergence}
(\bTheta,\bX) \toW (\Theta,X)
\end{equation}
where $\Theta \sim \pi_*$ and $X \mid \Theta \sim \cN(M_* \cdot
\Theta,\;\Sigma_*)$, and let $M_n,\Sigma_n \in \bbR^{k \times k}$ satisfy
$M_n \to M_*$ and $\Sigma_n \to \Sigma_*$.

Let $x_1,\ldots,x_n$ be the rows of $\bX$. Suppose $\pi_n \in \cP$ is any
approximate NPMLE in the sense
\begin{equation}\label{eq:approxNPMLE}
\liminf_{n \to \infty} \frac{1}{n}\sum_{i=1}^n
\log \frac{f_{(M_n \cdot \pi_n) * \cN(0,\Sigma_n)}(x_i)}
{f_{(M_* \cdot \pi_*) * \cN(0,\Sigma_*)}(x_i)} \geq 0.
\end{equation}
Then $\pi_n$ converges weakly to $\pi_*$ as $n \to \infty$.
\end{lemma}

This has the following oracle implication for the corresponding
empirical Bayes estimator of $\bTheta$: 

\begin{corollary}\label{cor:empBayes_relativeError}
In the setting of Lemma \ref{lemma:mNPMLE}, 
let $\theta(\bX \mid M_*,\Sigma_*,\pi_*)$ be the Bayes posterior mean function
as defined in \eqref{eq:multivariatedenoise}. Then
\begin{align}\label{eq:eb_optimal_bayes_risk}
	\frac{1}{n} \|\theta(\bX\mid M_*,\Sigma_*, \pi_*) - \theta(\bX\mid
M_n,\Sigma_n, \pi_n)\|_F^2 \to 0.
\end{align}
Consequently
\[\frac{1}{n}\|\theta(\bX\mid M_n,\Sigma_n, \pi_n)-\bTheta\|_F^2 \to
\mmse(\pi_* \mid M_*, \Sigma_*).\]
\end{corollary}

Thus the empirical Bayes estimate of $\bTheta$ achieves the oracle asymptotic
Bayes risk of the true prior $\pi_*$, as
long as the empirical convergence $(\bTheta,\bX) \toW (\Theta,X)$ holds, and
$M_n,\Sigma_n$ are consistent estimates of the parameters $M_*,\Sigma_*$
describing the conditional Gaussian law of $X \mid \Theta$.

The proofs of these results follow the ideas of
\cite{zhang2009generalized,jiang2009general,saha2020nonparametric}. Our
analyses are simpler, as we do not assume a particular rate of convergence
in \eqref{eq:ThetaXconvergence}, and we correspondingly do not study the
convergence rate of $\pi_n$ to $\pi_*$.

\def\suplim{\overline{\lim}\;}
\def\inflim{\underline{\lim}\;}

\begin{proof}[Proof of Lemma \ref{lemma:mNPMLE}]
Write as shorthand
\[f_*(x)=f_{(M_* \cdot \pi_*) * \cN(0,\Sigma_*)}(x),
\qquad f_n(x)=f_{(M_n \cdot \pi_n) * \cN(0,\Sigma_n)}(x),\]
and denote their squared Hellinger distance by
$d_H(f_n,f_*)^2=1-\int \sqrt{f_n(x)f_*(x)}\,\der x$. We show that
\begin{equation}\label{eq:hellingerconvergence}
d_H(f_n,f_*)^2 \to 0
\end{equation}
as $n \to \infty$. Fix any $\eps>0$, and suppose by contradiction that
$\limsup_{n \to \infty} d_H(f_n,f_*)^2>\eps$.
This implies that along some subsequence $\{n_m\}_{m=1}^\infty$, the function
$f_{n_m}$ belongs to the function class
\[
{
\mathcal{F}=\Big\{f_{(M \cdot \pi) * \cN(0,\Sigma)}:
d_H(f_{(M \cdot \pi) * \cN(0,\Sigma)},f_*)^2>\eps,
\,\pi \in \cP,\,\|M-M_*\|<\|M_*\|/2,\;
\|\Sigma-\Sigma_*\|<\|\Sigma_*\|/2\Big\}.
}
\]
We denote by $\inflim a_n \equiv \liminf_{m \to \infty} a_{n_m}$
and $\suplim a_n \equiv \limsup_{m \to \infty} a_{n_m}$
the limits along this subsequence.

Observe that from the form of \eqref{eq:Xmarginal}, the functions $f \in
\mathcal{F}$ are uniformly bounded by a constant $C_0>0$ and are equicontinuous.
Consider two constants $B,\eta>0$. Let
$\ball^k(B)=\{x \in \bbR^k:\|x\| \leq B\}$, and define the semi-norm
\begin{align*}
\|f\|_{\infty,B} = \sup_{x\in\ball^k(B)} |f(x)|.
\end{align*}
Then by Arzel\`a-Ascoli, there exists a finite ($n$-independent)
$\eta$-cover $\mathcal{C} \subset \mathcal{F}$ in this semi-norm: for
any $f \in \mathcal{F}$, there is some $g \in \mathcal{C}$ where
$\|f-g\|_{\infty,B}<\eta$. We define the smoothed indicator function
\begin{align*}
\mathfrak{f}(x)=\begin{cases} \eta & \text{ if } x \in \ball^k(B)\\
\eta \cdot (B/\|x\|)^{k+1} & \text{ if } x \not\in \ball^k(B).
\end{cases}
\end{align*}
For all large $n_m$, letting $g_{n_m} \in \mathcal{C}$ be such that
$\|f_{n_m}-g_{n_m}\|_{\infty,B}<\eta$, we then have the pointwise bounds
\begin{align*}
f_{n_m}(x) \leq 
\begin{cases}
g_{n_m}(x)+\mathfrak{f}(x) & \text{ if } x \in \ball^k(B) \\
C_0 & \text{ if } x \not \in \ball^k(B).
\end{cases}
\end{align*}

The given condition that $\pi_n$ is an approximate NPMLE implies that
\begin{align*}
0 \leq \inflim
\frac{1}{n}\sum_{i=1}^n \log \frac{f_n(x_i)}{f_*(x_i)}
&\leq \suplim \frac{1}{n}
\sum_{i=1}^n \log \frac{g_n(x_i)+\mathfrak{f}(x_i)}{f_*(x_i)}
+\suplim \frac{1}{n} \sum_{i:x_i \notin \ball^k(B)}
\log \frac{C_0}{g_n(x_i)+\mathfrak{f}(x_i)}\\
&=\mathrm{(I)}+\mathrm{(II)}.
\end{align*}
For $\mathrm{(I)}$, observe that
since $\mathcal{C}$ is a finite $n$-independent cover, we have
\[\mathrm{(I)} \leq
\max_{g \in \mathcal{C}}\;\suplim
\frac{1}{n}\sum_{i=1}^n \log \frac{g(x_i)+\mathfrak{f}(x_i)}{f_*(x_i)}.\]
For any $f=f_{(M \cdot \pi) * \cN(0,\Sigma)} \in \mathcal{F}$,
its posterior mean denoising function $\theta(x \mid M,\Sigma,\pi)$ is given by
\eqref{eq:tweedie}. Then applying Assumption \ref{assump:model}(d), for each
fixed $f \in \mathcal{F}$ and a constant $C_f>0$, we have
\[\|\der^2 \log f(x)\|=\|\Sigma^{-1} \cdot (\Id+M \cdot \der \theta(x \mid
M,\Sigma,\pi))\|<C_f.\]
Thus $\log f_* \in \PL(2)$
and $\log g \in \PL(2)$ for all $g \in \mathcal{C}$. It may be checked that
$\log \mathfrak{f} \in \PL(2)$, and since $(a,b) \mapsto \log(e^a+e^b)$ is
Lipschitz, this implies that $x \mapsto \log(g(x)+\mathfrak{f}(x))$ also
belongs to $\PL(2)$. Then by the convergence $\bX \toW X \sim f_*$, we obtain
\[\mathrm{(I)} \leq \max_{g \in \mathcal{C}}
\bbE_{X \sim f_*}\bigg[\log \frac{g(X)+\mathfrak{f}(X)}{f_*(X)}\bigg].\]
For each $g \in \mathcal{C}$, applying $(\log x)/2=\log \sqrt{x} \leq
\sqrt{x}-1$, $\sqrt{x+y} \leq \sqrt{x}+\sqrt{y}$, and Cauchy-Schwarz,
\begin{align*}
\frac{1}{2}\,\bbE_{X \sim f_*}\bigg[\log \frac{g(X)+\mathfrak{f}(X)}
{f_*(X)}\bigg]
&\leq \bbE_{X \sim f_*}
\bigg[\sqrt{\frac{g(X)+\mathfrak{f}(X)}{f_*(X)}}\bigg]-1\\
&=\int \sqrt{(g(x)+\mathfrak{f}(x))f_*(x)}\,dx -1\\
&\leq \int \sqrt{g(x)f_*(x)}\,dx-1+\int \sqrt{\mathfrak{f}(x)f_*(x)}\,dx\\
&\leq -d_H(g,f_*)^2+\sqrt{\int \mathfrak{f}(x)\,dx}.
\end{align*}
The condition $g \in \mathcal{C} \subset \mathcal{F}$ implies
$d_H(g,f_*)^2>\eps$. Furthermore,
$\int \mathfrak{f}(x)\,dx=\int \mathfrak{f}(By)B^kdy=C_k\eta \cdot B^k$
for a constant $C_k>0$. Thus
$\mathrm{(I)} \leq -2\eps+2\sqrt{C_k \eta \cdot B^k}$.

For $\mathrm{(II)}$, let us set
\[\psi_{\eta,B}(x)
=\log\;\max\bigg(\frac{C_0\|x\|^{k+1}}{\eta B^{k+1}},\;1\bigg).\]
Then clearly $\psi_{\eta,B} \in \PL(2)$, so
\[\mathrm{(II)}
\leq \suplim
\frac{1}{n}\sum_{i:x_i \notin \ball^k(B)} \log \frac{C_0}{\mathfrak{f}(x_i)}
\leq \suplim \frac{1}{n}\sum_{i=1}^n \psi_{\eta,B}(x_i)
=\bbE_{X \sim f_*}[\psi_{\eta,B}(X)].\]
Combining with the above bound for $\mathrm{(I)}$, we obtain
\[0 \leq -2\eps+2\sqrt{C_k \eta \cdot B^k}
+\bbE_{X \sim f_*}[\psi_{\eta,B}(X)].\]
This must hold for all $B,\eta>0$. However, taking $B \to \infty$ and $\eta \to
0$ such that $\eta B^{k+1} \to \infty$ but $\eta B^k \to 0$, we have
$2\sqrt{C_k \eta \cdot B^k} \to 0$ and
$\bbE_{X \sim f_*}[\psi_{\eta,B}(X)] \to 0$ by the dominated convergence
theorem. This yields a contradiction for sufficiently large $B$ and small
$\eta$. Then we must have $\limsup_{n \to \infty} d_H(f_n,f_*)^2 \leq \eps$.
Here $\eps>0$ is arbitrary, so we have shown \eqref{eq:hellingerconvergence}.

This implies the weak convergence of $(M_n \cdot \pi_n) * \cN(0,\Sigma_n)$
to $(M_* \cdot \pi_*) * \cN(0,\Sigma_*)$. If $\psi_n$ and $\psi_*$ are the
characteristic functions of $\pi_n$ and $\pi_*$, then for all $t \in \bbR^k$,
\[\psi_n(M_n^\top t) \cdot e^{-\frac{t^\top \Sigma_n t}{2}}
\to \psi_*(M_*^\top t) \cdot e^{-\frac{t^\top \Sigma_* t}{2}}\]
as the left side is the characteristic function of
$(M_n \cdot \pi_n) * \cN(0,\Sigma_n)$ while the right side is that of
$(M_* \cdot \pi_*) * \cN(0,\Sigma_*)$. Since $\Sigma_n \to \Sigma_*$ and
$M_n \to M_*$, this implies $\psi_n(M_*^\top t) \to \psi_*(M_*^\top t)$.
Then $\psi_n(t) \to \psi_*(t)$ for all $t \in \bbR^k$, because $M_*$ is
invertible. So $\pi_n$ converges weakly to $\pi_*$.
\end{proof}

\subsection{Bayes posterior mean and its derivatives}

We now show Corollary \ref{cor:empBayes_relativeError} and an analogous
statement for the derivative of $\theta(\cdot)$.

\begin{lemma}\label{lemma:weakconvergence}
Let $\pi_n,\pi_*$ be probability distributions on $\bbR^k$ such that $\pi_n$
converges weakly to $\pi_*$. Let $g_n,g_*:\bbR^k \times \bbR^k \to \bbR^m$ be
such that
\begin{enumerate}
\item[(a)] For any $B>0$,
$g_n(x,\theta) \to g_*(x,\theta)$ uniformly over $(x,\theta) \in \bbR^k
\times \ball^k(B)$.
\item[(b)] For some constant $C_0>0$ and all $(x,\theta) \in \bbR^k \times \bbR^k$,
$\|g_*(x,\theta)\|<C_0$ and $\|g_n(x,\theta)\|<C_0$.
\item[(c)] For any $B>0$, some constant $L_B>0$, and all 
$(x,\theta) \in \bbR^k \times \ball^k(B)$,
$\|\der_\theta g_*(x,\theta)\|<L_B$.
\end{enumerate}
Then as $n \to \infty$,
\[\sup_{x \in \bbR^k} \bigg\|\int g_n(x,\theta)\der \pi_n(\theta)
-\int g_*(x,\theta)\der \pi_*(\theta)\bigg\| \to 0.\]
\end{lemma}
\begin{proof}
Fix $\eps>0$. Then there is some $B>0$ for which $\pi_*(\bbR^k \setminus
\ball^k(B)) \leq \eps$ and $\pi_n(\bbR^k \setminus \ball^k(B)) \leq \eps$
for all $n$. Conditions (b) and (c) imply that the functions
$\theta \mapsto g_*(x,\theta)$ are uniformly bounded and equicontinuous over
$\theta \in \ball^k(B)$, for all $x \in \bbR^k$.
Then by Arzel\`a-Ascoli, there is a finite set $\mathcal{C}$ of continuous 
bounded functions on $\ball^k(B)$ such that for each $x \in \bbR^k$, there
exists $h_x \in \mathcal{C}$ for which
\[\sup_{\theta \in \ball^k(B)} \|g_*(x,\theta)-h_x(\theta)\| \leq \eps.\]
Let us write
\[\sup_{x \in \bbR^k} \bigg\|\int g_n(x,\theta) \der \pi_n(\theta)
-\int g_*(x,\theta) \der \pi_*(\theta) \bigg\|
\leq R_1+R_2+R_3+R_4\]
where
\begin{align*}
R_1&=\sup_{x \in \bbR^k} \int_{\bbR^k \setminus \ball^k(B)}
\|g_n(x,\theta)\|\der \pi_n(\theta)+
\int_{\bbR^k \setminus \ball^k(B)} \|g_*(x,\theta)\|\der \pi_*(\theta)\\
R_2&=\sup_{x \in \bbR^k} \int_{\ball^k(B)}
\|g_n(x,\theta)-g_*(x,\theta)\|\der \pi_n(\theta)
+\int_{\ball^k(B)} \|g_n(x,\theta)-g_*(x,\theta)\|\der \pi_*(\theta)\\
R_3&=\sup_{x \in \bbR^k} \int_{\ball^k(B)}
\|g_*(x,\theta)-h_x(\theta)\|\der \pi_n(\theta)
+\int_{\ball^k(B)} \|g_*(x,\theta)-h_x(\theta)\|\der \pi_*(\theta)\\
R_4&=\sup_{h \in \mathcal{C}} \bigg|\int_{\ball^k(B)}
h(\theta)\der \pi_n(\theta)-
\int_{\ball^k(B)} h(\theta)\der \pi_*(\theta)\bigg|.
\end{align*}
We have $R_1 \leq 2C_0\eps$ by condition (b), and $R_3 \leq 2\eps$.
The uniform convergence of condition (a) implies $R_2 \to 0$.
Since $\mathcal{C}$ is a finite set and each $h \in \mathcal{C}$ is continuous
and bounded, $R_4 \to 0$. As $\eps>0$ is arbitrary, this shows the lemma.
\end{proof}

\begin{corollary}\label{cor:derf_uniform_convergence}
In the setting of Lemma~\ref{lemma:mNPMLE}, denote
\[f_n(x)=f_{(M_n \cdot \pi_n) * \cN(0,\Sigma_n)}(x),
\qquad f_*(x)=f_{(M_* \cdot \pi_*) * \cN(0,\Sigma_*)}(x).\]
Then as $n \to \infty$, for each fixed $j \geq 1$
\[\sup_{x\in \mathbb{R}^k} |f_n(x)-f_*(x)| \to 0, \quad
\sup_{x\in \mathbb{R}^k} \|\der^j f_n(x)-\der^j f_*(x)\| \to 0.\]
\end{corollary}
\begin{proof}
Lemma \ref{lemma:mNPMLE} shows that $\pi_n \to \pi_*$ weakly. Consider the
functions $g_n(x,\theta)=\phi_{\Sigma_n}(x-M_n\theta)$ and
$g_*(x,\theta)=\phi_{\Sigma_*}(x-M_*\theta)$. From the form
\eqref{eq:Xmarginal}, it suffices to show that the conditions of Lemma
\ref{lemma:weakconvergence} hold for $g_n \to g_*$ and also for
$\der_x^j g_n \to \der_x^j g_*$. Observe that for any fixed integers
$a,b,c \geq 0$,
\[\|x-M\theta\|^a\|M\|^b\|\Sigma^{-1}\|^c
\cdot \phi_\Sigma(x-M\theta)\]
is uniformly bounded over all $(x,\theta) \in \bbR^k \times \bbR^k$ 
and all $M,\Sigma$ in a sufficiently small neighborhood of
$M_*,\Sigma_*$. Then differentiating $g(x,\theta)=\phi_\Sigma(x-M\theta)$ by the
chain rule, the derivative of any fixed orders in $(x,\theta)$ is uniformly
bounded, checking conditions (b) and (c) of Lemma \ref{lemma:weakconvergence}.
Furthermore, this shows that the derivative
of $g(x,\theta)$ of any fixed orders in $(x,\theta,M,\Sigma)$ is uniformly
bounded over all $x \in \bbR^k$, all $(M,\Sigma)$ in a neighborhood of
$(M_*,\Sigma_*)$, and all $\theta$ in any compact ball $\ball^k(B)$.
This implies $g_n \to g_*$ and $\der_x^j g_n \to \der_x^j g_*$ uniformly over
$(x,\theta) \in \bbR^k \times \ball^k(B)$, checking condition (a).
\end{proof}

\begin{proof}[Proof of Corollary \ref{cor:empBayes_relativeError}]
	Let $f_*=f_{(M_* \cdot \pi_*) * \mathcal{N}(0,\Sigma_*)}$ and
$f_n=f_{(M_n \cdot \pi_n) * \mathcal{N}(0,\Sigma_n)}$.
From Tweedie's formula (\eqref{eq:tweedie}),
		\begin{equation}\label{eq:first_order_tweedies_formula}
\theta(x  \mid M_*,\Sigma_*, \pi_*)=
M_*^{-1}\Big(\Sigma_* \frac{\der f_*^\top}{f_*}(x) + x\Big),\quad
		\theta(x \mid M_n, {\Sigma}_n, {\pi}_n)
		=M_n^{-1}\Big({\Sigma}_n \frac{\der {{f_n}}^\top}{{f_n}}(x)
+ x\Big).
\end{equation}
Let us write $g_*=\der f_*/f_*$ and $g_n=\der f_n/f_n$. Then
	\begin{align*}
	&\frac{1}{n}\|\theta(\bX \mid M_*,\Sigma_*,\pi_*)
-\theta(\bX \mid M_n,\Sigma_n,\pi_n)\|_F^2\\
&\leq \frac{3}{n} \sum_{i=1}^n \|M_*^{-1} \Sigma_*\|^2 \cdot \|g_*(x_i)
-g_n(x_i)\|^2 \\
        &\hspace{1in}+ \frac{3}{n} \sum_{i=1}^n \|
M_n^{-1}\Sigma_n - M_*^{-1}\Sigma_* \|^2 \cdot \|g_n(x_i)\|^2 
+ \frac{3}{n} \sum_{i=1}^n \|M_n^{-1} - M_*^{-1}\|^2 \cdot \|x_i\|^2 \\
 		&= {R_1}+ {R_2}+ {R_3}.
	\end{align*}
	Note that $M_n \to M_*$ and $n^{-1}\sum_i \|x_i\|^2<C$ for a constant
$C>0$ and all large $n$, by the assumption $\bX \toW X$. Thus $R_3 \to 0$.
By Proposition \ref{prop:uniformlipschitz}, $\theta(x\mid M_n,\Sigma_n,\pi_n)$ and
thus $g_n(x)$ are both Lipschitz in $x$, uniformly for all large $n$.
So also $n^{-1}\sum_i \|g_n(x_i)\|^2<C$ for a constant $C>0$
and all large $n$. Then since $\Sigma_n \to \Sigma_*$, we have $R_2 \to 0$.

To show $R_1 \to 0$, fix any $B > 0$ and apply $g_*-g_n=(f_n-f_*)\der
f_*/(f_*f_n)+(\der f_*-\der f_n)/f_n$ to write
	\begin{align*}
	R_1 &\leq
	\frac{C}{n} \sum_{i=1}^n \Big\|\frac{(f_n-f_*)\der f_*}{f_*f_n}(x_i)
+\frac{\der f_* - \der f_n}{f_n}(x_i)\Big\|^2 \mathbbm{1}\{x_i\in \ball^k(B)\}\\
	&\hspace{1in}+\frac{C}{n} \sum_{i=1}^n  \Big(\|g_*(x_i)\|^2
+\|g_n(x_i)\|^2\Big)
\mathbbm{1}\{x_i \not\in \ball^k(B)\} \equiv \mathrm{(I)}+\mathrm{(II)}.
\end{align*}
We have the bounds $f_*(x)>c_B$ and $\|\der f_*(x)\|<C_B$ for some
constants $C_B,c_B>0$ and all $x \in \ball^k(B)$. Combining with the uniform
convergence in Corollary \ref{cor:derf_uniform_convergence}, this yields
\[\mathrm{(I)} \leq C_B' \cdot \sup_{x \in \ball^k(B)}
|f_n(x)-f_*(x)|+\|\der f_n(x)-\der f_*(x)\| \to 0\]
as $n \to \infty$. For $\mathrm{(II)}$, applying again
that $g_*$ and $g_n$ are uniformly Lipschitz in $x$ for large $n$, for some
constant $C>0$ (independent of $B$) we have
\[\|g_*(x_i)\| \leq \|g_*(0)\|+C\|x_i\|, \quad
\|g_n(x_i)\| \leq \|g_n(0)\|+C\|x_i\|.\]
Then applying $g_n(0) \to g_*(0)$ as $n \to \infty$, this yields for some
constants $C,C'>0$ that
\begin{equation}\label{eq:IIbound}
\mathrm{(II)} \leq \frac{C}{n}\sum_{i=1}^n (C'+\|x_i\|^2)
 	\mathbbm{1}\{{x_i} \not\in \ball^k(B)\}.
\end{equation}
Define
\begin{align*}
	h(x) = 
	\begin{cases}
		C'+\|x\|^2 & \text{ if } x\not\in \ball^k(B) \\
		B^2(1-\frac{B-1}{\|x\|^2})(C'+\|x\|^2) & \text{ if }
                x \in \ball^k(B) \backslash \ball^k(B-1)\\
		0 & \text{ if } x\in \ball^k(B-1).
	\end{cases}
\end{align*}
Then $h\in \PL(2)$, and $\mathrm{(II)} \leq n^{-1}\sum_i h(x_i)
\to \bbE_{X \sim f_*}[h(X)]$. Combining with the bound for $\mathrm{(I)}$,
$\limsup_{n \to \infty} R_1 \leq \bbE_{X \sim f_*}[h(X)]$. Taking $B \to \infty$
and applying the dominated convergence theorem, we obtain $R_1 \to 0$, 
and this completes the proof of~\eqref{eq:eb_optimal_bayes_risk}. 

As $\theta(x\mid M_*,\Sigma_*,\pi_*)$ is Lipschitz in $x$
by Assumption~\ref{assump:model}(d) and $(\bTheta, \bX)\toW(\Theta,X)$,
\begin{align*}
	\frac{1}{n} \| \theta(\bX | M_*,\Sigma_*,\pi_*) - \bTheta\|_F^2 \to 
	\mathbb{E}[\| \Theta - \mathbb{E}[\Theta\mid X]\|^2] = \mmse(\pi \mid M_*,\Sigma_*),
\end{align*}
which immediately implies the second part of the corollary.
\end{proof}

\begin{proposition}\label{prop:surrogate_denoiser_derivative}
	In the setting of Lemma~\ref{lemma:mNPMLE}, as $n\to\infty$,
	\begin{align}\label{eq:surrogate_denoiser_derivative}
		\frac{1}{n}\sum_{i=1}^n \|\der \theta(x_i\mid M_*,\Sigma_*, \pi_*) - 
		  \der \theta(x_i \mid M_n,\Sigma_n, \pi_n)\| \to 0.
	\end{align}
\end{proposition}

\begin{proof}
The proof is similar to that of
Corollary \ref{cor:empBayes_relativeError} above.
Differentiating \eqref{eq:first_order_tweedies_formula} again in $x$,
	\begin{align}\label{eq:second_order_tweedies_formula}
		\der \theta(x\mid M_*,\Sigma_*, \pi_*) &=
		M_*^{-1}\bigg(\Sigma_*\Big(\frac{\der^2 f_*}{f_*} (x)
		-\frac{\der f_* \otimes \der f_*}{f_*^2} (x) \Big) + \Id\bigg),
	\end{align}
and similarly for $\der \theta(x \mid M_n,\Sigma_n,\pi_n)$.
Let us set $g_*=\der^2 f_*/f_*-\der f_* \otimes \der f_*/f_*^2$ and
$g_n=\der^2 f_n/f_n-\der f_n \otimes \der f_n/f_n^2$. Then
	\begin{align*}
		&\frac{1}{n}\sum_{i=1}^n \|\der \theta(x_i\mid M_*,\Sigma_*, \pi_*) - 
		  \der \theta(x_i \mid M_n,\Sigma_n, \pi_n)\| \\
&\leq \frac{3}{n} \sum_{i=1}^n \|M_*^{-1} \Sigma_*\| \cdot \|g_*(x_i)
-g_n(x_i)\|
        + \frac{3}{n} \sum_{i=1}^n \|
M_n^{-1}\Sigma_n - M_*^{-1}\Sigma_* \| \cdot \|g_n(x_i)\| 
+ 3\|M_n^{-1} - M_*^{-1}\|\\
 		&= {R_1}+ {R_2}+ {R_3}.
	\end{align*}
Here, by Proposition \ref{prop:uniformlipschitz}, $\|\der \theta(x \mid
M_n,\Sigma_n,\pi_n)\|$ and hence also $\|g_n(x)\|$ are uniformly bounded
in $x$ for all large $n$. Then since $M_n,\Sigma_n \to M_*,\Sigma_*$, we have
$R_2,R_3 \to 0$. For $R_1$, fix $B>0$ and write
	\begin{align*}
		R_1 &\leq \frac{C}{n}\sum_{i=1}^n \|g_*(x_i)-g_n(x_i)\|
\cdot \mathbbm{1}\{x_i \in \ball^k(B)\}
+\frac{C}{n}\sum_{i=1}^n \Big(\|g_*(x_i)\|+\|g_n(x_i)\|\Big)
\mathbbm{1}\{x_i \notin \ball^k(B)\}\\
&\equiv \mathrm{(I)}+\mathrm{(II)}.
\end{align*}
Applying the lower bounds $f_n(x),f_*(x) \geq c_B$ for all $x \in
\ball^k(B)$ and the uniform convergence of $f_n$, $\der f_n$, and $\der^2 f_n$
in Corollary \ref{cor:derf_uniform_convergence}, we have 
$\lim_{n \to \infty} \mathrm{(I)}=0$ for each fixed $B>0$.
Then applying $\bX \toW X$ and the uniform boundedness of $\|g_n(x)\|$ and
$\|g_*(x)\|$, we have
$\lim_{B \to \infty} \limsup_{n \to \infty} \mathrm{(II)}=0$. This shows $R_1
\to 0$, concluding the proof.
\end{proof}

Finally, we combine the above to prove Corollary \ref{cor:PCAmmse} of the
Main Text. Note that we have shown Lemma \ref{lemma:mNPMLE} and its
corollaries for any $\pi_n$ satisfying the condition (\ref{eq:approxNPMLE}).
Then in particular, these results hold for the NPMLE
$\pi_n=\MLE(\bX \mid M_n,\Sigma_n,\cP)$ defined by the maximization
(\ref{eq:NPMLE}), and we record this implication here.

\begin{corollary}\label{cor:trueNPMLE}
The conclusions of Lemma \ref{lemma:mNPMLE},
Corollary \ref{cor:empBayes_relativeError}, and
Proposition \ref{prop:surrogate_denoiser_derivative} hold for any NPMLE
$\pi_n=\MLE(\bX \mid M_n,\Sigma_n,\cP)$ solving the
maximization (\ref{eq:NPMLE}).
\end{corollary}

\begin{proof}
Applying the argument of Corollary \ref{cor:derf_uniform_convergence} with
$\pi_n=\pi_*$, we have
\[\sup_{x \in \mathbb{R}^k} \Big|f_{(M_n \cdot \pi_*)*\cN(0,\Sigma_n)}(x)
-f_{(M_* \cdot \pi_*)*\cN(0,\Sigma_*)}(x)\Big| \to 0.\]
Fixing any $B>0$, since $f_{(M_* \cdot \pi_*)*\cN(0,\Sigma_*)}(x)$ is bounded
away from 0 for $x \in \ball^k(B)$, this implies
\[\frac{1}{n}\sum_{i=1}^n
\left|\log \frac{f_{(M_n \cdot \pi_*)*\cN(0,\Sigma_n)}(x_i)}
{f_{(M_* \cdot \pi_*)*\cN(0,\Sigma_*)}(x_i)}\right| \cdot
\mathbbm{1}\{x_i \in \ball^k(B)\} \to 0.\]
To control this sum for $x_i \notin \ball^k(B)$, note that for some constant
$C>0$ independent of $n$, and for all $x \in \mathbb{R}^k$ and all large $n$,
\[C \geq \log f_{(M_n \cdot \pi_*)*\cN(0,\Sigma_n)}(x)
=\log \mathbb{E}\Big[\phi_{\Sigma_n}(x-M_n\Theta)\Big]
\geq \mathbb{E}\Big[\log \phi_{\Sigma_n}(x-M_n\Theta)\Big] \geq -C(1+\|x\|^2).\]
The same bounds hold for $\log f_{(M_* \cdot \pi_*)*\cN(0,\Sigma_*)}(x)$.
Hence
\[\Big|\log f_{(M_* \cdot \pi_*)*\cN(0,\Sigma_*)}(x)\Big|
+\Big|\log f_{(M_n \cdot \pi_*)*\cN(0,\Sigma_n)}(x)\Big|
\leq C(1+\|x\|^2),\]
so
\[\frac{1}{n}\sum_{i=1}^n
\left|\log \frac{f_{(M_n \cdot \pi_*)*\cN(0,\Sigma_n)}(x_i)}
{f_{(M_* \cdot \pi_*)*\cN(0,\Sigma_*)}(x_i)}\right| \cdot
\mathbbm{1}\{x_i \notin \ball^k(B)\}
\leq \frac{2C}{n}\sum_{i=1}^n (1+\|x_i\|^2)\mathbbm{1}\{x_i \notin \ball^k(B)\}.\]
As shown for (\ref{eq:IIbound}), this vanishes in the limit $n \to \infty$
followed by $B \to \infty$. Thus
\[\frac{1}{n}\sum_{i=1}^n
\log \frac{f_{(M_n \cdot \pi_*)*\cN(0,\Sigma_n)}(x_i)}
{f_{(M_* \cdot \pi_*)*\cN(0,\Sigma_*)}(x_i)} \to 0.\]
Now letting $\pi_n=\MLE(\bX \mid M_n,\Sigma_n,\cP)$, by definition
\[\frac{1}{n}\sum_{i=1}^n \log \frac{f_{(M_n \cdot \pi_n)*\cN(0,\Sigma_n)}(x_i)}
{f_{(M_n \cdot \pi_*)*\cN(0,\Sigma_n)}(x_i)} \geq 0.\]
Combining the above two displays shows that $\pi_n$ satisfies the condition
(\ref{eq:approxNPMLE}). Then Lemma \ref{lemma:mNPMLE},
Corollary \ref{cor:empBayes_relativeError}, and
Proposition \ref{prop:surrogate_denoiser_derivative} hold for $\pi_n$.
\end{proof}

\begin{proof}[Proof of Corollary \ref{cor:PCAmmse}]
By Proposition \ref{prop:PCAgaussian}, we have $(\bU,\bF) \toW (U,F)$ 
and $(\bV,\bG) \toW (V,G)$. The convergence of $\sqrt{\gamma}
\cdot \lambda_i$ in Lemma \ref{lemma:BBP} implies that
$\hat{S}$ as defined in (\ref{eq:shat}) is consistent for $S$.
Then we have also the consistency of the plug-in estimators
$\bar{M},M,\bar{\Sigma},\Sigma$ for
$\bar{M}_*,M_*,\bar{\Sigma}_*,\Sigma_*$. Corollary \ref{cor:PCAmmse}
now follows directly
from Corollary \ref{cor:empBayes_relativeError}, applied to the NPMLE $\pi_n$.
\end{proof}


\def\ampbv{\tilde{\mathbf{v}}}
\def\ampbu{\tilde{\mathbf{u}}}
\def\ampbf{\tilde{\mathbf{f}}}
\def\ampbg{\tilde{\mathbf{g}}}
\def\ampbV{\tilde{\mathbf{V}}}
\def\ampbU{\tilde{\mathbf{U}}}
\def\ampbF{\tilde{\mathbf{F}}}
\def\ampbG{\tilde{\mathbf{G}}}
\def\ampv{\tilde{v}}
\def\ampu{\tilde{u}}
\def\ampf{\tilde{f}}
\def\ampg{\tilde{g}}
\def\ampb{\tilde{{b}}}
\def\ampB{\tilde{{B}}}
\def\ampbbar{\tilde{\bar{b}}}
\def\ampsigmabar{\tilde{\bar{\Sigma}}}
\def\ampmubar{\tilde{\bar{M}}}

\section{State evolution of EB-PCA}\label{appendix:oracleSE}

We prove Theorem~\ref{thm:EBPCA} of the Main Text by comparing the EB-PCA trajectory with the
trajectory of the oracle Bayes AMP iterates. Let $\bG^t,\bV^t,\bF^t,\bU^t$
denote the iterates of EB-PCA. We use $\tilde{\phantom{0}}$ to denote the
analogous iterates of the oracle Bayes AMP algorithm. For example,
corresponding to the EB-PCA iterates
\begin{align}
\bV^t&=\theta(\bG^t \mid M_t, \Sigma_t, \pi_t)\nonumber\\
\bF^t&=\bY\bV^t - \bU^{t-1}\cdot \gamma B_t^{\top} \quad \text{ where }
B_t = \langle \der \theta(\bG^t \mid M_t, \Sigma_t, \pi_t) \rangle,\label{eq:Bt}
\end{align}
we have the oracle Bayes AMP iterates
\begin{align}
\tilde{\bV}^t&=\theta(\tilde{\bG}^t \mid M_{*,t}, \Sigma_{*,t}, \pi_*)\nonumber\\
\tilde{\bF}^t&= \bY\tilde{\bV}^t - \tilde{\bU}^{t-1}\cdot \gamma
\tilde{B}_t^\top \quad \text{ where } \tilde{B_t} =	\langle \der
\theta(\tilde{\bG}^t \mid M_{*,t}, \Sigma_{*,t}, \pi_*)
\rangle.\label{eq:tildeBt}
\end{align}
This oracle Bayes AMP algorithm is initialized with $\tilde{\bG}^0=\bG$,
$\tilde{\bU}^{-1}=\bF \cdot \Sigma_{*,0}^{1/2}$, and $M_{*,0}=M_*$
and $\Sigma_{*,0}=\Sigma_*$ as described in \eqref{eq:PCAM} and \eqref{eq:PCASigma}.
The matrices $\{\bar{M}_{*,t},\bar{\Sigma}_{*,t}\}$ and
$\{M_{*,t+1}, \Sigma_{*,t+1}\}$ are defined iteratively by the true state
evolution (\eqref{eq:SE}), where $u_t(\cdot)$ and $v_t(\cdot)$
are the true posterior mean denoisers defined by \eqref{eq:bayes_amp_denoisors}.

\subsection{State evolution of oracle Bayes AMP}

The validity of the state evolution for this oracle Bayes AMP procedure was
shown in \cite{montanari2017estimation}.

\begin{theorem}
\label{thm:bayesAMP}
Consider the rank-$k$ signal-plus-noise model of \eqref{eq:model}, and suppose that
Assumption~\ref{assump:model} holds.  Then for each fixed iterate $t$, as
$n,d\to\infty$, almost surely
\begin{align*}
	(\bU, \tilde{\bF}^t) \toW (U, F_t) \quad \textnormal{and} \quad
	(\bV, \tilde{\bG}^t) \toW (V, G_t),
\end{align*}
where $U \sim \bar{\pi}_*$, $F_t \mid U \sim \cN(\bar{M}_{*,t} \cdot U,\;
\bar{\Sigma}_{*,t})$, $V \sim \pi_*$, and
$G_t \mid V \sim \cN(M_{*,t} \cdot V,\;\Sigma_{*,t})$.
\end{theorem}
\begin{proof}
This follows from \cite[Theorem 7]{montanari2017estimation}, specializing to the
setting where all $k$ signal values $s_1,\ldots,s_k$ are distinct and
super-critical.
\end{proof}

We elaborate on the oracle Bayes initialization
\[\bU^{-1}=\bF \cdot \Sigma_*^{1/2}\]
discussed in Remark \ref{remark:initialization} of the Main Text,
corresponding to $\bF \cdot \Sigma_0^{1/2}$ in Algorithm \ref{alg:EBPCA}.
Informally, initializing AMP with the sample PCs
$\bG$ may be understood as first applying a large number of iterations
of a \emph{linear} AMP iteration whose fixed points are approximately
the sample PCs in the $n,d \to \infty$ limit, and then transitioning the
algorithm to apply the non-linear oracle Bayes denoisers.
This linear AMP algorithm applies the functions $u_{t,\text{lin}}(\bF)
=\bF \cdot D$ and $v_{t,\text{lin}}(\bG)=\bG \cdot E$ for two fixed diagonal
matrices
$D,E \in \bbR^{k \times k}$. Then $\langle \der u_{t,\text{lin}}(\bF)
\rangle^\top=D$ and $\langle \der v_{t,\text{lin}}(\bG) \rangle^\top=E$.
For $\bF,\bG$ to be approximate fixed
points of the linear AMP iterations, we require
\begin{align*}
\bF& \approx \bY v_{t,\text{lin}}(\bG)-u_{t-1,\text{lin}}(\bF) \cdot \gamma E
=\gamma \bF \Lambda E-\gamma \bF DE\\
\bG& \approx \bY^\top u_{t,\text{lin}}(\bF)-v_{t,\text{lin}}(\bG) \cdot D
=\bG \Lambda D-\bG ED.
\end{align*}
Thus $D,E$ should satisfy $\Id \approx \gamma(\Lambda-D)E$ and
$\Id \approx (\Lambda-E)D$. Here $\Lambda$ is random, but by
\eqref{eq:PCAsigma} converges to $\Sigma_*^{1/2} \cdot (S^2+\Id)$
for large $n,d$. Thus we choose $D,E$ to solve the pair of equations
\[\Id=\gamma\Big(\Sigma_*^{1/2} \cdot (S^2+\Id)-D\Big)E,
\qquad \Id=\Big(\Sigma_*^{1/2} \cdot (S^2+\Id)-E\Big)D,\]
yielding
\[D=\Sigma_*^{1/2}, \qquad E=\gamma^{-1/2}\bar{\Sigma}_*^{1/2}\]
where $\Sigma_*,\bar{\Sigma}_*$ are as described in \eqref{eq:PCASigma}.
Then corresponding to the PCA initialization $\bG^0=\bG$, we should set
\[\bU^{-1}=u_{t,\text{lin}}(\bF)=\bF \cdot \Sigma_*^{1/2}.\]

More formally, the analysis of \cite{montanari2017estimation} shows that the
state evolution for the iterates of the oracle Bayes AMP procedure initialized
at $\bG^0=\bG$ coincides with that of an AMP procedure using the matrix
\[\tilde{\bY}=\frac{1}{n}\bF \Lambda \bG^\top
+\bP_\bF^\perp\Big(\frac{1}{n}\bU S \bV^\top+\tilde{\bW}\Big)
\bP_\bG^\perp,\]
where $\bP_\bF^\perp$ and $\bP_\bG^\perp$ are the projections orthogonal to
the column spans of $\bF$ and $\bG$, and $\tilde{\bW}$ is a copy of $\bW$
independent of all other quantities. Following the
calculations in \cite{montanari2017estimation}, it may be verified that setting
$\bU^{-1}=\bF \cdot \Sigma_*^{1/2}$ rather than $\bU^{-1}=0$ is required for
the state evolution to be correct for describing
\[\bF^0=\tilde{\bY}v_0(\bG)-\bU^{-1} \cdot \gamma \langle v_0(\bG) \rangle.\]

\subsection{Comparison with oracle Bayes AMP}

\begin{proposition}\label{prop:states_non_degeneracy}
	Under Assumption~\ref{assump:model}, for all $t$, the matrices
$\Sigma_{*,t},\bar{\Sigma}_{*,t} \in \mathbb{R}^{k \times k}$ are invertible.
\end{proposition}

\begin{proof}
The initialization $\Sigma_{*,0}$ is a diagonal matrix with non-zero diagonal
entries, thus invertible. We will show, if $\Sigma_{*,t}$ is invertible, then
$\bar{\Sigma}_{*,t}$ is also invertible.

Suppose by contradiction that for some non-zero $\phi \in \mathbb{R}^k$,
\[0=\phi^\top \bar{\Sigma}_{*,t}\phi=\gamma \cdot
\mathbb{E}\Big[\mathbb{E}[\phi^\top V \mid M_{*,t} V +
\Sigma_{*,t}^{1/2}Z]^2\Big].\]
Note that by \eqref{eq:SigmaMrelation}, $M_{*,t}$ must be invertible if
$\Sigma_{*,t}$ is invertible. Then almost surely with respect to the
distribution of $V+M_{*,t}^{-1}\Sigma_{*,t}^{1/2}Z$, we must have
$\mathbb{E}[\phi^\top V \mid V+M_{*,t}^{-1}\Sigma_{*,t}^{1/2}Z]=0$.
Denote $\Omega=\operatorname{Cov}[M_{*,t}^{-1}\Sigma_{*,t}^{1/2}Z]^{-1}
=\Sigma_{*,t}^{-1/2}M_{*,t}^2\Sigma_{*,t}^{-1/2}$.
Since $V+M_{*,t}^{-1}\Sigma_{*,t}^{1/2}Z$ is supported on all of
$\mathbb{R}^k$, this implies that the (continuous) function
\[f(x)=\mathbb{E}\Big[\phi^\top V \;\Big|\;V + M_{*,t}^{-1}
\Sigma_{*,t}^{1/2}Z=x \Big]
=\frac{\mathbb{E}\Big[\phi^\top V \cdot
\exp\big(-\frac{1}{2}V^\top \Omega V+x^\top \Omega V\big)\Big]}
{\mathbb{E}\Big[\exp\big(-\frac{1}{2}V^\top \Omega
V+x^\top \Omega V\big)\Big]}\]
must be identically 0 for all $x \in \mathbb{R}^k$. Then the gradient of its
numerator in $x$,
\[g(x)=\mathbb{E}\Big[\Omega V \cdot \phi^\top V \cdot
\exp\big(-\tfrac{1}{2}V^\top \Omega V+x^\top \Omega V\big)\Big],\]
is also identically 0 for all $x \in \mathbb{R}^k$, so
\[0=\phi^\top \Omega^{-1}g(x)
=\mathbb{E}[(\phi^\top V)^2 \exp\big(-\tfrac{1}{2}V^\top \Omega V
+x^\top \Omega V\big)\Big].\]
Then we must have $\phi^\top V=0$ almost surely. However, by assumption
$\mathbb{E}[\phi^\top V V^\top \phi] = \phi^\top \Id_{k\times k} \phi = \|\phi\|^2 \neq 0$, a contradiction.
So $\bar{\Sigma}_{*,t}$ is invertible.

If $\bar{\Sigma}_{*,t}$ is invertible, an analogous argument shows
$\Sigma_{*,t+1}$ is invertible, concluding the proof.
\end{proof}	

\begin{proof}[Proof of Theorem~\ref{thm:EBPCA}]
Let $H_t$ be the hypothesis that the following hold almost surely as $n,d \to
\infty$:
\begin{enumerate}
	\item $(M_t, \Sigma_t) \to (M_{*,t}, \Sigma_{*,t})$,
	\item $\frac{1}{n} \|\bU^{t-1} - \ampbU^{t-1} \|_F^2 \to 0$,
	\item $\frac{1}{d} \| \bG^t - \ampbG^t\|_F^2 \to 0$.
\end{enumerate}
Similarly, let $\bar{H}_t$ be the hypothesis that the following hold almost
surely as $n,d \to \infty$:
\begin{enumerate}
	\item $(\bar{M}_t,\bar{\Sigma}_t) \to (\bar{M}_{*,t},\bar{\Sigma}_{*,t})$,
	\item $\frac{1}{d} \|\bV^t - \ampbV^t\|_F^2 \to 0$,
	\item $\frac{1}{n}\|\bF^t -\ampbF^t \|_F^2 \to 0$.
\end{enumerate}
Note that by Theorem~\ref{thm:bayesAMP}, $(\bV, \ampbG^t) \toW (V, G_t)$.
Then $H_t.3$ implies $(\bV,\bG^t) \toW (V,G_t)$. Similarly, $\bar{H}_t.3$
implies $(\bU,\bF^t) \toW (U,F_t)$, which establishes Theorem \ref{thm:EBPCA}.

To complete the proof, we show that $H_0$ holds, that $H_t$ implies $\bar{H}_t$,
and that $\bar{H}_t$ implies $H_{t+1}$.

\paragraph{Step 1: $H_0$ holds.} 

Lemma~\ref{lemma:BBP} implies that the estimates $\hat{S}$, $M_0$, and
$\Sigma_0$ in Algorithm~\ref{alg:EBPCA} are consistent for $S$, $M_{*,0}$, and
$\Sigma_{*,0}$ as $n,d \to \infty$. Then $H_0.1$ holds. $H_0.3$ holds trivially
because $\bG^0=\tilde{\bG}^0$, as the two algorithms have the same
initialization. $H_0.2$ also holds, because
\begin{align*}
\frac{1}{n} \|\bU^{-1} - \tilde{\bU}^{-1}\|_F^2 &\leq
\frac{1}{n}  \| \bF \|_F^2 \cdot \| \Sigma_0 - \Sigma_{*,0}\| \to 0
\end{align*}
by the consistency of $\Sigma_{0}$ for $\Sigma_{*,0}$ and the normalization
$n^{-1}\|\bF\|_F^2=k$.

\paragraph{Step 2: $H_t \Rightarrow \bar{H}_t.2$.} Let $g_j^t$ and
$\tilde{g}_j^t$ denote the rows of $\bG^t$ and $\tilde{\bG}^t$. Then
\begin{align*}
	\frac{1}{d} \|\bV^t - \ampbV^t\|_F^2
	 &= \frac{1}{d} \sum_{j=1}^d  \| \theta(\ampg^t_j\mid M_{*,t}, \Sigma_{*,t}, \pi_*) - \theta(g_{j}^t\mid M_t, \Sigma_t, \pi_t)\|^2 \\
	 &\leq 
	 \frac{2}{d} \sum_{j=1}^d  \| \theta(g^t_j\mid M_{*,t}, \Sigma_{*,t}, \pi_*) - \theta(g_{j}^t\mid M_t, \Sigma_t, \pi_t)\|^2 && \mathrm{(I)} \\
	 &\hspace{1in}+ 
	 \frac{2}{d} \sum_{j=1}^d  \| \theta(\ampg^t_j\mid M_{*,t}, \Sigma_{*,t}, \pi_*) - \theta(g_{j}^t\mid M_{*,t}, \Sigma_{*,t}, \pi_*)\|^2. && \mathrm{(II)}
\end{align*}
By Corollary~\ref{cor:empBayes_relativeError}, $\mathrm{(I)} \toas 0$. By
Assumption~\ref{assump:model}, $\theta(\cdot \mid M_{*,t}, \Sigma_{*,t}, \pi_*)$
is Lipschitz continuous with some Lipschitz constant $L$. Applying this and
$H_t.3$,
\begin{align*}
	\mathrm{(II)} \leq \frac{2L^2}{d} \sum_{j=1}^d \|\ampg_j^t - g_j^t\|^2 =
2L^2 \cdot \frac{1}{d} \|\ampbG^t - \bG^t\|_F^2 \toas 0.
\end{align*}

\paragraph{Step 3: $H_t \Rightarrow \bar{H}_t.3$.} Recalling \eqref{eq:Bt} and
\eqref{eq:tildeBt},
\begin{align*}
	\frac{1}{\sqrt{n}}\| \bF^t - \ampbF^t\|_F
	&\leq \|\bY\|\cdot \frac{1}{\sqrt{n}}
\|\bV^t - \ampbV^t\|_F + \|B_t - \ampB_t\| \cdot \frac{\gamma}{\sqrt{n}}
\|\ampbU^{t-1}\|_F + \|B_t\| \cdot \frac{\gamma}{\sqrt{n}}
\| \ampbU^{t-1} - \bU^{t-1}\|_F\\
	&=R_1+R_2+R_3.
\end{align*}
Since $\|\bY\|$ converges to a constant almost surely, and $n^{-1}\|\bV^t -
\ampbV^t\|_F^2\toas 0$, we have $R_1 \toas 0$. To show $R_2 \toas 0$, note that
$n^{-1}\|\ampbU^{t-1}\|_F^2$ converges to a constant almost surely by
Theorem~\ref{thm:bayesAMP}, so we only need to show $\|B_t - \ampB_t\|\toas 0$.
Let us write
\begin{align*}
	\|B_t - \ampB_t\|
	&=  \bigg\|\frac{1}{d} \sum_{j=1}^d \der \theta(\ampg^t_j\mid M_{*,t},
\Sigma_{*,t}, \pi_*) - \frac{1}{d} \sum_{j=1}^d \der \theta(g_{j}^t\mid M_t,
\Sigma_t, \pi_t)\bigg\|\\
	&\leq
	\frac{1}{d}\sum_{j=1}^d \|\der \theta(g^t_j\mid M_{*,t}, \Sigma_{*,t}, \pi_*) - \der \theta(g_{j}^t\mid M_t, \Sigma_t, \pi_t) \| && \mathrm{(I)} \\
	&\hspace{1in}+ 
	\frac{1}{d}\sum_{j=1}^d \|\der \theta(\ampg^t_j\mid M_{*,t}, \Sigma_{*,t}, \pi_*) - \der \theta(g_{j}^t\mid M_{*,t}, \Sigma_{*,t}, \pi_*) \| && \mathrm{(II)}.
\end{align*}
By Proposition~\ref{prop:surrogate_denoiser_derivative}, $\mathrm{(I)}\toas 0$.
To show $\mathrm{(II)} \toas 0$, note that the Lipschitz assumption implies
$\|\der \theta(\cdot \mid M_{*,t}, \Sigma_{*,t}, \pi_*)\| \leq L$.
Then for any fixed $B > 0$,
\begin{align*}
	\mathrm{(II)} &\leq \frac{1}{d}\sum_{j=1}^d \|  \der
\theta(\ampg^t_j\mid M_{*,t}, \Sigma_{*,t}, \pi_*) - \der \theta(g_{j}^t\mid
M_{*,t}, \Sigma_{*,t}, \pi_*) \| \cdot \mathbbm{1}\{g_j^t,\ampg_j^t \in \ball^k(B)\}  \\
	&\hspace{1in}+ \frac{2L}{d} \sum_{j=1}^d \mathbbm{1}\{g_j^t \not\in \ball^k(B)\}
	+ \frac{2L}{d} \sum_{j=1}^d \mathbbm{1}\{\ampg_j^t \not\in \ball^k(B)\}.
\end{align*}
The last two terms are at most
\begin{align*}
	\frac{2L}{d B^2} (\|\bG^t\|_F^2 + \|\ampbG^t\|_F^2).
\end{align*}
To control the first term, observe that 
$f(x)>c$ and $\|\der^i f(x)\|<C$ for all $x \in \ball^k(B)$, all $i=1,2,3$,
and some constants $C,c>0$ (depending on $B$). Then differentiating
\eqref{eq:second_order_tweedies_formula} again in $x$ yields
$\|\der^2 \theta(x \mid M_{*,t}, \Sigma_{*,t},
\pi_*)\|<C_B$ for a constant $C_B>0$ and all $x \in
\ball^k(B)$. Thus the first term is at most
\begin{align*}
	\frac{C_B}{d} \sum_{j=1}^d \| \ampg_j^t - g_j^t \| \leq
\frac{C_B}{\sqrt{d}} \| \ampbG^t - \bG^t \|_F.
\end{align*}
Applying $H_t.3$ and taking first the limit $n,d \to \infty$, followed by the
limit $B \to \infty$, we have $\mathrm{(II)}\toas 0$, and thus $\|B_t -
\ampB_t\|\toas 0$. This shows $R_2 \toas 0$. This further implies
$\limsup_n \|B_t\|<\infty$ because Theorem \ref{thm:bayesAMP} guarantees that
$\limsup_n \|\ampB_t\|<\infty$. Combining this with $H_t.2$ shows $R_3\toas 0$.

\paragraph{Step 4: $H_t \Rightarrow \bar{H}_t.1$.}
As $\theta(\cdot \mid M_{*,t}, \Sigma_{*,t}, \pi_t) \otimes \theta(\cdot \mid
M_{*,t}, \Sigma_{*,t}, \pi_t)\in\PL(2)$, Theorem~\ref{thm:bayesAMP} implies
\begin{align*}
	\bar{\Sigma}_{*,t}&=\lim_{n,d \to \infty} \frac{\gamma}{d} \sum_{j=1}^d
\tilde{v}_j^t\otimes \tilde{v}_j^t=(\tilde{\bV}^t)^\top \tilde{\bV}^t/n.
\end{align*}
Applying the definition of $\bar{\Sigma}_t$ in Algorithm~\ref{alg:EBPCA},
together with $\bar{H}_t.2$ already shown and Cauchy-Schwarz, 
this limit is the same as
\[\lim_{n,d \to \infty} \bar{\Sigma}_t=\lim_{n,d \to \infty}
(\bV^t)^\top\bV^t/n,\]
so $\bar{\Sigma}_t \to \bar{\Sigma}_{*,t}$. Recalling \eqref{eq:SigmaMrelation},
$\bar{M}_t=\bar{\Sigma}_t \cdot \hat{S}$, and applying the consistency of
$\hat{S}$, this implies also $\bar{M}_t \to \bar{M}_{*,t}$.

This completes the proof that $H_t$ implies $\bar{H}_t$. The proof that
$\bar{H}_t$ implies $H_{t+1}$ is the same as steps 2--4 above.
\end{proof}


\section{Analysis of the limiting risk}\label{sec:limiting_risk}

\def\Ebb{\mathbb{E}}
\def\Rbb{\mathbb{R}}

\subsection{Reparametrization of the states}\label{sec:reparametrize_states}

Recall from Section \ref{sec:bayesoptimal} of the Main Text the definitions
\[\bar{Q}_{*,t}= \frac{1}{\gamma} S^{-1/2}\bar{M}_{*,t}^\top \bar{\Sigma}_{*,t}^{-1}
\bar{M}_{*,t}S^{-1/2},
\qquad Q_{*,t}=S^{-1/2}M_{*,t}^\top \Sigma_{*,t}^{-1} M_{*,t}S^{-1/2}\]
and the map
\[F_\pi(Q)=\bbE\Big[\bbE[\Theta \mid \Theta+Q^{-1/2}Z]^{\otimes 2}\Big],
\qquad \Theta \sim \pi,\;Z \sim \cN(0,\Id) \text{ independent}.\]
For each $\bar{Q}_{*,t}$ and $Q_{*,t+1}$ with $t \geq 0$,
recalling the identities $\bar{M}_{*,t}=\bar{\Sigma}_{*,t} \cdot S$
and $M_{*,t+1}=\Sigma_{*,t+1} \cdot S$ from \eqref{eq:SigmaMrelation}, we have
\[\bar{Q}_{*,t}=\frac{1}{\gamma} S^{1/2}  \bar{\Sigma}_{*,t}  S^{1/2},
\qquad Q_{*,t+1}=S^{1/2}  \Sigma_{*,t+1} S^{1/2}.\]
Then substituting from \eqref{eq:SE} the definition
\begin{align*}
\bar{\Sigma}_{*,t}&=\gamma \cdot \bbE\Big[\Ebb[V \mid M_{*,t} V +
\Sigma_{*,t}^{1/2}Z]^{\otimes 2}\Big]\\
&=\gamma \cdot \bbE\Big[\Ebb[V \mid  S^{1/2}V
+S^{1/2}M_{*,t}^{-1}\Sigma_{*,t}^{1/2}Z]^{\otimes 2}\Big]
=\gamma \cdot \bbE\Big[\Ebb[V \mid S^{1/2}V+Q_{*,t}^{-1/2}Z]^{\otimes 2}\Big],
\end{align*}
we obtain
\begin{align*}
	\bar{Q}_{*,t}&=\frac{1}{\gamma} S^{1/2} \bar{\Sigma}_{*,t} S^{1/2} 
	=\Ebb\Big[\Ebb[S^{1/2} V \mid
S^{1/2}V+Q_{*,t}^{-1/2}Z]^{\otimes_2}\Big]=F_{S^{1/2}\pi_*} (Q_{*,t}).
\end{align*}
Similarly,
\begin{align*}
	{Q}_{*,t+1}&= S^{1/2} {\Sigma}_{*,t+1} S^{1/2} 
	=\Ebb\Big[\Ebb[S^{1/2} U \mid S^{1/2}U+(\gamma
\bar{Q}_{*,t})^{-1/2}Z]^{\otimes_2}\Big]=F_{S^{1/2}\bar{\pi}_*} (\gamma\cdot \bar{Q}_{*,t}).
\end{align*}
This verifies the equivalent forms of \eqref{eq:SEQ} for the state evolution
as stated in the Main Text.

\subsection{Progression of SNR}

We now prove Proposition \ref{prop:SE} of the Main Text. In the standardized compound decision
model
\[\Theta \sim \pi, \qquad X \mid \Theta \sim \cN(0,Q^{-1})\]
parametrized by $Q$ and $\pi$, let us define
\[M_\pi(Q)=\bbE\big[\operatorname{Cov}[\Theta \mid X]\big]
=\bbE\big[(\Theta-\bbE[\Theta \mid X])^{\otimes 2}\big]\]
so that
\begin{equation}\label{eq:FMrelation}
F_\pi(Q)=\bbE[\Theta\Theta^\top]-M_\pi(Q).
\end{equation}
We will use the following properties of $F_\pi(Q)$ and $M_\pi(Q)$ established in
\cite{miolane2017fundamental} and \cite{reeves2018mutual}.

\begin{lemma}[Lemma 9 in \cite{miolane2017fundamental}]\label{lemma:Fnondecreasing}
If $Q_1 \preceq Q_2$, then $F_{\pi}(Q_1) \preceq F_{\pi}(Q_2)$.
\end{lemma}
\begin{lemma}[Theorem 2 in \cite{reeves2018mutual}]\label{lemma:Mbound}
$M_{\pi}(Q) \preceq (\operatorname{Cov}[\Theta]^{-1} + Q )^{-1}$.
\end{lemma}

\begin{proof}[Proof of Proposition \ref{prop:SE}]
For part (a), observe that
by~\eqref{eq:PCAsigma}, $Q_{*,0}$ is the diagonal
matrix
\begin{align*}
Q_{*,0}=\diag\bigg(\frac{1}{s_i} \cdot \frac{\gamma s_i^4 - 1}{\gamma s_i^2
+1}\bigg).
\end{align*}
Furthermore, applying \eqref{eq:FMrelation}, Lemma \ref{lemma:Mbound}, and
$\bbE[VV^\top]=\bbE[UU^\top]=\Id$,
\begin{align*}
\bar{Q}_{*,0}= F_{S^{1/2} \cdot \pi_*}(Q_0) 
&=\Ebb[S^{1/2}VV^\top S^{1/2}]-M_{S^{1/2} \cdot \pi_*}(Q_{*,0}) \\
&\succeq \Ebb[S^{1/2}VV^\top S^{1/2}]-\Big(\Ebb[S^{1/2}VV^\top
S^{1/2}]^{-1}+Q_{*,0}\Big)^{-1}\\
&=S-(S^{-1}+Q_0)^{-1}=\diag\bigg(\frac{\gamma s_i^4 - 1}{\gamma s_i
(1+s_i^2)}\bigg),\\
Q_1=F_{S^{1/2} \cdot \bar{\pi}_*}(\gamma \bar{Q}_{*,0}) 
&= \Ebb[S^{1/2}UU^\top S^{1/2}]-M_{S^{1/2} \cdot \bar{\pi}_*}(\gamma
\bar{Q}_{*,0})  \\
&\succeq \Ebb[S^{1/2}UU^\top S^{1/2}] - \Big(\Ebb[S^{1/2}UU^\top
S^{1/2}]^{-1} + \gamma \bar{Q}_{*,0}\Big)^{-1}\\
&=S-\bigg(S^{-1} + \diag\bigg(\frac{\gamma s_i^4 - 1}{s_i (1+s_i^2)}\bigg)
\bigg)^{-1}=Q_{*,0}.
\end{align*}
Thus $Q_{*,1} \succeq Q_{*,0}$. Then $\bar{Q}_{*,1} \succeq \bar{Q}_{*,0}$,
$Q_{*,2} \succeq Q_{*,1}$, etc.\ by \eqref{eq:SEQ} and the monotonicity of
$F_\pi$ established in Lemma \ref{lemma:Fnondecreasing}.
Since $\bar{Q}_{*,t} \preceq
\bbE[S^{1/2}VV^\top S^{1/2}]=S$ and similarly $Q_{*,t} \preceq S$,
this implies that $\{\bar{Q}_{*,t}\}_{t \geq 0}$ and $\{Q_{*,t}\}_{t \geq 0}$
must converge to some limits $\bar{Q}$ and $Q$, which must satisfy
\eqref{eq:fixedpoint}. This shows part (a).

For part (b), in the compound decision model $X=M\Theta+\Sigma^{1/2}Z$ where
$\Theta \sim \pi$ and $Z \sim \cN(0,\Id)$,
set $Q=\Sigma^{-1/2}MS^{-1/2}$ and note that
\begin{align*}
\mmse(\pi \mid M, \Sigma) 
	&=\bbE\Big[\|\Theta - \Ebb[\Theta \mid M \Theta + \Sigma^{1/2} Z ]\|^2\Big] \\
	&=\Tr \bbE\Big[\operatorname{Cov}[\Theta \mid M \Theta + \Sigma^{1/2}
Z]\Big]\\
&=\Tr \bbE\Big[\operatorname{Cov}[\Theta \mid S^{1/2}\Theta + S^{1/2}M^{-1}
\Sigma^{1/2} Z]\Big]\\
&=\Tr S^{-1/2}
\bbE\Big[\operatorname{Cov}[S^{1/2}\Theta \mid S^{1/2}\Theta + Q^{-1/2}Z]\Big]
S^{-1/2}\\
&= \Tr S^{-1/2} M_{S^{1/2} \cdot \pi}(Q)S^{-1/2}.
\end{align*}
The relation of \eqref{eq:FMrelation} and part (a) imply
$M_{S^{1/2} \cdot \pi_*}(Q_{*,t+1}) \preceq M_{S^{1/2} \cdot \pi_*}(Q_{*,t})$,
so
\begin{align*}
\mmse(\pi_* \mid M_{*,t+1}, \Sigma_{*,t+1})&=
\Tr S^{-1/2} M_{S^{1/2} \cdot \pi_*}(Q_{*,t+1})S^{-1/2}\\
&\leq \Tr S^{-1/2} M_{S^{1/2} \cdot \pi_*}(Q_{*,t})S^{-1/2}
=\mmse(\pi_* \mid M_{*,t}, \Sigma_{*,t}).
\end{align*}
Similarly $\mmse(\bar{\pi}_* \mid \bar{M}_{*,t+1}, \bar{\Sigma}_{*,t+1})
\leq \mmse(\bar{\pi}_* \mid \bar{M}_{*,t}, \bar{\Sigma}_{*,t})$.
Since $\bV^0=\hat{\bV}$ is exactly the initial empirical Bayes estimate for
$\bV$, we have $\mmse(\pi_* \mid M_{*,0},\Sigma_{*,0})
=\mmse(\pi_*~\mid~M_*~,~\Sigma_*)$.
For $\bU$, we observe that the signal-to-noise matrix analogous to $Q_{*,0}$
defined by $\bar{M}_*$ and $\bar{\Sigma}_*$ for the left sample PCs is,
from~\eqref{eq:PCAsigma},
\[\frac{1}{\gamma}S^{-1/2}\bar{M}_*^\top \bar{\Sigma}_*^{-1}
\bar{M}_*S^{-1/2}
=\diag\bigg(\frac{1}{\gamma s_i} \cdot \frac{\gamma s_i^4-1}{s_i^2+1}\bigg).\]
This is exactly the lower bound established above for $\bar{Q}_{*,0}$, so
we also have $\mmse(\bar{\pi}_* \mid \bar{M}_{*,0}, \bar{\Sigma}_{*,0})
\leq \mmse(\pi_* \mid \bar{M}_*,\bar{\Sigma}_*)$. This shows part (b).
\end{proof}

\subsection{Bayes optimality}
\begin{proof}[Proof of Proposition~\ref{prop:bayes_optimality}]
Expanding the square, we have
\[\|\bU^t \hat{S} (\bV^t)^\top - \bU S \bV^{\top}\|_F^2
=\Tr (\bU^t)^\top\bU^t \hat{S} (\bV^t)^\top \bV^t \hat{S}
-2\Tr \bU^\top\bU^t \hat{S} (\bV^t)^\top \bV S
+\Tr \bU^\top\bU S \bV^\top \bV S.\]
Theorem \ref{thm:EBPCA} implies
\[n^{-1}(\bU^t)^\top \bU^t \to \bbE[U_tU_t^\top],
\quad n^{-1}(\bU^t)^\top \bU \to \bbE[U_tU^\top]
=\bbE[U_tU_t^\top], \quad n^{-1} \bU^\top \bU \to \bbE[UU^\top]\]
\[d^{-1}(\bV^t)^\top \bV^t \to \bbE[V_tV_t^\top],
\quad d^{-1}(\bV^t)^\top \bV \to \bbE[V_tV^\top]
=\bbE[V_tV_t^\top], \quad d^{-1} \bV^\top \bV \to \bbE[VV^\top]\]
where $\bbE[U_tU_t^\top]=S^{-1/2}Q_{t+1}S^{-1/2}$
and $\bbE[V_tV_t^\top]=S^{-1/2}\bar{Q}_tS^{-1/2}$. Together with the consistency
$\hat{S} \to S$,
\[\frac{1}{nd}\|\bU^t \hat{S} (\bV^t)^\top - \bU S \bV^{\top}\|_F^2
\to \Tr \bbE[UU^\top SVV^\top S]-\Tr \bar{Q}_tQ_{t+1}.\]
By Proposition \ref{prop:SE}, we have
$\Tr \bar{Q}_tQ_{t+1}=\Tr \bar{Q}Q+o_t(1)$ for the unique fixed point
$(\bar{Q},Q)$ of \eqref{eq:fixedpoint}, establishing the result.
\end{proof}

\section{Details of simulations and data analyses}\label{sec:data}

\subsection{Details of EB-PCA.}

A software implementation of EB-PCA is publicly available
at \url{https://github.com/TraceyZhong/EBPCA}.

In our implementation, we take 
$\cP$ as the class of all probability distributions
on $\bbR^k$,
and approximate this class $\cP$ using a discrete support. We
apply the ``exemplar method'' of \cite{lashkari2008convex} and take
the support points to be
\[\{z_1,\ldots,z_n\}=\{M^{-1}x_1,\ldots,M^{-1}x_n\}.\]
This is motivated by the observation that these values should provide a fine
grid that covers the high density regions under any true prior $\pi_* \in \cP$.
This grid automatically adapts to these high density regions, and the number
of support points is independent of the dimension $k$. Thus we express
\begin{equation}\label{eq:discretepi}
\pi=\sum_{i=1}^n w_i \delta_{z_i}=\sum_{i=1}^n w_i \delta_{M^{-1}x_i}
\end{equation}
and maximize \eqref{eq:NPMLE} over the probability weights $w_1,\ldots,w_n$.

The resulting maximization problem is concave
over the weights $w_1,\ldots,w_n$. We solved this maximization using the 
generic interior point solver implemented in MOSEK. When analyzing real data of
high dimension, we set the maximum number of prior support points to be 2000
to reduce the computational cost. More specifically, when either dimension
($n$ or $d$) exceeded 2000, we drew a random subsample of size 2000 to be the 
prior support points.

\subsection{Details of simulations and applications}\label{subsec:practical}

\paragraph{Simulation details.}

To provide a direct comparison between mean-field VB and EB-PCA,
we used the same NPMLE procedure in the two methods. EB-PCA, oracle
Bayes AMP, and EBMF were all run for 10 iterations, which was sufficient for
convergence.

For \texttt{spca}, we used the implementation in the R package \texttt{elasticnet}.
We tested \texttt{spca} with
a range of sparsity tuning parameter between 0.025 and 0.175 and report the
best results.
\paragraph{Genotype data pre-processing.}

For the 1000 Genomes Project, we used the Phase III genotypes publicly available at
\url{https://www.internationalgenome.org}. For the International HapMap Project, we used
the third phase genotype data available at 
\url{https://www.sanger.ac.uk/resources/downloads/human/hapmap3.html}.

For both 1000 Genomes and HapMap3, we used
\texttt{Plink(v1.90b6.9)} to retain only common variants with minor allele
frequency $>0.1$, and generated a set of such variants in approximate linkage
disequilibrium \texttt{(--indep 50 5 1.5)}. This yielded 466{,}431 SNPs for
1000 Genomes and 142{,}185 SNPs for HapMap3. For 1000 Genomes, we chose a random subset of
100{,}000 SNPs and used this to compute the ground truth. For HapMap3, 
we used all 142{,}185 selected SNPs to compute the ground truth.

\paragraph{Gene expression data pre-processing.}

We used raw PBMC single-cell data from 10X Genomics, publicly available at
\url{https://cf.10xgenomics.com/samples/cell/pbmc3k/pbmc3k_filtered_gene_bc_matrices.tar.gz.}
We performed quality control following
\url{https://satijalab.org/seurat/v3.2/pbmc3k_tutorial.html} 
and applied an iterative procedure to clean the gene expression count matrix. We
first removed genes with no variation across cells. We then centered and
scaled the counts for each gene, and computed sample PCs. We used the PCs
to identify outlier cells, and repeated this procedure with outliers removed.
In total, we removed 12 cells and 3 genes in this cleaning step.

\paragraph{General pre-processing.}

In all data examples, we centered and scaled the samples (SNPs for genotype
data, and genes for expression data) before applying either PCA or EB-PCA.
This reduced the influence of single SNPs or single genes on the PCs, and
improved their interpretability.

\paragraph{Accuracy of the noise model.}

Figure~\ref{fig:singvals} displays a scree plot of all singular values, and a
histogram of all bulk singular values, for each of the three data examples to
which EB-PCA was applied. The distributions of singular values predicted by the
square-root of the Marcenko-Pastur law, corresponding to the modeling
assumption $w_{ij} \sim \cN(0,1/n)$, are overlaid on the histograms. We observe
a near-perfect fit for the 1000 Genomes example, suggesting that a model of
independent noise entries with homoscedastic variance may be quite accurate for
the subsampled genotype data. The fit to the HapMap3 and the single-cell gene expression data
are not as close, indicating that the noise model $w_{ij} \sim \cN(0,1/n)$ is a
rougher approximation for these data, and suggesting the possibility for further
improvement using a method developed around more general models of
correlated noise.

\paragraph{Re-estimation of priors.}\label{par:iterNPMLE}

Figure~\ref{fig:iterNPMLE} compares estimation accuracy with and without
re-estimating the priors $\bar{\pi}_*$ and $\pi_*$ after the first iteration,
as discussed in Remark~\ref{remark:iterNPMLE}, on the 1000 Genomes example.
For both subsample sizes of $100$ and $1000$ SNPs, the difference in estimation
error between these approaches is minimal.

The runtime for EB-PCA in our implementation is dominated by the NPMLE
computation. (For subsamples of 1000 SNPs, total runtime for 5 EB-PCA
iterations was 53 seconds using the NPMLE update in each iteration, compared to
11 seconds using only the NPMLE in the first iteration.) These observations
suggest that re-estimating the priors may be avoided without compromising
accuracy, if computational cost is a concern.

\paragraph{Contribution of iterative refinement.}\label{par:AMPvsInit}

Figure~\ref{fig:iterNPMLE} depicts also the estimation errors across EB-PCA
iterations, on the same 1000 Genomes example.
The decrease in error from Iteration 0 (sample PCs)
to Iteration 1 reflects the initial denoising step, and subsequent decreases in
error indicate gains from iterative refinement using AMP. As discussed in the
Main Text, for this data example, gains in accuracy for EB-PCA resulted
mostly from the initial denoising, and EB-PCA typically converged within 1--2
iterations.
This is in contrast to the simulated examples of
Figures~\ref{fig:iteratesuniform}(c) and \ref{fig:iteratestwopoint}(c) for
weaker signal strengths, where both initial denoising and iterative refinement
contribute to the improved accuracy of EB-PCA.

\begin{figure}
\minipage{0.5\columnwidth}
\xincludegraphics[width=0.95\textwidth,label=(a)]{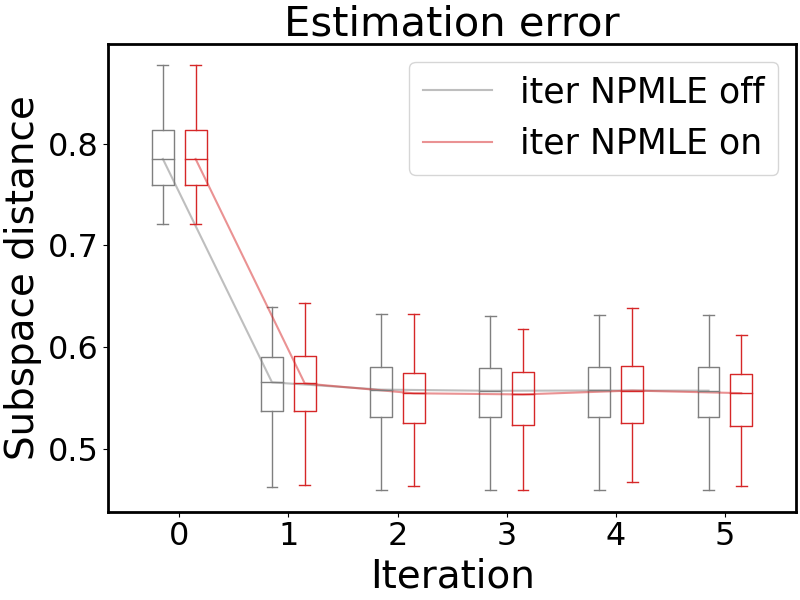}
\endminipage
\minipage{0.5\columnwidth}
\xincludegraphics[width=0.95\textwidth,label=(b)]{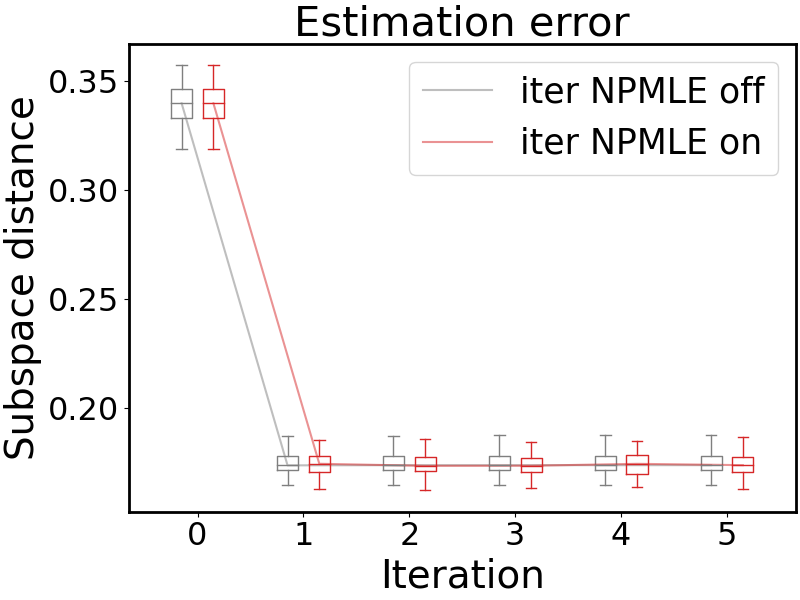}
\endminipage
\caption{Estimation errors for EB-PCA with and without iterative updates of priors using NPMLE. 
(a)-(b) Comparison on 50 random subsamples of $100$ 
SNPs or $1000$ SNPs from the 1000 Genomes project.}\label{fig:iterNPMLE}
\end{figure} 


\begin{figure}[t]
\xincludegraphics[width=0.33\textwidth,label=(a)]{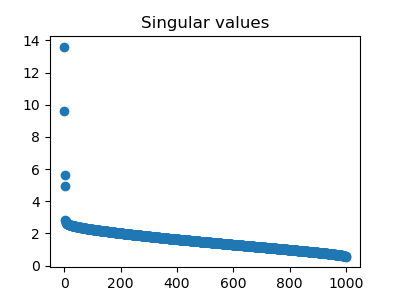}%
\xincludegraphics[width=0.33\textwidth,label=(b)]{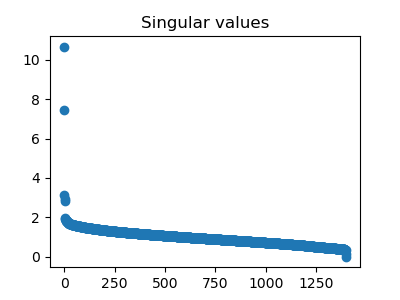}%
\xincludegraphics[width=0.33\textwidth,label=(c)]{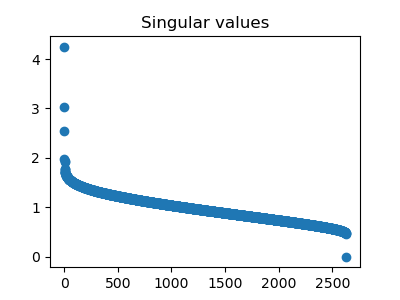}\\
\xincludegraphics[width=0.33\textwidth]{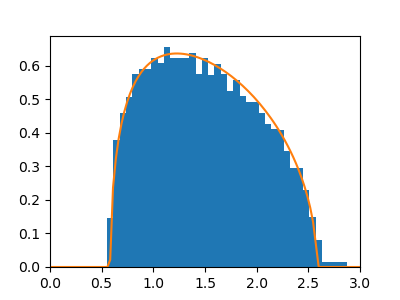}%
\xincludegraphics[width=0.33\textwidth]{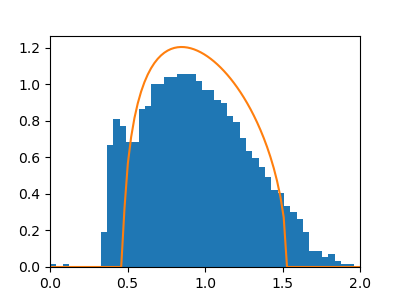}%
\xincludegraphics[width=0.33\textwidth]{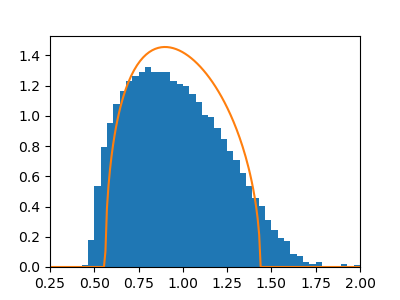}
\caption{Scree plots and histograms of the singular values for (a) the matrix
of 1000 subsampled SNPs from the 1000 Genomes Project, corresponding to Figure
\ref{fig:1000genomesintro}, (b) the matrix of 5000 subsampled SNPs from the HapMap3, corresponding to Figure \ref{fig:hapmap3}, and (c) the matrix of all
13{,}711 gene expressions from the PBMC single-cell RNA-seq data,
corresponding to Figure \ref{fig:PBMC}. A theoretical bulk distribution for the
singular values, predicted by the Marcenko-Pastur law under the noise model
$w_{ij} \overset{iid}{\sim} \cN(0,1/n)$, is overlaid on the histograms.}
\label{fig:singvals}
\end{figure}

\begin{figure}
\minipage{0.33\columnwidth}
\xincludegraphics[width=\textwidth,label=(a)]{Fig/univariate/Figure2/Two_points/EB-PCA_u_iter00.png}\\\\
\xincludegraphics[width=\textwidth]{Fig/univariate/Figure2/Two_points/EB-PCA_u_iter01.png}\\\\
\xincludegraphics[width=\textwidth]{Fig/univariate/Figure2/Two_points/EB-PCA_u_iter04.png}
\endminipage%
\minipage{0.33\columnwidth}
\xincludegraphics[width=\textwidth,label=(b)]{Fig/univariate/Figure2/Two_points/EBMF_u_iter00.png}\\\\
\xincludegraphics[width=\textwidth]{Fig/univariate/Figure2/Two_points/EBMF_u_iter01.png}\\\\
\xincludegraphics[width=\textwidth]{Fig/univariate/Figure2/Two_points/EBMF_u_iter04.png}
\endminipage%
\minipage{0.33\columnwidth}
\xincludegraphics[width=\textwidth,label=(c)]{Fig/univariate/Figure3/Two_points_u.png}
\endminipage
\caption{Same comparisons as in Figure \ref{fig:iteratesuniform},
with a two point $\operatorname{Bernoulli}\{+1,-1\}$ prior
and overlaid convolution densities
$\operatorname{Bernoulli}\{+1,-1\} * \cN(0,\bar{\sigma}_t^2/\bar{\mu}_t^2)$.
}
\label{fig:iteratestwopoint}
\end{figure}

\begin{table}
  \caption{PC estimation errors for two simulated bivariate priors\label{tab:bivariate}}
  \centering
  \begin{tabular}[c]{c|ccc|ccc}
      \toprule
       & \multicolumn{3}{c|}{Concentric circle} & \multicolumn{3}{c}{Three-point
mixture}\\ \hline
      Error & PC1 & PC2 & Joint & PC1 & PC2 & Joint \\ \hline\hline
      PCA  & .25(.0060) & .52(.014) & .41(.0080) & .25(.0064) & .52(.016) & .40(.011)  \\ 
      Marginal EB-PCA & .23(.0068) & .47(.014) & .37(.0088) & .081(.020) & .30(.027) & .22(.019)  \\ 
      Joint EB-PCA    & .22(.0065) & .37(.019) & .30(.011)  & .046(.018) & .080(.033) & .067(.025)  \\
      \bottomrule
  \end{tabular}
\end{table}

\begin{table}
  \caption{PC estimation errors on subsampled genotype matrices from the HapMap3\label{tab:hapmap3}}
  \centering
  \begin{tabular}[c]{c|c|ccccc}
      \toprule
      & Error & PC1 & PC2 & PC3 & PC4 & Joint \\ \hline\hline
      \multirow{2}{*}{1000 SNPs} & PCA &.11(.0069) & .16(.0059) & .45(.060) & .54(.094) & .51(.098) \\ 
      & EB-PCA &.064(.012) & .084(.0091) & .33(.088) & .37(.14) & .36(.14) \\ \hline
      \multirow{2}{*}{5000 SNPs} & PCA  & .049(.0029) & .071(.0023) & .22(.043) & .27(.051) & .25(.043) \\ 
      & EB-PCA & .032(.0044) & .043(.0038) & .17(.052) & .20(.068) & .17(.053) \\\hline
      \multirow{2}{*}{10000 SNPs} & PCA &.034(.0023) & .049(.0020) & .15(.021) & .18(.029) & .17(.027) \\ 
      & EB-PCA &.024(.0027) & .032(.0026) & .12(.027) & .14(.035) & .13(.033) \\ 
      \bottomrule
  \end{tabular}
\end{table}

\clearpage

\small{
\bibliographystyle{alpha}
\bibliography{EBPCA}}

\end{document}